\documentclass[journal]{IEEEtran}

\usepackage{amsfonts}
\usepackage{cite}
\usepackage{graphicx}
\usepackage{amssymb}
\usepackage{subfigure}
\usepackage{array}
\usepackage{color}
\usepackage{amsmath}
\usepackage{amsthm}
\usepackage{algorithmic}
\usepackage{algorithm}
\usepackage{cases}
\usepackage{bm}
\usepackage{mathrsfs}
\usepackage{epsfig}
\usepackage{psfrag}
\usepackage{url}
\usepackage{setspace}

\usepackage{etoolbox}
\makeatletter
\patchcmd{\@makecaption}
  {\scshape}
  {}
  {}
  {}
\makeatletter
\patchcmd{\@makecaption}
  {\\}
  {.\ }
  {}
  {}
\makeatother

\newtheorem{Thm}{Theorem}

\usepackage[dvipsnames, svgnames, x11names]{xcolor}
\newcommand{\lscr}{\textcolor{black}}
\newcommand{\lscc}{\textcolor{black}}
\newcommand{\lsccc}{\textcolor{black}}
\newcommand{\lscca}{\textcolor{black}}
\newcommand{\lscv}{\textcolor{black}}
\newcommand{\lscn}{\textcolor{black}}
\newcommand{\lscb}{\textcolor{black}}

\newcommand{\wqm}{\textcolor{black}}

\newcommand{\nwqa}{\textcolor{black}}

\newcommand{\nwqc}{\textcolor{black}}
\newcommand{\nwqd}{\textcolor{black}}

\newcommand{\nwqf}{\textcolor{black}}
\newcommand{\nwqh}{\textcolor{black}}
\newcommand{\wq}{\textcolor{black}}
\IEEEoverridecommandlockouts

\begin{document}

\title{Jointly Sparse Signal Recovery and Support Recovery via Deep Learning {with} Applications in MIMO-based Grant-Free Random Access}

%\author{\IEEEauthorblockN{A}
%\IEEEauthorblockA{A}
%}
\author{
\IEEEauthorblockN{Ying Cui, Shuaichao Li, Wanqing Zhang}
%\IEEEauthorblockA{Dept. of EE, Shanghai Jiao Tong University, China}
\thanks{The authors are with the Department of Electronic Engineering, Shanghai Jiao Tong University, Shanghai 200240, China. Ying Cui is the corresponding author (cuiying@sjtu.edu.cn).
}
\thanks{This work was supported in part by the National Key R\&D Program of China under Grant 2018YFB1801102 \lscca{and in part by Natural Science Foundation of Shanghai under Grant 20ZR1425300}. This article was presented in part at IEEE WCNC 2020 \cite{LS2019} and IEEE SPAWC 2020 \cite{zhang2020}.
}}

\maketitle

\vspace{-8mm}
\begin{abstract}
In this paper, we investigate jointly sparse signal recovery and jointly sparse support recovery in Multiple Measurement Vector (MMV) models for complex signals, which arise in many applications in communications and signal processing. Recent key applications include channel estimation and device activity detection in MIMO-based grant-free random access which is proposed to support massive machine-type communications (mMTC) for \lscr{ Internet of Things (IoT)}. Utilizing techniques in compressive sensing, optimization and deep learning, we propose \wq{two model-driven} approaches, based on the standard auto-encoder structure for real numbers. One is to jointly design the common measurement matrix and jointly sparse signal recovery method, and the other aims to jointly design the common measurement matrix and jointly sparse support recovery method. The proposed model-driven approaches can effectively \lscr{utilize features of} sparsity patterns in \lscr{designing common} measurement matrices and \lscr{adjusting \wq{model-driven}} decoders, and can greatly benefit from the underlying state-of-the-art recovery methods with theoretical guarantee. Hence, \lscr{the obtained common measurement matrices and recovery methods} \lscb{can significantly outperform} the underlying advanced recovery methods. We conduct extensive numerical results on channel estimation and device activity detection in MIMO-based grant-free random access. The numerical results show that the proposed approaches \lscb{provide} pilot sequences and channel estimation or device activity detection methods \lscb{which can} achieve \lscb{higher} \lscr{estimation or detection} accuracy with shorter computation time than \lscb{existing ones}. Furthermore, the numerical results explain how such gains are achieved \lscb{via the} proposed approaches.
\end{abstract}

\begin{IEEEkeywords}
Compressive sensing, jointly sparse support recovery, jointly sparse signal recovery, mMTC, grant-free random access, device activity detection, channel estimation, optimization, auto-encoder, deep learning.
\end{IEEEkeywords}

\section{Introduction}
\label{intro}
Jointly sparse signal recovery and jointly sparse support recovery in Multiple Measurement Vector (MMV) models refer to the estimation of $M$ jointly sparse $N$-dimensional signals and the estimation of the common support of $M$ jointly sparse $N$-dimensional signals, respectively, from $L$ ($\ll$${N}$) limited noisy linear measurements based on a common measurement matrix. When $M=1$, they reduce to sparse signal recovery and sparse support recovery in Single Measurement Vector (SMV) models, respectively. Two main challenges exist in jointly sparse signal recovery and jointly sparse support recovery. One is to design a common measurement matrix which maximally retains the information on jointly sparse signals or their common support when reducing signal dimension. The other is to recover the jointly sparse signals or the common support with high recovery accuracy and short computation time.

The jointly sparse signal recovery problem and the jointly sparse support recovery problem arise in many applications in communications and signal processing \cite{zhang2018ista,7952561,hammernik2018learning, 8437359,chen2019covariance,DD2019,liu2018sparse,8264818,8323217,8323218,yang2018admm}. Recently, their key applications include channel estimation and device activity detection in MIMO-based grant-free random access. \lscr{Grant-free random access has been widely regarded as one promising solution for supporting} massive machine-type communications (mMTC) for Internet of Things (IoT), \lscr{which is} one of the three main use cases in 5G \cite{8437359,chen2019covariance,DD2019,liu2018sparse,8264818,8323217,8323218}. In mMTC, there are massive \lscr{IoT devices in each cell. However}, only a small number of devices are active at a time, and a very small amount of data is transmitted from each active device. In grant-free random access, each device is assigned a specific pilot sequence, all active devices send their pilot sequences, and each base station (BS) detects the activities of its associated devices or estimates their channel states, \lscr{depending on specific application requirements}. Common measurement matrix design, jointly sparse signal recovery and jointly sparse support recovery for complex signals \lscr{correspond to design of pilot sequences}, channel estimation \cite{8323217,8323218} and device activity detection \cite{8437359,chen2019covariance,DD2019,liu2018sparse,8264818}, respectively, \lscr{in MIMO-based grant-free random access, which has} gained increasing interest and attention. In this paper, our primary \lscr{goals are} to address the aforementioned two challenges in jointly sparse signal recovery and jointly sparse support recovery for complex signals, \lscr{and to} provide practical solutions with high recovery accuracy and short computation time for channel estimation and device activity detection in MIMO-based grant-free random access.

Most existing works on sparse support recovery \lscr{for SMV models \cite{5319750, 5319742,4839045}} and jointly \lscb{sparse} support recovery \lscr{for MMV models \cite{6994860,8437359,chen2019covariance,DD2019,liu2018sparse,8264818,8323218,ke2020compressive, wei2016approximate}} focus on \lscr{investigating \lscb{sparse} support recovery methods} for a given measurement matrix. For example, optimization-based methods, such as least absolute shrinkage and selection operator (LASSO) \cite{4839045, 6994860}, maximum likelihood (ML) estimation  \cite{8437359,chen2019covariance} and maximum a posterior (MAP) estimation \cite{DD2019}\lscr{,} are adopted for sparse support recovery. In particular, \cite{4839045} directly deals with noisy linear measurements, while \cite{6994860,8437359,chen2019covariance,DD2019} operate on the covariance matrix of noisy linear measurements. \lscr{In addition}, iterative thresholding methods\lscr{,} such as approximate message passing (AMP) \cite{liu2018sparse,8264818} and \lscr{expectation maximization-based approximate message passing (EM-AMP) \cite{ke2020compressive, wei2016approximate}, are} \lscb{adopted} to achieve jointly sparse support recovery. Sparse signal recovery \lscr{for SMV models} \cite{tibshirani1996regression,donoho2009message} and jointly sparse signal recovery \lscr{for MMV models} \cite{qin2013efficient,ziniel2012efficient,8323217,8323218,ke2020compressive, wei2016approximate} are more widely investigated. Similarly, most \lscb{existing studies on sparse signal recovery for SMV models \cite{tibshirani1996regression,donoho2009message} and jointly sparse signal recovery for MMV models  \cite{qin2013efficient,ziniel2012efficient,8323217,8323218,ke2020compressive, wei2016approximate} focus on tackling sparse} signal recovery problems for a given measurement matrix. Classic methods include optimization-based methods\lscr{, such as LASSO \cite{tibshirani1996regression} and GROUP LASSO \cite{qin2013efficient}, and iterative thresholding methods, such as AMP \cite{donoho2009message,8323217,8323218,ziniel2012efficient} and EM-AMP \cite{ke2020compressive, wei2016approximate}.}

\lscr{Notice that LASSO \cite{4839045, 6994860,tibshirani1996regression} and GROUP LASSO \cite{qin2013efficient} do not rely on any information of signals or noise; ML \cite{8437359,chen2019covariance} relies on the distributions of non-zero components of signals and noise; MAP \cite{DD2019} and AMP \cite{liu2018sparse,8264818,donoho2009message,8323217,8323218,ziniel2012efficient} are based on the distributions of all components of signals and noise; and EM-AMP \cite{ke2020compressive, wei2016approximate} requires the forms of the distributions of all components of signals and noise with unknown parameters, which are learnt by EM. Furthermore, note that ML  \cite{8437359,chen2019covariance}, MAP \cite{DD2019}, AMP \cite{liu2018sparse,8264818,donoho2009message,8323217,8323218,ziniel2012efficient} and EM-AMP \cite{ke2020compressive, wei2016approximate} all assume that the components of signals are independent, and hence their recovery performance may be unsatisfactory when such assumption is not satisfied.} \lscr{When the components of signals are correlated and prior distributions do not have analytic models, neural networks are recently utilized to exploit properties of sparse signals from training samples, \lscv{for the purpose of} designing effective sparse signal recovery methods \lscv{in such scenario}. For instance}, \cite{gregor2010learning,yao2017deepiot,he2018deep,taha2019enabling,yang2018admm,zhang2019deep,wu2019learning} exploit properties of sparsity patterns of real signals \cite{gregor2010learning,yao2017deepiot,he2018deep,wu2019learning} and complex signals \cite{taha2019enabling,yang2018admm,zhang2019deep} from training samples\lscr{,} using data-driven \cite{gregor2010learning,yao2017deepiot,taha2019enabling} or model-driven \cite{he2018deep,yang2018admm,zhang2019deep,wu2019learning} approaches\lscr{,} to improve sparse signal recovery performance in SMV models. More specifically, \cite{yang2018admm} adopts a sequential alternating direction method of multipliers (ADMM) algorithm (which \lscv{does not allow parallel computation and thus} cannot make use of the parallelizable neural network architecture) for solving the LASSO problem, \cite{wu2019learning} \lscv{utilizes} a gradient-based algorithm for solving the basis pursuit problem, and \cite{zhang2019deep} relies on \wq{an} iterative thresholding algorithm. \lscv{It is worth noting that} in general, a model-driven approach uses fewer training samples, has \lscb{better} performance guarantee, and provides \lscb{more} design insight than a data-driven approach.

\lscr{Design of the common measurement matrix is seldomly studied.} Very few papers \cite{candes2008restricted, 8264817} investigate the impact of the measurement matrix \lscr{in sparse \wq{signal} recovery}. For instance, \lscr{in \cite{candes2008restricted}, the \wq{authors show}} that a measurement matrix can preserve \lscr{sparsity information} in sparse signals, if it satisfies the restricted isometry property (RIP). In addition, in \cite{8264817}, the authors consider group sparse signals and show that block-coherence and sub-coherence of a measurement matrix affect signal recovery performance. However, the results in \cite{candes2008restricted, 8264817} may not hold \lscr{for sparse support recovery}.

It has been shown that joint design of signal compression and recovery methods for real signals \cite{wu2019learning,nguyen2017deep,8262812,adler2017block} or complex signals \cite{8322184,8861085} using \lscr{deep auto-encoders} can significantly improve recovery performance. \lscv{Note that} \cite{wu2019learning,nguyen2017deep,8262812,adler2017block,8322184,8861085} are all \lscr{for SMV} models, and their extensions to MMV models are highly nontrivial. In addition, note that neither the neural network for complex signals in \cite{8322184} nor direct extensions of the neural networks for real signals \lscr{in} \cite{wu2019learning,nguyen2017deep,8262812,adler2017block} to complex signals can achieve linear compression for complex signals. Thus, how to jointly design the common measurement matrix and jointly sparse signal or support recovery methods   in MMV models for complex signals remains open.

In this paper, we investigate jointly sparse signal recovery and jointly sparse support recovery in MMV models for complex signals. Utilizing techniques in compressive sensing, optimization, and deep learning, we propose two-model driven approaches, based on the standard auto-encoder structure for real numbers. \lscv{Specifically}, one is to jointly design the common measurement matrix and jointly sparse signal recovery method, and the other aims to jointly design the common measurement matrix and jointly sparse support recovery method, both for complex signals. \lscr{A common goal is to effectively utilize features of sparsity patterns in designing common measurement matrices \wq{and recovery methods} and to greatly benefit from the underlying state-of-the-art recovery methods with theoretical guarantee}. The main contributions of this paper are summarized as follows.
\begin{itemize}
\item The model-driven approach for jointly sparse signal recovery consists of an encoder that mimics the noisy linear measurement process and a model-driven decoder that mimics the jointly sparse signal recovery process of a particular method \lscr{using an approximation part and a correction part}. We further propose two instances for the model-driven decoder, namely, the GROUP LASSO-based decoder and the AMP-based decoder. \lscv{They} are based on \lscb{an} \lscr{ADMM} algorithm for GROUP LASSO \cite{boyd2011distributed} and the AMP algorithm with the \lscv{Minimum Mean Squared Error (MMSE)} denoiser \cite{8323218}, respectively, and are effective at different system parameters. \lscr{Note that the ADMM algorithm for GROUP LASSO can make use of the parallelizable neural network structure, while the block coordinate descent algorithm for GROUP LASSO in \cite{qin2013efficient,LS2019} cannot. In addition}, note that the LASSO-based \lscr{neural networks in \cite{yang2018admm,wu2019learning}} and the AMP-based neural network in \cite{zhang2019deep} are \lscr{for SMV} models.
\item The model-driven approach for jointly sparse support recovery consists of an encoder which is the same as in the model-driven approach for jointly sparse signal recovery, a model-driven decoder that mimics the jointly sparse support recovery process of a certain method \lscr{via an approximation part and a correction part}, and a thresholding module \lscv{that implements} binary approximation to obtain the common support. We also propose two instances for the model-driven decoder, namely, the covariance-based decoder and the MAP-based decoder. \lscv{They} are based on a covariance-based estimation method \cite{6994860} and a coordinate decent algorithm for the MAP estimation \cite{DD2019}, respectively, and are suitable for different system parameters. \lscr{It is worth noting that the proposed covariance-based decoder can achieve much better recovery performance than the covariance-based estimation method in \cite{6994860} when $M$ is not large.} To our knowledge, this is the first time that the MAP estimation \lscv{for jointly sparse support recovery} is implemented using \lscr{a} neural network, thanks to the structure of the coordinate descent algorithm recently proposed in \cite{DD2019}. Besides, \lscv{when followed by a thresholding module,} the AMP-based decoder for jointly sparse signal recovery can also be used as a model-driven decoder for jointly sparse support recovery.
\item We consider two application examples, namely, channel estimation and device activity detection in MIMO-based grant-free random access, and apply our proposed approaches therein. By extensive numerical results, we demonstrate the substantial gains of the proposed approaches over existing methods in terms of both recovery accuracy and computation time, and show \lscr{how such gains} are achieved. Specifically, the proposed model-driven approaches can effectively \lscr{utilize features of sparsity patterns} in designing the common measurement matrix and adjusting the correction layers of the model-driven decoders; they benefit from the state-of-the-art underlying recovery methods via the approximation \lscb{parts} of the \lscv{model-driven} decoders; and they provide recovery methods with higher recovery accuracy and \lscb{shorter} computation time, owing to the careful choices for the numbers of the approximation layers and correction layers in the \lscv{model-driven} decoders.
\end{itemize}

%\textcolor{green}{In this paper, we extend our former letter [] which proposes a data-driven auto-encoder architecture to jointly design the measurement matrix and support recovery method for complex sparse signals in SMV model to MMV model, using deep auto-encoder. }Similarly, the proposed architecture includes an auto-encoder and a hard thresholding module. The auto-encoder consists of an encoder which mimics the noisy linear measurement process, and a decoder which approximately performs sparse support recovery from the under-sampled linear measurements. The proposed auto-encoder successfully handles complex signals using standard auto-encoder for real numbers. The data-driven approach is especially useful when the underlying structures of sparsity patterns are hard to model, and can achieve sparse support recovery with low computational complexity due to the parallelizable neural network architecture. Experiments are conducted on an application example: device activity detection in grant-free massive access for massive machine type communications (mMTC). Numerical results show that the proposed approach achieves significantly better performance with much less computation time than classic methods, in the presence of additional properties of sparsity patterns. The substantial gains derive from the effective joint design that exploits these structures.

\subsection*{\textbf{Notation}}
We represent vectors by boldface small letters (e.g., $\mathbf{x}$), matrices by boldface capital letters (e.g., $\mathbf{X}$), scalar constants by non-boldface letters (e.g., $x$ or $X$) and sets by calligraphic letters (e.g., $\mathcal{X}$). The notation $X(i,j)$ denotes the $(i,j)$-th element of matrix $\mathbf{X}$, $\mathbf{X}_{i,:}$ represents the $i$-th row of matrix $\mathbf{X}$, $\mathbf{X}_{:,i}$ represents the $i$-th column of matrix $\mathbf{X}$, and $x(i)$ represents the $i$-th element of vector $\mathbf{x}$. Superscript $^H$ , superscript $^T$ and superscript $^*$ denote transpose conjugate, transpose and conjugation, respectively. The notation $\||\cdot\||_F$ represents the Frobenius norm of a matrix, ${\rm vec}(\cdot)$ denotes the column vectorization of a matrix, ${\rm tr}(\cdot)$ denotes the trace of a matrix, ${\rm Cov}(\cdot)$ represents the covariance matrix of a random vector, $\odot$ represents the Khatri-Rao product between two matrices, $\mathbb{I}[\cdot]$ denotes the indicator function, and $\Re(\cdot)$ and $\Im(\cdot)$ represent the real part and imaginary part, respectively. $\mathbf{0}_{m\times n}$ and $\mathbf{I}_{n\times n}$ represent the $m\times n$ zero matrix and the $n\times n$ identity matrix, respectively. The complex field and real field are denoted by $\mathbb{C}$ and $\mathbb{R}$, respectively.

\section{Problems and Applications}
\label{app}
In this section, we first introduce jointly sparse signal recovery and jointly sparse support recovery \lscr{in MMV models for complex signals}. Then, we illustrate their application examples, i.e., \lscr{channel estimation and device activity detection} in MIMO-based grant-free random access.

The support of a sparse $N$-dimensional complex signal $\mathbf{x}\in \mathbb{C}^N$ is defined as the set of locations of non-zero elements of $\mathbf{x}$, and is denoted by ${\rm supp}(\mathbf{x}) \triangleq \{n\in\mathcal{N}|x(n) \neq 0\}$, where $\mathcal{N} \triangleq \{1,\cdots,N\}$. If the number of non-zero elements of $\mathbf{x}$ is much smaller than \lscr{the total number of elements of $\mathbf{x}$}, i.e., $|{\rm supp}(\mathbf{x})| \ll N $, $\mathbf{x}$ is sparse. \lscr{Consider $M>1$ jointly sparse signals $\mathbf{x}_m\in \mathbb{C}^N, m \in \mathcal{M} \triangleq \{1,\cdots,M\}$, which share a common support, i.e., ${\rm supp}(\mathbf{x}_m), m \in \mathcal{M}$ are identical. Let $\mathcal{S}$ denote the common support. Denote $\alpha(n)\triangleq \mathbb{I}[n \in \mathcal{S}]$. Equivalently, $\mathcal{S}$ can be expressed using $\boldsymbol{\alpha}\in\{0,1\}^{N}$, i.e.,} $\mathcal{S}=\{n \in \mathcal{N}|\alpha(n)=1\}$. For all $m\in \mathcal{M}$, consider $L \ll N$ noisy linear measurements $\mathbf{y}_m\in \mathbb{C}^L$ of $\mathbf{x}_m$, i.e., $\mathbf{y}_m = \mathbf{A}\mathbf{x}_m+\mathbf{z}_m$, where $\mathbf{A}\in \mathbb{C}^{L\times N}$ is the common measurement matrix, and $\mathbf{z}_m\sim\mathcal{CN}(\mathbf{0}_{L\times 1},\sigma^2\mathbf{I}_{L\times L})$ is the additive white Gaussian noise. \lscr{Define} $\mathbf{X} \in \mathbb{C}^{N\times M}$ with $\mathbf{X}_{:,m} = \mathbf{x}_m, m\in \mathcal{M}$, $\mathbf{Y}\in \mathbb{C}^{L\times M}$ with $\mathbf{Y}_{:,m} = \mathbf{y}_m, m\in \mathcal{M}$ and $\mathbf{Z}\in \mathbb{C}^{L\times M}$ with $\mathbf{Z}_{:,m} = \mathbf{z}_m, m\in \mathcal{M}$, respectively. \lscr{More compactly}, we have:
\begin{align}
\mathbf{Y} =\mathbf{A}\mathbf{X}+\mathbf{Z}
\end{align}
\begin{itemize}
\item Jointly sparse signal recovery in MMV models aims to \lscr{estimate $M$} jointly sparse signals $\mathbf{x}_m,m \in \mathcal{M}$ (i.e., $\mathbf{X}$) from $M$ noisy linear measurement vectors $\mathbf{y}_m,m \in \mathcal{M}$ (i.e., $\mathbf{Y}$), obtained through a common measurement matrix $\mathbf{A}$.

\item Jointly sparse support recovery in MMV models aims to identify the common support $\mathcal{S}$ (or $\boldsymbol{\alpha}$) shared by $M$ jointly sparse signals $\mathbf{x}_m,m \in \mathcal{M}$ (i.e., $\mathbf{X}$) from $M$ noisy linear measurement vectors $\mathbf{y}_m,m \in \mathcal{M}$ (i.e., $\mathbf{Y}$), obtained through a common measurement matrix $\mathbf{A}$.
\end{itemize}

\begin{figure*}[htp]
\begin{center}
 \subfigure[\scriptsize{\lscr{Proposed approach with GROUP LASSO-based decoder. The \lsccc{approximation} part implements $U$ iterations of Algorithm~1.}}]
 {\resizebox{17.4cm}{!}{\includegraphics{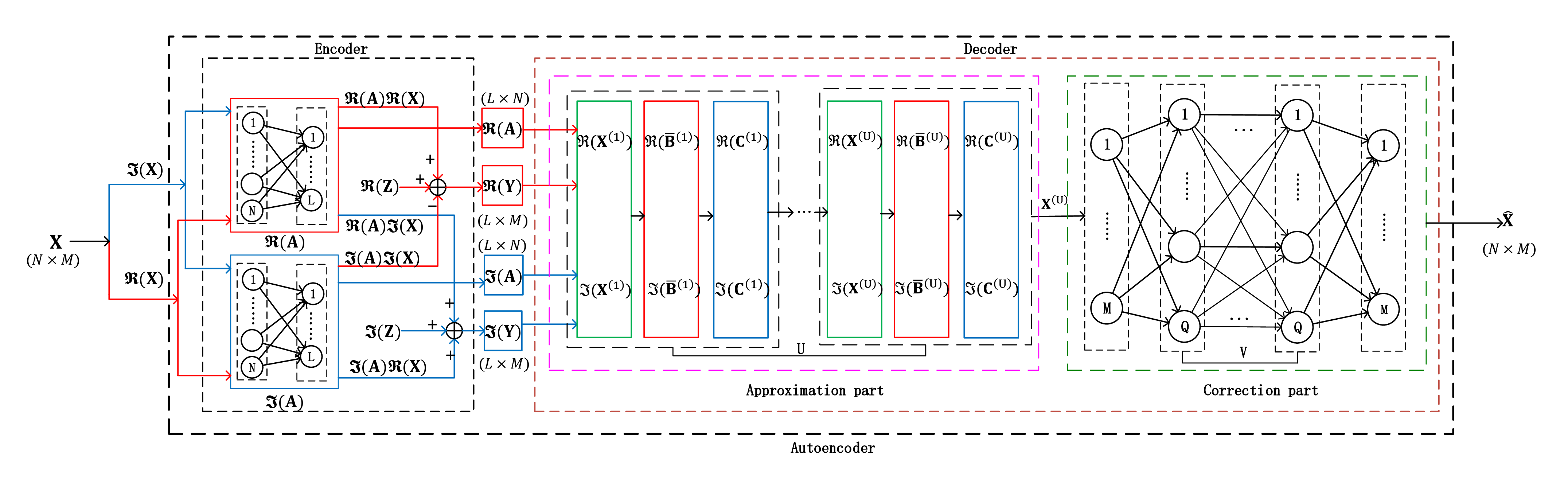}}}
 \subfigure[\scriptsize{Proposed approach with AMP-based decoder. \lscr{The \lsccc{approximation} part implements $U$ iterations of Algorithm~2.}}]
 {\resizebox{17.4cm}{!}{\includegraphics{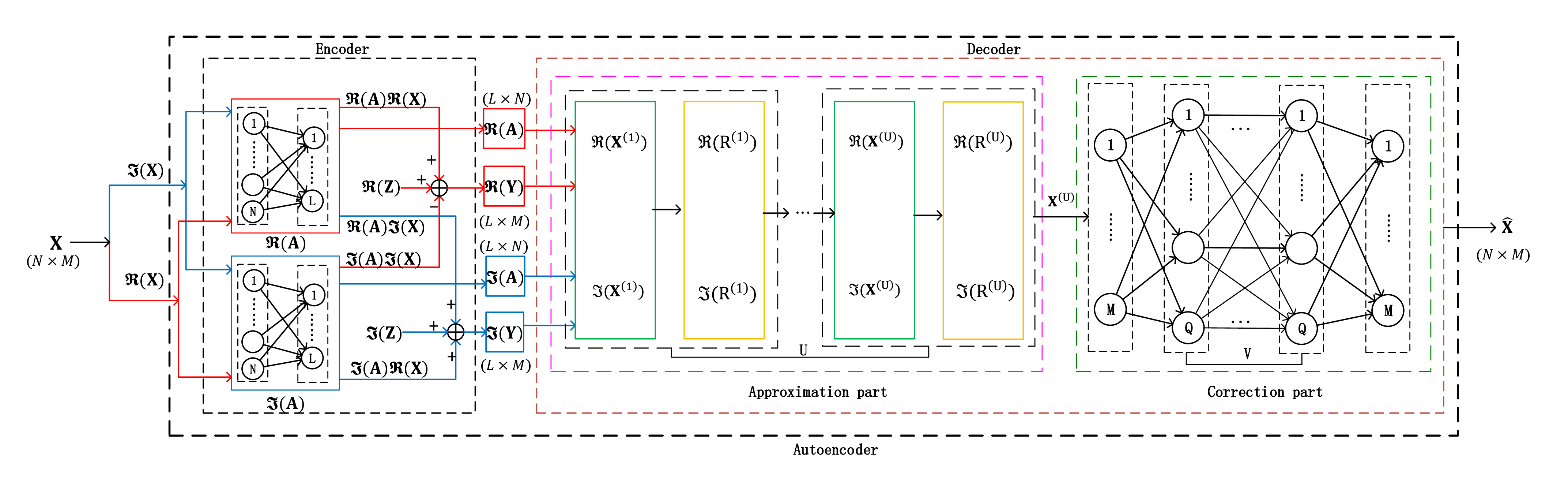}}}
  \end{center}
  \vspace*{-0.3cm}
     \caption{Proposed model-driven approach for jointly sparse signal recovery.}
     \vspace*{-0.3cm}
\label{networkmodel-sig}
\end{figure*}

As important application examples, we consider \lscr{channel estimation and device activity detection in} MIMO-based grant-free random access, which is recently proposed to support massive access in mMTC for IoT \cite{8323218,8264818,8323217,8846802,liu2018sparse}. Consider a single cell with one $M$-antenna base station (BS) and $N$ single-antenna (IoT) devices. \lscr{For all $n\in\mathcal{N},$ let $\alpha(n) \in \{0,1\}$ represent} the active state of device $n$, where $\alpha(n)=1$ means that \lscr{device $n$} accesses the channel (i.e., is active), and $\alpha(n)=0$ otherwise. \lscr{For all $m \in \mathcal{M}$ and $n\in\mathcal{N}$}, let $h_{m}(n)\in \mathbb{C}$ denote the state of the wireless channel between the $m$-th antenna at the BS and device $n$, and view $\alpha(n)h_{m}(n)$ as $x_m(n)$. That is, $x_m(n)$ represents the corresponding channel state if it is not zero. \lscr{As device} activity patterns for IoT traffic are typically sporadic, $\mathbf{x}_{m}\in \mathbb{C}^N,m \in \mathcal{M}$ are jointly sparse \lscr{with common} support $\mathcal{S}=\{n \in \mathcal{N}|\alpha(n)=1\}$. In grant-free random access, each device $n$ has a unique pilot sequence $\mathbf{a}_n \in \mathbb{C}^L$ of length $L \ll N$. View the pilot matrix $\mathbf{A} \in \mathbb{C}^{L\times N}$ with $\mathbf{A}_{:,n} = \mathbf{a}_n, n\in \mathcal{N}$, which is known at the BS, as the common measurement matrix. In the pilot transmission phase, active devices synchronously send their pilot sequences to the BS. Then, $\mathbf{Y}$ in (1) represents the received signal at the BS. There are two potential tasks for the BS in the pilot transmission phase. One is channel estimation, and the other is device activity detection.

\begin{itemize}
\item In channel estimation, the BS aims to estimate the channel states of all active devices, i.e., $h_{m}(n), m \in \mathcal{M}$ for all $n \in \mathcal{N}$ with $\alpha(n)=1$, via estimating $\mathbf{x}_{m}\in \mathbb{C}^N, m \in \mathcal{M}$ \big(i.e., $\mathbf{X}$\big) from $\mathbf{Y}$.  \lscb{Obviously, channel estimation in MIMO-based grant-free random access} corresponds to jointly sparse signal recovery in MMV models \cite{8437359,chen2019covariance,DD2019,liu2018sparse,8264818}.

\item In device activity detection, the BS tries to estimate the device activities, i.e., $\alpha(n)$, $n \in \mathcal{N}$, by estimating \lscr{the common support  $\mathcal{S}$ of $\mathbf{x}_{m}\in \mathbb{C}^N,m \in \mathcal{M}$ from $\mathbf{Y}$. \lscb{Apparently, device activity detection in MIMO-based grant-free random access} corresponds to} jointly sparse support recovery in MMV models \cite{8323217,8323218}.
\end{itemize}

\lscr{There are} some applications where devices have no data to transmit or only very limited data to transmit (which can be embedded into the pilots \cite{chen2019covariance,8323217}). \lscr{For such applications, there exits only the pilot transmission phase, and the BS performs device activity detection in the pilot transmission phase. There are also applications where each device has certain amount of data to transmit. For those applications, there is a \lscb{subsequent} phase following the pilot transmission phase, namely the data transmission phase, and channel states are necessary for the BS to receive data from active devices in the data transmission phase. Hence, in the pilot transmission phase, the BS directly performs channel estimation rather than device activity detection.}

%additive white Gaussian noise, and the ${\rm diag}(s_{1},\cdots,s_{k})$ denotes a diagonal matrix with elements $(s_{1},\cdots,s_{k})$.The main goal of the BS in the pilot transmission phase is to detect user activities by recovering $\boldsymbol{\alpha}$ from the noisy observation $\mathbf{y}_{\rm pilot}$. If we use noisy linear measurements $\mathbf{y}$ to denote $\mathbf{y}_{\rm pilot}$, measurement matrix $\mathbf{A}$ to denote pilot matrix $\mathbf{A}^{\rm pilot}$, sparse signal $\mathbf{x}$ to denote ${\rm diag}(\boldsymbol{\alpha})\mathbf{h}$, and $\mathbf{n}$ to denote $\mathbf{z}$, then the activity detection problem in pilot transmission phase can be equivalent to the sparse support recovery problem described above.

%In this work, we mainly try to recover the support of $\mathbf{x}$, which is the set of indices corresponding to the nonzero entries and can be denoted by the sparse activation vector $\mathbf{a}\in \{0,1\}^N$.

\section{A Model-driven Approach for jointly sparse signal recovery}
\label{approachsig}
In this section, we propose a model-driven approach, based on the standard auto-encoder structure for real numbers in deep learning, to jointly\footnote{\lscb{Note that the proposed model-driven approach can also be used for designing a sparse signal recovery method for a given common measurement matrix via fixing the encoder while training the decoder.}} design the common measurement matrix and jointly sparse signal recovery method for sparse complex signals. As shown in Fig.~\ref{networkmodel-sig}, the model-driven approach consists of an encoder that mimics the noisy linear measurement process and a model-driven decoder that mimics the jointly sparse signal recovery process of a particular method. We further propose two instances for the model-driven decoder for jointly sparse signal recovery, namely the GROUP LASSO-based decoder and the AMP-based decoder. \lscb{An encoder that is} jointly trained with \lscb{a model-driven decoder provides} the respective common measurement \lscb{matrix}. After training, \lscr{the design} of the common measurement matrix \lscr{can be obtained by} extracting the weights of the encoder\lscr{, and} the model-driven decoder can be directly used for jointly sparse signal recovery. Note that the obtained design of the common measurement matrix and model-driven decoder for jointly sparse signal recovery should be jointly utilized, as a common measurement matrix designed for one jointly sparse signal recovery method may not be effective for another.
%First, we introduce the converting module that converts the sparse signal $\mathbf{x}$ and measurement matrix $\mathbf{A}$ in terms of complex numbers to one in terms of real numbers. Then, we elaborate on the auto-encoder structure where the encoder mimics the noisy linear measurement process, and the encoder approximates the sparse support recovery process. Finally, we introduce the hard thresholding module which guarantees to achieve sparse support recovery.

\subsection{Encoder for Jointly Sparse Signal Recovery}
The encoder mimics the noisy linear measurement process for jointly sparse signals with a common measurement matrix in (1). To mimic (1) using \lscr{a} standard auto-encoder for real numbers in deep learning, we equivalently express (1) as:
\begin{align}
\Re(\mathbf{Y})=\Re(\mathbf{A})\Re(\mathbf{X})-\Im(\mathbf{A})\Im(\mathbf{X})+\Re(\mathbf{Z})\\
\Im(\mathbf{Y})=\Im(\mathbf{A})\Re(\mathbf{X})+\Re(\mathbf{A})\Im(\mathbf{X})+\Im(\mathbf{Z})
\end{align}
Two fully-connected neural networks, each with two layers, are built to implement the multiplications with matrices $\Re(\mathbf{A})\in \mathbb{R}^{L\times N}$ and $\Im(\mathbf{A})\in \mathbb{R}^{L\times N}$, respectively. For each neural network, there are $N$ neurons and $L$ neurons in the input layer and output layer, respectively; the weight of the connection from the $n$-th neuron in the input layer to the $l$-th neuron in the output layer corresponds to the $(l,n)$-th element of the corresponding matrix; and no activation functions are used in the output layer. The elements of $\Re(\mathbf{Z})\in \mathbb{R}^{L\times M}$ and $\Im(\mathbf{Z})\in \mathbb{R}^{L\times M}$ are generated independently according to $\mathcal{N}(0,\frac{\sigma^2}{2})$. As shown in Fig.~\ref{networkmodel-sig}, when $\Re(\mathbf{X})\in \mathbb{R}^{N\times M}$ and $\Im(\mathbf{X})\in \mathbb{R}^{N\times M}$ are input to the encoder, $\Im(\mathbf{Y})\in \mathbb{R}^{L\times M}$ and $\Re(\mathbf{Y})\in \mathbb{R}^{L\times M}$ can be readily obtained.

\subsection{Model-driven Decoder for Jointly Sparse Signal Recovery}
\label{signal}
\lscr{In this subsection}, we propose a model-driven decoder \lscr{that} consists of two parts, namely the approximation part and the correction part. As illustrated in Fig.~\ref{networkmodel-sig}, the approximation part which is used to approximate  a particular method for jointly sparse signal recovery has $U$ ($\geq 0$) \wq{building blocks}. We can directly input $\mathbf{Y}$ and $\mathbf{A}$ to the approximation part and obtain $\mathbf{X}^{(U)}$ as the output. Note that $\mathbf{X}^{(U)}$ can be treated as an approximation of the estimate obtained by the \lscr{underlying} method, \lscr{and the approximation error decreases with $U$.} Two instances for the approximation part will be illustrated shortly. The correction part \lscr{which aims to reduce the difference between the obtained approximate estimation $\mathbf{X}^{(U)}$ and the actual jointly sparse signals $\mathbf{X}$} consists of $V$ ($\geq 1$) fully connected layers. Specially, \lscc{the \lscv{first correction layer} has \lsccc{$M$ neurons with $\left(\Re(\mathbf{X}^{(U)}_{:,m}),\ \Im(\mathbf{X}^{(U)}_{:,m})\right), \ m\in\mathcal{M}$ as the input.}} The last correction layer has \lsccc{$M$} neurons with \lsccc{$\left(\Re(\mathbf{\hat{X}}_{:,m}),\ \Im(\mathbf{\hat{X}}_{:,m})\right), \ m\in\mathcal{M}$} \lscv{as the output.} In each of the first $V-1$ correction layers, rectified linear unit (ReLU) is chosen as the activation function. \lscv{In the last correction layer, there is no activation function.} \lscr{$V$ influences the training error and generalization error.} Note that $U$ and $V$ \lscr{should be} jointly chosen \lscr{so that} the proposed model-driven decoder can achieve higher recovery accuracy \lscr{(i.e., a \wq{smaller} gap between $\hat{\mathbf{X}}$ and $\mathbf{X}$)} and shorter computation time than \lscr{the underlying} method. When $U$ is sufficiently large and $V=0$,  the proposed model-driven decoder reduces to a particular method, \lscr{and} the proposed model-driven approach is used for designing a common measurement matrix for \lscr{such} method. When $U=0$ and $V>0$, the proposed model-driven approach degrades to a purely data-driven one for \lscr{joint} design of \lscr{the} common measurement matrix and jointly sparse signal recovery method.

In the following, we propose two instances of the model-driven decoder, more precisely, two instances of its approximation part, \lscr{based} on a \lscr{carefully designed ADMM} algorithm for GROUP LASSO and the AMP algorithm in \cite{8323218}, respectively. They are referred as the GROUP LASSO-based decoder and the AMP-based decoder, and are illustrated in Fig.~\ref{networkmodel-sig} (a) and Fig.~\ref{networkmodel-sig} (b), respectively.

\subsubsection{GROUP LASSO-based Decoder}
\label{sigGL}
In this part, we \lscb{present} the GROUP LASSO-based decoder. First, we introduce GROUP LASSO, which is a natural optimization-theoretic formulation of the jointly sparse signal recovery problem in MMV models \cite{qin2013efficient}\lscr{:}
\vspace*{-0.2cm}

\begin{align}
\label{GL}
\min_\mathbf{X}\;\frac{1}{2}\||\mathbf{A}\mathbf{X}-\mathbf{Y}\||_F^2+\lambda\lscb{\sum_{n\in\mathcal{N}}\|\mathbf{X}_{n,:}\|_2}
\end{align}
where $\lambda \geq 0$ is a \lscr{regularization} parameter. Let $\mathbf{X}^\star$ denote an optimal solution of this problem. \lscr{Note that GROUP LASSO does not rely on any information of sparse signals or noise, and the difference between $\mathbf{X}^\star$ and the actual jointly sparse signals \wq{$\mathbf{X}$} varies with $\lambda$. When} $M=1$, the formulation in (\ref{GL}) reduces to LASSO \cite{tibshirani1996regression}. \lscr{The problem in \eqref{GL} is convex and can be solved optimally. Applying standard convex optimization methods for solving it may not be computationally efficient. For instance, the computational complexity of the gradient descent method is $\mathcal{O}(N^2M)$, which is prohibitively high at large $N$. In \cite{qin2013efficient,LS2019}, the block coordinate descent method with computational complexity $\mathcal{O}(LNM)$ is adopted. \lscb{In particular, $\mathbf{X}_{n,:}$, $n\in\mathcal{N}$ are updated in a sequential manner.} Although the computational complexity \lscb{at large $N$} is reduced, the sequential operations in the block coordinate descent method cannot make use of the parallelizable neural network architecture, and hence still has unsatisfactory computation time.}

\lscr{Next, by carefully exploiting structural properties of the problem in \eqref{GL} and using ADMM \cite{boyd2011distributed}, we develop a fast algorithm which allows parallel operations and can utilize the parallelizable neural network architecture. \lscb{The problem in \eqref{GL} can be rewritten as:
$$\min_{\mathbf{X}}\ \frac{1}{2}\|\sum_{n\in\mathcal{N}}\mathbf{A}_{:,n}\mathbf{X}_{n,:}-\mathbf{Y}\|_F^2+\lambda\sum_{n\in\mathcal{N}} \|\mathbf{X}_{n,:}\|_2$$
We follow the approach used for the sharing problem in \cite{boyd2011distributed}.} \lscb{Specifically,} introducing auxiliary variables $\mathbf{B}_{n}\in\mathbb{C}^{L\times M}, n\in\mathcal{N}$ and extra constraints $\mathbf{A}_{:,n}\mathbf{X}_{n,:}-\mathbf{B}_n=0,n\in\mathcal{N}$, we can obtain the following equivalent problem in ADMM form:}
\lscr{\begin{align}
\label{admm}
	\min_{\mathbf{X},\mathbf{B}_n,n\in\mathcal{N}}\ &\frac{1}{2}\|\sum_{n\in\mathcal{N}}\mathbf{B}_n-\mathbf{Y}\|_F^2+\lambda\sum_{n\in\mathcal{N}} \|\mathbf{X}_{n,:}\|_2\nonumber\\
	\mathrm{s.t.} \quad \ &\mathbf{A}_{:,n}\mathbf{X}_{n,:}-\mathbf{B}_n=0, \quad n \in \mathcal{N}
\end{align}}By the scaled form of ADMM \cite{boyd2011distributed}, we \lscb{can obtain the update equations in \eqref{xx}, \eqref{zz} and \eqref{uu}, as shown at the top of the next page,} where $k$ is the iteration index, $\mathbf{\overline{AX}}^{(k)}=\frac{1}{N}\lscb{\sum_{n\in\mathcal{N}}}\mathbf{A}_{:,n}\mathbf{X}^{(k)}_{n,:}$, $\mathbf{\overline{B}}^{(k)} = \frac{1}{N}\lscb{\sum_{n\in\mathcal{N}}}\mathbf{B}^{(k)}_{n}$, \lscb{$\mathbf{C}^{(k)} \in \mathbb{C}^{L\times M}$ correspond to the dual variables with respect to the constraints in \eqref{admm}}, and $\rho>0$ is the augmented Lagrangian parameter. The problems in \eqref{xx} and \eqref{zz} are convex. By setting the gradients of their objective functions to be zero, we can obtain the optimal solutions.
\begin{figure*}[htp]
\begin{align}
\label{xx}
\mathbf{X}^{(k+1)}_{n,:} = \ &\arg \min_{\mathbf{X}_{n,:}} \ \lambda\|\mathbf{X}_{n,:}\|_2 + \lscb{\frac{\rho}{2}}\|\mathbf{A}_{:,n}\mathbf{X}_{n,:} -\mathbf{A}_{:,n}\mathbf{X}_{n,:}^{(k)}-\mathbf{\overline{B}}^{(k)} +\mathbf{\overline{AX}}^{(k)} +\mathbf{C}^{(k)}\|_F^2, \quad n\in\mathcal{N}\\
\label{zz}
\mathbf{\overline{B}}^{(k+1)} = \ &\arg\min_{\mathbf{\overline{B}}} \ \frac{1}{2}\|N\mathbf{\overline{B}}-\mathbf{Y}\|_F^2 + \lscb{\frac{N\rho}{2}}\|\mathbf{\overline{B}} -\mathbf{\overline{AX}}^{(k+1)} -\mathbf{C}^{(k)}\|_F^2 \\
\label{uu}
\mathbf{C}^{(k+1)} = \ &\mathbf{C}^{(k)} + \mathbf{\overline{AX}}^{(k+1)} - \mathbf{\overline{B}}^{(k+1)}
\end{align}
\hrulefill
\end{figure*}
\begin{algorithm}[tp]
\label{ADMM}
    \caption{ADMM for GROUP LASSO}
    \begin{algorithmic}[1]
        \STATE Set ${\mathbf{X}^{(0)}}=\mathbf{0}_{N\times M}$, $\mathbf{B}^{(0)}=\mathbf{0}_{N\times M}$, $\mathbf{C}^{(0)}=\mathbf{0}_{N\times M}$ and $k=0$.
	 \REPEAT
	\STATE For all $n\in\mathcal{N}$, compute $\mathbf{X}^{(k+1)}_{n,:}$ according to \eqref{cfxx}.

	\STATE Compute $\mathbf{\overline{B}}^{(k+1)}$ according to \eqref{cfzz}.

	\STATE Compute $\mathbf{C}^{(k+1)}$ according to \eqref{uu}.

	\STATE Set $k=k+1$.
	 \UNTIL $k=k_{max}$ or $\mathbf{X}^{(k)}$ satisfies some stopping criterion.
    \end{algorithmic}
\end{algorithm}
\begin{Thm}[Optimal Solutions]
The optimal solutions of the problems in \eqref{xx} and \eqref{zz} are given by:
\begin{align}
\label{cfxx}
\mathbf{X}^{(k+1)}_{n,:} =& \frac{\max \left\{1- \frac{\lambda}{\rho\|\mathbf{t}^{(k)}\|_2}, 0\right\} \mathbf{t}^{(k)}}{ \mathbf{A}^H_{:,n}\mathbf{A}_{:,n}}\\
\label{cfzz}
\mathbf{\overline{B}}^{(k+1)} =& \frac{1}{N+\rho} \left(\mathbf{Y}+\rho \mathbf{\overline{AX}}^{(k+1)}+\rho \mathbf{C}^{(k)}\right)
\end{align}
where $\mathbf{t}^{(k)} = \mathbf{A}^H_{:,n}\left(\mathbf{A}_{:,n}\mathbf{X}^{(k)}_{n,:} + \mathbf{\overline{B}}^{(k)}-\mathbf{\overline{AX}}^{(k)}-\mathbf{C}^{(k)}\right)$.
\end{Thm}

% Wang Qi wrote ============================================================
\begin{proof}
	First, by setting the gradient of the objective function of the problem in \eqref{xx} to be zero, we have $\lambda \frac{\mathbf{X}_{n,:}}{||\mathbf{X}_{n,:}||_2}+\rho \mathbf{A}_{:,n}^H\Big(\mathbf{A}_{:,n}\mathbf{X}_{n,:}-\mathbf{A}_{:,n}\mathbf{X}_{n,:}^{(k)}-\mathbf{\overline{B}}^{(k)}+\mathbf{\overline{AX}}^{(k)}+\mathbf{\overline{C}}^{(k)}\Big)	=0 $,  implying
%	\begin{align}
%	\lambda \frac{\mathbf{X}_{n,:}}{||\mathbf{X}_{n,:}||_2}+\rho \mathbf{A}_{:,n}^H&\Big(\mathbf{A}_{:,n}\mathbf{X}_{n,:}-\mathbf{A}_{:,n}\mathbf{X}_{n,:}^{(k)}-\mathbf{\overline{B}}^{(k)}+\mathbf{\overline{AX}}^{(k)}\nonumber\\
%	&+\mathbf{\overline{C}}^{(k)}\Big)	=0 \nonumber
%	\end{align}
$\mathbf{X}_{n,:}+\frac{\lambda}{\rho \mathbf{A}_{:,n}^H\mathbf{A}_{:,n}}\frac{\mathbf{X}_{n,:}}{||\mathbf{X}_{n,:}||_2}=-\rho\mathbf{A}_{:,n}^H -\mathbf{A}_{:,n}\mathbf{X}_{n,:}^{(k)}-\mathbf{\overline{B}}^{(k)}
		+\mathbf{\overline{AX}}^{(k)}+\mathbf{\overline{C}}^{(k)}$.
%	\begin{align}
%		\mathbf{X}_{n,:}+\frac{\lambda}{\rho \mathbf{A}_{:,n}^H\mathbf{A}_{:,n}}\frac{\mathbf{X}_{n,:}}{||\mathbf{X}_{n,:}||_2}=&-\rho\mathbf{A}_{:,n}^H -\mathbf{A}_{:,n}\mathbf{X}_{n,:}^{(k)}-\mathbf{\overline{B}}^{(k)}\nonumber\\
%		&+\mathbf{\overline{AX}}^{(k)}+\mathbf{\overline{C}}^{(k)}.\nonumber
%	\end{align}
	By the equations in (3)-(5) in [22], we can obtain \eqref{cfxx}.
	Next, by setting the gradient of the objective function of the problem in \eqref{zz} to be zero, we have $N(N\mathbf{\overline{B}}-\mathbf{Y})+N\rho (\mathbf{\overline{B}}-\mathbf{\overline{AX}}^{(k+1)}-\mathbf{C}^{(k)})=0$.
%	\begin{align}
%	N(N\mathbf{\overline{B}}-\mathbf{Y})+N\rho (\mathbf{\overline{B}}-\mathbf{\overline{AX}}^{(k+1)}-\mathbf{C}^{(k)})=0\nonumber 	
%	\end{align}
	Thus, we can obtain \eqref{cfzz}. Hence, we complete the proof.
\end{proof}

% Wang Qi wrote ============================================================

\lscr{The details of the ADMM algorithm are summarized in Algorithm~1. By Appendix A in \cite{boyd2011distributed}, we know that $\mathbf{X}^{(k)}$ converges to an optimal solution of the problem in \eqref{GL}, as $k\to \infty$. \wqm{Algorithm~1 has two parameters, i.e., $\lambda>0$ and $\rho>0$, which influence the recovery accuracy and convergence speed, respectively.} The computational complexity is $\mathcal{O}(LNM)$. As $\mathbf{X}^{(k)}_{n,:} \in \mathbb{C}^{1\times M} ,n\in \mathcal{N}$ can be computed in parallel and the size of $\mathbf{\overline{B}}^{(k)} \in \mathbb{C}^{L\times M}$ is \lscb{usually not large} (due to $L\ll N$), the computation time of Algorithm~1 can be greatly reduced, }\wqm{compared to the block coordinate descent algorithm in \cite{qin2013efficient,LS2019}.}
\wq{To our knowledge, Algorithm~1 is so far the most efficient algorithm for GROUP LASSO.}

\wqm{Finally, we introduce the approximation part of the GROUP LASSO-based decoder, which is to approximate the estimation obtained by Algorithm~1.}
\lscr{Note that the operations for complex numbers in Algorithm~1 can be easily expressed in terms of $\Re(\mathbf{A})$, $\Im(\mathbf{A})$, $\Re(\mathbf{Y})$, $\Im(\mathbf{Y})$, $\Re(\mathbf{\overline{B}})$, $\Im(\mathbf{\overline{B}})$, $\Re(\mathbf{C})$ and $\Im(\mathbf{C})$. Thus, the operations for complex numbers in Algorithm~1 are readily implemented with operations for real numbers using a standard neural network.} As illustrated in Fig.~\ref{networkmodel-sig} (a), each \lscr{building block} of the approximation part of the GROUP LASSO-based decoder realizes one iteration of Algorithm~1. We can directly input $\mathbf{Y}$ and $\mathbf{A}$ to the approximation part to obtain $\mathbf{X}^{(U)}$, which is the estimate of $\mathbf{X}$ at the $U$-th iteration of Algorithm~1, as the output. \wqm{Note that $\lambda>0$ and $\rho>0$ are tunable parameters.}

\begin{figure*}[htp]
\begin{align}
\label{computeX}
\mathbf{X}_{n,:}^{(k+1)} &= \left(\frac{\frac{1}{(\tau^{(k)})^2+1}\left((\mathbf{R}^{(k)})^H\mathbf{A}_{:,n}+(\mathbf{X}_{n,:}^{(k)})^H\right)}{1+\frac{1-\epsilon(n)}{\epsilon(n)}\lsccc{\left(\frac{(\tau^{(k)})^2+1}{(\tau^{(k)})^2}\right)}^M\exp\left(\frac{-\|(\mathbf{R}^{(k)})^H\mathbf{A}_{:,n}+(\mathbf{X}_{n,:}^{(k)})^H\|_2^2}{(\tau^{(k)})^2((\tau^{(k)})^2+1)}\right)}\right)^H&\\
\label{computeR}
\mathbf{R}^{(k+1)} &= \mathbf{Y}-\mathbf{A}\mathbf{X}^{(k+1)}+\frac{N}{L}\mathbf{R}^{(k)}\sum_{n\in\mathcal{N}}\Bigg(\frac{\frac{1}{\lscca{(\tau^{(k)})^2+1}}\mathbf{I}_{M\times M}}{1+\frac{1-\epsilon(n)}{\epsilon(n)}\lsccc{\left(\frac{(\tau^{(k)})^2+1}{(\tau^{(k)})^2}\right)}^M\exp\left(\frac{-\|(\mathbf{R}^{(k)})^H\mathbf{A}_{:,n}+(\mathbf{X}_{n,:}^{(k)})^H\|_2^2}{(\tau^{(k)})^2((\tau^{(k)})^2+1)}\right)}\nonumber\\
&\quad\quad\quad\quad\quad\quad\quad\quad\quad\quad\quad\quad\quad\quad+\frac{(t(n))^{(k)}\|(\mathbf{R}^{(k)})^H\mathbf{A}_{:,n}+(\mathbf{X}_{n,:}^{(k)})^H\|_2^2}{\lsccc{\lscca{((\tau^{(k)})^2+1)^2(\tau^{(k)})^4}(1+(t(n))^{(k)})^2}}\Bigg)
\end{align}
\hrulefill
\end{figure*}
\subsubsection{AMP-based Decoder}
\begin{algorithm}[tp]
    \caption{Generalization of AMP \cite{8323218}}
    \begin{algorithmic}[1]
        \STATE Set $\mathbf{X}^{(0)}=\mathbf{0}_{M\times N}$, $\mathbf{R}^{(0)}=\mathbf{Y}$ and $k=0$.
        \REPEAT
        \STATE For all $n\in\mathcal{N}$, compute $\mathbf{X}_{n,:}^{(k+1)}$ according to \eqref{computeX}, \lscb{as shown at the top of the next page,} where $\tau^{(k)} = \sqrt{\frac{1}{ML}}\||\mathbf{R}^{(k)}\||_F$.
        \STATE Compute $\mathbf{R}^{(k+1)}$ according to \eqref{computeR}, \lscb{as shown at the top of the next page,} where $(t(n))^{(k)}=\frac{1-\epsilon(n)}{\epsilon(n)}\lsccc{\left(\frac{(\tau^{(k)})^2+1}{(\tau^{(k)})^2}\right)}^M\exp\left(\frac{-\|(\mathbf{R}^{(k)})^H\mathbf{A}_{:,n}+(\mathbf{X}_{n,:}^{(k)})^H\|_2^2}{(\tau^{(k)})^2((\tau^{(k)})^2+1)}\right)$.
        \STATE Set $k=k+1$.
        \UNTIL $k=k_{max}$ or $\mathbf{X}^{(k)}$ satisfies some stopping criterion.
    \end{algorithmic}
\end{algorithm}
In this part, we \lscb{present} the AMP-based decoder. First, we introduce the state-of-the-art AMP algorithm with the MMSE denoiser \cite{8323218}, which achieves \lscr{excellent recovery accuracy and short computation time for large $N$, $M$ and $L$}. Note that the AMP algorithm in \cite{8323218} is for the scenario where $\alpha(n)$, $n \in \mathcal{N}$ are \lscb{independently and identically distributed (i.i.d.) } Bernoulli random variables, each with probability $\epsilon\in(0,1)$ being 1,  \lscr{non-zero elements of $\mathbf{X}$ are i.i.d. complex Gaussian random variables with zero mean and unit variance}, and $\epsilon$ is assumed to be known. We slightly extend it to a more general scenario where $\alpha(n)$, $n\in\mathcal{N}$ are independent Bernoulli random variables and the probability of $\alpha(n)$ being 1 is $\epsilon(n) \in (0,1)$. In particular, we generalize the AMP algorithm in \cite{8323218} by replacing $\epsilon$ in the update for the estimate \lscb{of} $\mathbf{X}_{n,:}$ with $\epsilon(n)$. The details of the generalized version of the AMP algorithm \cite{8323218} are summarized in Algorithm~2. In Algorithm~2, $\mathbf{X}^{(k)}$ represents the estimation of $\mathbf{X}$ at the $k$-th iteration and $\mathbf{R}^{(k)}$ represents the corresponding residual. \lscr{Algorithm~2 has \wq{$N$} parameters, \lscb{i.e., } $\epsilon(n)$, $n \in \mathcal{N}$, which will influence the recovery
accuracy. The computational complexity of Algorithm~2 is $\mathcal{O}(LNM)$.}

%In addition, $\delta$ in Step 6 of Algorithm 2 is also treated as a trainable parameter in our proposed AMP-based decoder.

Next, we introduce the approximation part of the AMP-based decoder, which is to approximate the estimate obtained by Algorithm~2. Note that the operations for complex numbers in Step 4 of Algorithm~2 can be easily expressed in terms of \lscb{$\Re(\mathbf{R}^{(k)})$}, $\Re(\mathbf{A}_{:,n})$, \lscb{$\Re(\mathbf{X}_{n,:}^{(k)})$}, \lscb{$\Im(\mathbf{R}^{(k)})$}, $\Im(\mathbf{A}_{:,n})$ and \lscb{$\Re(\mathbf{X}_{n,:}^{(k)})$}. Thus, the operations for complex numbers in  Algorithm~2 are readily implemented with operations for real numbers using a standard neural network. As illustrated in Fig.~\ref{networkmodel-sig} (b), each \lscr{building block} of the approximation part of the AMP-based decoder realizes one iteration of Algorithm~2. We can directly input $\mathbf{Y}$ and $\mathbf{A}$ to the approximation part and obtain $\mathbf{X}^{(U)}$, which is the estimate of $\mathbf{X}$ at the $U$-th iteration of Algorithm~2, as the output. \lscr{Note that $\epsilon(n)$, $n \in \mathcal{N}$  are tunable parameters.}

\subsection{Training Process for Jointly Sparse Signal Recovery}
We introduce the training procedure for our proposed model-driven approach for jointly sparse signal recovery. Consider $I$ training samples $\mathbf{X}^{[i]},i=1,\cdots,I$. Let $\hat{\mathbf{X}}^{[i]}$ denote the output of the auto-encoder corresponding to input $\mathbf{X}^{[i]}$. To measure the difference between $\hat{\mathbf{X}}^{[i]}$ and $\mathbf{X}^{[i]}$, we adopt the mean squared error (MSE) loss function:
\vspace*{-0.2cm}
\begin{align*}&{\rm Loss}\left((\mathbf{X}^{[i]})_{i=1,\cdots,I},(\hat{\mathbf{X}}^{[i]})_{i=1,\cdots,I}\right)\\
\nonumber
&=\frac{1}{MNI}\sum_{i=1}^{I}\|\mathbf{X}^{[i]}-\hat{\mathbf{X}}^{[i]}\|_F^2
\end{align*}
We train the auto-encoder using the \lscb{adaptive moment estimation (ADAM)} algorithm \cite{kingma2015adam}.\footnote{\lscr{ADAM is an algorithm for first-order gradient-based optimization of stochastic objective functions. It is well suited for problems that are large in terms of data and/or parameters due to its high computation efficiency and little memory requirement.}}

\begin{figure*}[t]
\begin{center}
 \subfigure[\scriptsize{Proposed approach with covariance-based decoder ($U=0$).}]
 {\resizebox{17.4cm}{!}{\includegraphics{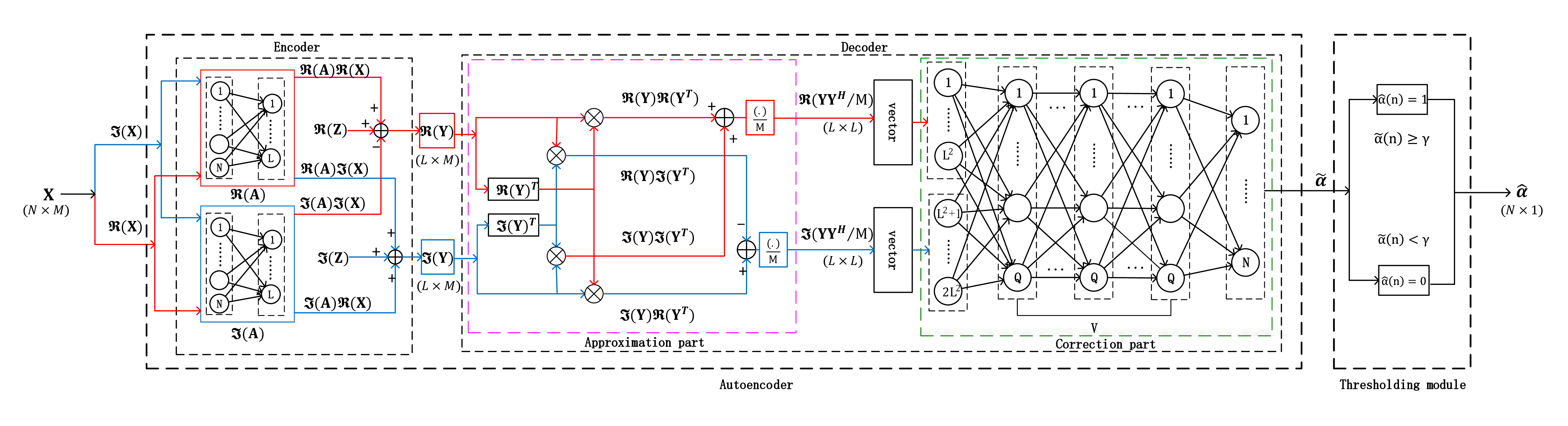}}}
 \subfigure[\scriptsize{Proposed approach with MAP-based decoder. \lscr{The operations for generating $\Re(\mathbf{YY}^H/M)$ and $\Im(\mathbf{YY}^H/M)$ are the same as those in Fig.~\ref{networkmodel-sup}(a) and are omitted here. The \lsccc{approximation} part implements $U$ iterations of Algorithm~3.}}]
 {\resizebox{17.4cm}{!}{\includegraphics{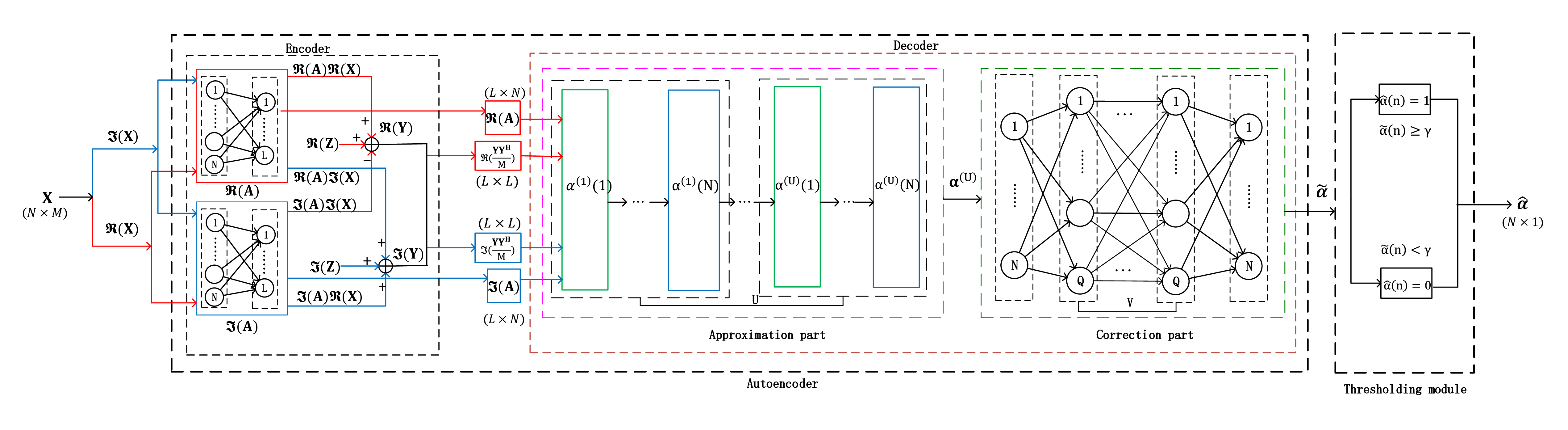}}}
  \subfigure[\scriptsize{Proposed approach with AMP-based decoder. The auto-encoder is the same as that in Fig.~\ref{networkmodel-sig} (b). \lscr{The \lsccc{approximation} part implements $U$ iterations of Algorithm~2.}}]
 {\resizebox{17.4cm}{!}{\includegraphics{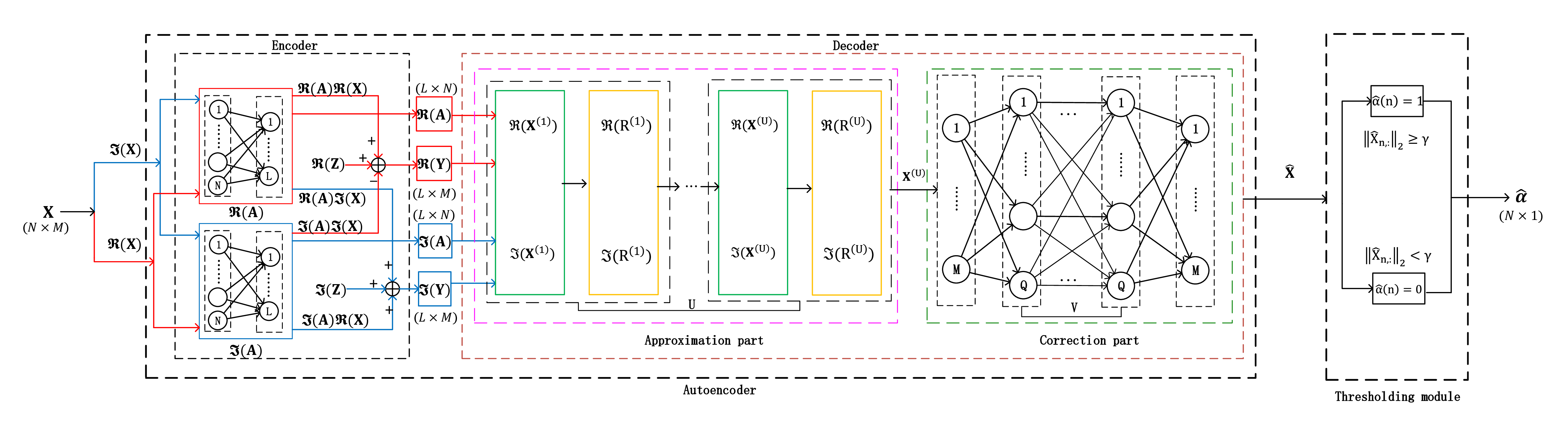}}}
  \end{center}
  \vspace*{-0.3cm}
     \caption{Proposed model-driven approach for jointly sparse support recovery.}
     \vspace*{-0.3cm}
\label{networkmodel-sup}
\end{figure*}
%\clearpage
\section{A Model-driven Approach for Jointly Sparse Support Recovery}
In this section, we propose a model-driven approach, based on the standard auto-encoder structure for real numbers in deep learning, to jointly  design the common measurement matrix and the jointly sparse support recovery method for sparse complex signals. As shown in Fig.~\ref{networkmodel-sup}, the model-driven approach consists of an encoder that mimics the noisy linear measurement process, a model-driven decoder that mimics the jointly sparse support recovery process of a particular method, and a thresholding module that generates binary approximations to obtain the common support. The encoder is the same as that in Section~\ref{signal}. Thus, in \lscr{the following,} we focus on elaborating the model-driven decoder for jointly sparse support recovery and the thresholding module. \lscb{We} propose two instances for the model-driven decoder for jointly sparse support recovery, namely the covariance-based decoder and the MAP-based decoder. As shown in Fig.~\ref{networkmodel-sup}(c), the AMP-based decoder in Section~\ref{signal} can also be viewed as an instance for the model-driven decoder for jointly sparse support recovery when followed by the thresholding module.\footnote{When used for jointly sparse support recovery, the training process for the encoder and the AMP-based decoder remains the same as that in Section~\ref{approachsig}. C for jointly sparse signal recovery. The only difference is that a thresholding module is used for producing the common support.} Analogously, \lscb{an encoder that is jointly trained with a model-driven decoder provides} the respective common measurement \lscb{matrix}. After training, we obtain the design of the common measurement matrix via extracting the weights of the encoder, and directly use the model-driven decoder \lsccc{and the hard thresholding module with the optimized threshold} for jointly sparse support recovery together with the obtained common measurement matrix. Note that a common measurement matrix designed for one jointly sparse support recovery method may not be effective for another, and a common measurement matrix designed for jointly sparse signal recovery may not be effective for jointly sparse support recovery.

\subsection{Model-driven Decoder \lscb{ for Jointly Sparse Support Recovery}}
\label{supco}
\lscr{In this subsection}, we propose a model-driven decoder \lscr{that} consists of two parts, namely the approximation part and the correction part. As illustrated in Fig.~\ref{networkmodel-sup}, the approximation part which is used to approximate a particular method for jointly sparse support recovery has $U$ ($\geq 0$) \lscr{building blocks} and generates \lscr{$\boldsymbol{\alpha}^{(U)}$} as the output. \lscr{Generally, the gap between $\boldsymbol{\alpha}^{(U)}$ and the actual common support $\boldsymbol{\alpha}$ reduces with $U$.} The correction part \lscr{which  aims to reduce the difference between the obtained approximate estimation $\boldsymbol{\alpha}^{(U)}$ and the actual common support $\boldsymbol{\alpha}$} consists of $V$ ($\geq 1$) fully connected layers.  \lscc{Specially, in each of the first $V-1$ correction layers, ReLU is chosen as the activation function. The} last correction layer has $N$ neurons, and uses the sigmoid function as the activation function for producing output $\tilde{\boldsymbol{\alpha}}$. \lscr{Similarly, $V$ influences the training error and generalization error. } As in the model-driven decoder for jointly sparse signal recovery, $U$ and $V$ are jointly chosen \lscr{so that the proposed model-driven decoder can achieve higher recovery accuracy \lscr{(i.e., a \wq{smaller} gap between $\tilde{\boldsymbol{\alpha}}$ and $\boldsymbol{\alpha}$)} and shorter computation time than the underlying method.} The same rules for the choices of $U$ and $V$ also apply here.

In the following, we propose two instances for the model-driven decoder for jointly sparse support recovery, \lscr{based} on \lscb{the} covariance-based estimation method \lscb{in \cite{6994860}} and \lscb{the} coordinate descend algorithm for the MAP estimation \lscb{in \cite{DD2019}}, respectively. They are referred as the covariance-based decoder and the MAP-based decoder, as illustrated in Fig.~\ref{networkmodel-sup} (a) and Fig.~\ref{networkmodel-sup} (b), respectively.

%Thus, we propose a model-driven decoder which can elegantly utilize the feature of common support to effectively improve the performance for jointly sparse support recovery. In the following, we propose two instances of the model-driven decoder, basing on the estimation of \textcolor{blue}{the} covariance matrix of noisy linear measurements and \textcolor{blue}{the coordinate descend algorithm for MAP estimation}, respectively.

\subsubsection{Covariance-based Decoder}
In this part, we \lscb{present} the covariance-based decoder, which is motivated by the covariance-based estimation method for jointly sparse support recovery \cite{6994860}, but successfully remedies its defect.

First, we introduce the covariance-based method in \cite{6994860}. Specially, by (1), we have:
\begin{align}
\mathbf{YY}^H/M = (\mathbf{A}\mathbf{XX}^H\mathbf{A}^H+\mathbf{A}\mathbf{X}\mathbf{Z}^H+\mathbf{Z}\mathbf{X}^H\mathbf{A}^H+\mathbf{ZZ}^H)/M\nonumber
\end{align}
which can be equivalently expressed as:
\begin{align}
\label{cov}
{\rm vec}(\mathbf{YY}^H/M) = \mathbf{A}^*\odot\mathbf{A}\mathbf{r}+{\rm vec}(\mathbf{E_1})+{\rm vec}(\mathbf{E_2})
\end{align}
\lsccc{Here,} ${\rm vec}(\mathbf{YY}^H/M) \in\mathbb{R}^{L^2}$, $\mathbf{r}\in\mathbb{R}^N$, $\mathbf{E}_1\in \mathbb{C}^{L\times L}$ and $\mathbf{E}_2\triangleq(\mathbf{A}\mathbf{X}\mathbf{Z}^H+\mathbf{Z}\mathbf{X}^H\mathbf{A}^H+\mathbf{ZZ}^H)/M$, where
\begin{align*}
	r(n) &= \|\mathbf{X}_{n,:}\|^2_2/M\nonumber\\
	E_1(k,l)&\triangleq\sum_{i,j\in \mathcal{N},i\neq j}A(k,i)A^*(l,j)\sum_{m\in \mathcal{M}}x_m(i)x_m^*(j)\nonumber
\end{align*}
For any given $\mathbf{A}$, if the non-zero elements of $\mathbf{X}$ are i.i.d. random variables with zero mean, then $\mathbf{y}_m,m\in \mathcal{M}$ are i.i.d. random vectors and $\mathbf{YY}^H/M\to{\rm Cov}(\mathbf{y}_m)$, $\mathbf{E}_1\to\mathbf{0}_{L\times L}$ and $\mathbf{E}_2\to\sigma^2\mathbf{I}_{L\times L}$ as $M\to\infty$. Thus, when the non-zero elements of $\mathbf{X}$ are i.i.d. random variables with zero mean and $M\to\infty$, \lscb{\eqref{cov}} provides \lsccc{noiseless linear} measurements of $\mathbf{r}$ with ${\rm supp}(\mathbf{r})={\rm supp}(\mathbf{x}_m),m\in \mathcal{M}$, and hence can be used for jointly sparse support recovery for $\mathbf{X}$. Based on \lscb{\eqref{cov}}, the authors in \cite{6994860} use LASSO for SMV models to achieve jointly sparse support recovery for MMV models in the case of very large $M$. In Section~\ref{sim}, we shall see that the covariance-based method via LASSO in \cite{6994860} does not work well for small $M$ (as $\mathbf{E}_1$ is nonnegligible and $\mathbf{E}_2$ is non-diagonal at small $M$), while the proposed approach with the covariance-based decoder can perfectly resolve this issue.

Next, we introduce the covariance-based decoder. \lsccc{As}
\vspace*{-0.2cm}
\begin{align}
\Re(\mathbf{YY}^H)/M=(\Re(\mathbf{Y})\Re(\mathbf{Y}^T)+\Im(\mathbf{Y})\Im(\mathbf{Y}^T))/M\\
\Im(\mathbf{YY}^H)/M=(\Im(\mathbf{Y})\Re(\mathbf{Y}^T)-\Re(\mathbf{Y})\Im(\mathbf{Y}^T))/M
\end{align}
we can obtain \lscr{$\Re(\mathbf{YY}^H)/M \in \mathbb{R}^{L\times L}$} and \lscr{$\Im(\mathbf{YY}^H)/M\in \mathbb{R}^{L\times L}$} based on the output of the encoder $\Im(\mathbf{Y})$ and $\Re(\mathbf{Y})$\lscb{,} and use them as the input of the covariance-based decoder. Note that the input has already captured the concept of the covariance-based method in \cite{6994860}, \lscc{and the recovery accuracy increases with $M$ (which can be seen from the \lscv{above} asymptotic analysis)}. Thus, we no longer need the approximation part, i.e., we set $U=0$\lscr{,} and rely \wq{only} on the correction part to perform the estimation based on the empirical covariance matrix $\mathbf{YY}^H/M$. The \lscc{first correction layer} has $2L^2$ neurons with ${\rm vec}(\Re(\mathbf{YY}^H)/M)$ as the input of the first $L^2$ neurons and ${\rm vec}(\Im(\mathbf{YY}^H)/M)$ as the input of the last $L^2$ neurons.

%In each of the $V-2$ hidden layers, there are $Q$ neurons where $Q$ is properly chosen according to the problem size. The output layer has $N$ neurons and the Sigmoid function is chosen as the activation function.

\subsubsection{MAP-based Decoder}
\label{map-based}
%Similar to the model-driven decoder in Section~\ref{signal}, the MAP-based decoder consists of two parts. The first part aims to approximate the coordinate descend algorithm for MAP estimation and can be implemented using neural network, and is thus called the approximation part. The second part aims to reduce the difference between the approximate solution and the actual jointly sparse signals, and is thus referred to as the correction part.

In this part, we \lscb{present} the MAP-based decoder. \lscr{First, we introduce} a coordinate descent algorithm for the MAP estimation which incorporates prior marginal distributions in sparse patterns \cite{DD2019}. \lscb{In \cite{DD2019}}, $\alpha(n)$, $n\in\mathcal{N}$ are independent Bernoulli random variables and the probability of $\alpha(n)$ being 1 is $\lscn{\epsilon(n)} \in (0,1)$; \lscb{$\epsilon(n)$, $n\in\mathcal{N}$ are known}; \lscv{and} non-zero elements of $\mathbf{X}$ are i.i.d. complex Gaussian random variables \lscr{with zero mean and unit variance}. Then, the prior distribution of $\boldsymbol{\alpha}$ is given by $p(\boldsymbol{\alpha})  \triangleq \prod_{n=1}^{N}\lscn{\epsilon(n)}^{\alpha(n)}(1-\lscn{\epsilon(n)})^{1-\alpha(n)}$. By multiplying $p(\boldsymbol{\alpha})$ with the conditional density of $\mathbf{Y}$ given $\boldsymbol{\alpha}$ \cite{DD2019}, we can obtain the conditional density of $\boldsymbol{\alpha}$ given $\mathbf{Y}$, and express the negative log function of it as:
\begin{align}
\label{map1}
&f_{MAP}(\boldsymbol{\alpha}) \nonumber\\
&\quad \quad = \log|\mathbf{A}\boldsymbol{\Gamma}\mathbf{A}+\sigma^2\mathbf{I}_{L\times L}|+{\rm tr}((\mathbf{A}\boldsymbol{\Gamma}\mathbf{A}+\sigma^2\mathbf{I}_{L\times L})^{-1}\hat{\boldsymbol{\Sigma}}) \nonumber\\
&\quad \quad \quad -\frac{1}{M}\sum_{n\in\mathcal{N}}(\alpha(n)\log \lscn{\epsilon(n)}+(1-\alpha(n))\log(1-\lscn{\epsilon(n)}))
\end{align}
where $\boldsymbol{\Gamma}$ is a diagonal matrix with diagonal elements $\Gamma(n,n)=\alpha(n),n\in \mathcal{N}$, \lscb{$\hat{\boldsymbol{\Sigma}} = \mathbf{YY}^H/M$}, and an additive constant is omitted in \eqref{map1} for simplicity (without loss of optimality). Note that the first \lscb{two terms} in $f_{MAP}(\boldsymbol{\alpha})$ corresponds to the conditional density of $\mathbf{Y}$, and the \lscb{last} term in $f_{MAP}(\boldsymbol{\alpha})$ corresponds to the prior distribution of $\boldsymbol{\alpha}$ \cite{DD2019}. Thus, we formulate the MAP estimation problem for the common support recovery as follows:
\begin{align}
\label{map}
\min_{\boldsymbol{\alpha}\succeq 0} f_{MAP}(\boldsymbol{\alpha})
\end{align}

\begin{figure*}[htp]
\begin{align}
\label{dd}
&(d(n))^{(k)}= {\rm max}\lsccc{\left\{\frac{\left(\mathbf{A}_{:,n}^H(\boldsymbol{\Sigma}^{-1})^{(k)}\mathbf{A}_{:,n}\right)^2-
\frac{2}{M}\log\left(\frac{\epsilon(n)}{1-\epsilon(n)}\right)\mathbf{A}_{:,n}^H(\boldsymbol{\Sigma}^{-1})^{(k)}\mathbf{A}_{:,n}
-\sqrt{(\Delta(n))^{(k)}}}{\frac{2}{M}\log\left(\frac{\epsilon(n)}{1-\epsilon(n)}\right)(\mathbf{A}_{:,n}^H(\boldsymbol{\Sigma}^{-1})^{(k)}\mathbf{A}_{:,n})^2}, -(\alpha(n))^{(k)}\right\}}\\
%\end{align}
%\end{figure*}
%\begin{figure*}[htp]
%\begin{align}
\label{de}
&\lsccc{(\Delta(n))^{(k)}}=\left(\left(\mathbf{A}_{:,n}^H(\boldsymbol{\Sigma}^{-1})^{(k)}\mathbf{A}_{:,n}\right)^2+\frac{2}{M}\log\left(\frac{\epsilon(n)}{1-\epsilon(n)}\right)\mathbf{A}_{:,n}^H(\boldsymbol{\Sigma}^{-1})^{(k)}\mathbf{A}_{:,n}\right)^2\nonumber\\
&+\frac{4}{M}\log\left(\frac{\epsilon(n)}{1-\epsilon(n)}\right)\left(\mathbf{A}_{:,n}^H(\boldsymbol{\Sigma}^{-1})^{(k)}\mathbf{A}_{:,n}\right)^2
\left(\mathbf{A}_{:,n}^H(\boldsymbol{\Sigma}^{-1})^{(k)}\mathbf{A}_{:,n}-\mathbf{A}_{:,n}^H(\boldsymbol{\Sigma}^{-1})^{(k)}\hat{\boldsymbol{\Sigma}}(\boldsymbol{\Sigma}^{-1})^{(k)}\mathbf{A}_{:,n}-\frac{1}{M}\log\left(\frac{\epsilon(n)}{1-\epsilon(n)}\right)\right)
\end{align}
\hrulefill
\end{figure*}
Note that this problem is non-convex. By extending the coordinate descent algorithm for the MAP estimation in \cite{DD2019}, we obtain a coordinate descent algorithm for the MAP estimation in \lscb{the problem in} \eqref{map}, which is summarized in Algorithm~3. In Algorithm~3, $\boldsymbol{\alpha}^{(k)}$ represents the estimate at the $k$-th iteration. Following the proof of convergence in \cite{DD2019}, we can show that $\boldsymbol{\alpha}^{(k)}$ converges to a \lscr{locally} optimal solution of \lscr{the problem in \eqref{map}, as $k \to\infty$. The computational complexity of Algorithm~3 is $\mathcal{O}(NL^2)$.  Algorithm~3 has parameters $\lscn{\epsilon(n)}$, $n\in\mathcal{N}$, which will influence the recovery accuracy.}

\begin{algorithm}[tp]
\caption{MAP}
\begin{algorithmic}[1]
\STATE Set $(\boldsymbol{\Sigma}^{-1})^{(0)}=\frac{1}{\sigma^2}\mathbf{I}_{L\times L}, \boldsymbol{\alpha}^{(0)}=\mathbf{0}_{N\times 1}$ and $k=0$.
\REPEAT
\FOR{$n=1,\dots,N$}
\STATE Calculate \lscb{$(d(n))^{(k)}$ according to \eqref{dd}, where \lsccc{$(\Delta(n))^{(k)}$} is given by \eqref{de}}, \lscb{as shown at the top of the next page}.
\STATE Update $(\alpha(n))^{(k+1)} = (\alpha(n))^{(k)}+(d(n))^{(k)}$.
\STATE Update $(\boldsymbol{\Sigma}^{-1})^{(k)} = (\boldsymbol{\Sigma}^{-1})^{(k)}- \frac{(d(n))^{(k)}(\boldsymbol{\Sigma}^{-1})^{(k)}\mathbf{A}_{:,n}\mathbf{A}_{:,n}^H(\boldsymbol{\Sigma}^{-1})^{(k)}}{1+(d(n))^{(k)}\mathbf{A}_{:,n}(\boldsymbol{\Sigma}^{-1})^{(k)}\mathbf{A}_{:,n}^H}$.
\ENDFOR
\STATE Update $(\boldsymbol{\Sigma}^{-1})^{(k+1)} = (\boldsymbol{\Sigma}^{-1})^{(k)}$.
\STATE Set $k=k+1$.
\UNTIL $k=k_{max}$ or $\boldsymbol{\alpha}^{(k)}$ satisfies some stopping criterion.
\end{algorithmic}
\end{algorithm}

Next, we introduce \lscr{the MAP-based decoder. As in Section~\ref{map-based}, we can obtain $\Re(\mathbf{YY}^H/M)$ and $\Im(\mathbf{YY}^H/M)$. $\Re(\mathbf{YY}^H/M)$, $\Im(\mathbf{YY}^H/M)$, $\Re(\mathbf{A})$ and $\Im(\mathbf{A})$ are \wq{the} input of the MAP-based decoder.} The approximation part of the MAP-based decoder is to approximate the estimate obtained by Algorithm~3. Note that \lscr{the operations for complex numbers in Steps 4, 6 of Algorithm~3 can be easily expressed in terms of $\Re(\hat{\boldsymbol{\Sigma}})$, $\Re(\mathbf{A})$, $\Im(\hat{\boldsymbol{\Sigma}})$, $\Im(\mathbf{A})$ and $(\boldsymbol{\Sigma}^{-1})^{(k)}$}. Thus, the operations for complex numbers in Algorithm~3 are readily implemented with operations for real numbers \lscr{using a neural network.} As illustrated in Fig.~\ref{networkmodel-sup} (b), \lscr{each building block of the approximation part of the MAP-based decoder realizes one iteration of Algorithm~3}. We can directly obtain $\boldsymbol{\alpha}^{(U)}$, which is the estimation of $\boldsymbol{\alpha}$ obtained at the $U$-th iteration of Algorithm~3, as the output \lscr{of the approximation part. Note that $\lscn{\epsilon(n)}$, $n\in\mathcal{N}$ are tunable parameters. The \lscc{first correction layer} has $N$ neurons with $\boldsymbol{\alpha}^{(U)}$ as the input.}
%For all $k=0,1,\cdots,\frac{U}{N}-1$, layer $kN+1$, $\dots$, layer $(k+1)N$ realize one iteration for updating $\alpha(n), n \in \mathcal{N}$ in Algorithm 3.
\begin{figure*}[htp]
\begin{align}
\label{loss}
{\rm Loss}\left((\boldsymbol{\alpha}^{[i]})_{i=1,\cdots,I},(\tilde{\boldsymbol{\alpha}}^{[i]})_{i=1,\cdots,I}\right)=\lsccc{-\frac{1}{NI}\sum_{i=1}^{I}\sum_{n\in\mathcal{N}}\Bigg((\alpha(n))^{[i]}
 \log((\tilde{\alpha}(n))^{[i]}
+\left(1-(\alpha(n))^{[i]}\right)\log\left(1-(\tilde{\alpha}(n))^{[i]}\right)\Bigg)}
\end{align}
\hrulefill
\end{figure*}
\subsection{Training Process \lscb{for Jointly Sparse Support Recovery}}
We introduce the training procedure for the model-driven approach for jointly sparse support recovery. Consider $I$ training samples $(\mathbf{X}^{[i]}, \boldsymbol{\alpha}^{[i]}),i=1,\cdots,I$. Let $\tilde{\boldsymbol{\alpha}}^{[i]}$ denote the output of the \lsccc{auto-encoder} corresponding to input $\mathbf{X}^{[i]}$. To measure the distance between $\boldsymbol{\alpha}^{[i]}$ and $\tilde{\boldsymbol{\alpha}}^{[i]}$, we use the binary cross-entropy loss function \lscb{given by \eqref{loss}, as shown at the top
of the next page}.
We train the auto-encoder using the ADAM algorithm.

\subsection{Thresholding Module \lscb{for Jointly Sparse Support Recovery}}
Note that even after training, there is no guarantee that the proposed auto-encoder can produce an output $\tilde{\boldsymbol{\alpha}}$ that is in $\{0,1\}^N$. Thus, we require a hard thresholding module parameterized by threshold $\gamma$ to convert the output of the auto-encoder to the final output of the proposed approach \lscb{$\hat{\boldsymbol{\alpha}} \in \{0,1\}^N$}. We adopt the hard thresholding module proposed in our previous work \cite{8861085}, and present the details here for completeness. Given the input of the thresholding module $\tilde{\boldsymbol{\alpha}}\in \mathbb{R}^N$, $\hat{\alpha}(n) = \mathbb{I}[\tilde{\alpha}(n)\geq \gamma], n\in \mathcal{N}$. \lscb{Consider} \lscb{$I$} training samples \lscb{$(\mathbf{X}^{[i]},\boldsymbol{\alpha}^{[i]}),i=1,\cdots,I$}. \lscb{Let} $P_E(\gamma) \triangleq \lscb{\frac{1}{NI} \sum_{i=1}^I\|\boldsymbol{\alpha}^{[i]}-\hat{\boldsymbol{\alpha}}^{[i]} \|_1}$ denote the error rate for the given threshold $\gamma$. We choose the optimal threshold $\gamma^*=\mathop{\arg\min}_{\gamma}P_E(\gamma)$ as the threshold for the hard thresholding module.

\section{Numerical results}
\label{simulation}
In this section, we conduct numerical experiments on device activity detection and channel estimation in MIMO-based grant free random access, respectively, \lscb{which are} illustrated in Section~\ref{app}.
We generate the jointly sparse signal $\mathbf{X} \in \mathbb{C}^{N\times M}$ with $x_m(n) = \alpha(n)h_{m}(n), m\in \mathcal M, n\in \mathcal N$ according to \nwqh{$h_m(n) \sim \mathcal{CN}(0, 1), m \in \mathcal{M},n\in \mathcal N$.}
To demonstrate the ability of the proposed approaches for effectively utilizing the properties of sparse patterns, we consider three types of distributions of $\boldsymbol{\alpha}$, i.e., the independent case, the correlated case with a single active group, and the correlated case with i.i.d group activity. In the independent case, $N$ devices are divided into two groups, i.e., $\mathcal{N}_1$ and $\mathcal{N}_2$, of the same size with the devices in $\mathcal{N}_i$ accessing the channel in an i.i.d. manner with access probability $\Pr[\alpha(n)=1]=p_i$, $n \in \mathcal{N}_i$, \lscb{where $i=1,2$;} and let $p = \frac{p_1+p_2}{2}$ denote the average access probability. In the correlated case with a single active group, $N$ devices are divided into $G$ groups of the same size $N/G$; the active states of the devices within each group are the same; only one of the $G$ groups is selected to be active uniformly at random; \lscb{and} $1/G$ can be viewed as the access probability $p$. In the correlated case with i.i.d. group activity, $N$ devices are divided into $G$ groups of the same size $N/G$; the active states of the devices within each group are the same; and all groups access the channel in an i.i.d. manner with group access probability $p$. Note that the two correlated cases are particularly constructed to test performance for channel estimation and device activity detection, respectively.
\nwqc{We generate the additive white Gaussian noise  $\mathbf{Z}\in \mathbb{C}^{L\times M}$ according to  $\mathbf{Z}_{:,m} \sim\mathcal{CN}(\mathbf{0}_{L\times 1},\sigma^2\mathbf{I}_{L\times L}), m\in \mathcal{M}$, where $\sigma^2=0.1$.} We consider two choices for $N$, i.e., $N=100$ and $N=1000$.

We compare the proposed model-driven approaches with the state-of-the-art methods. Specifically, all the baseline schemes adopt the same set of pilot sequences for the $N$ devices whose entries are independently generated from $\mathcal{CN}(0,1)$.
For fair comparison, we require ${\|\mathbf{a}_n\|}_2=\sqrt{L}$ in training the architectures of the proposed approaches, as in \cite{8861085}. The proposed approaches adopt the common measurement matrices obtained from the encoders of the respective trained architectures as pilot sequences, and use the decoders of the respective trained architectures for recovery. The sizes of training samples and validation samples for training the architectures of the proposed approaches and the size of testing samples for evaluating the proposed approaches and the baseline schemes are selected as $9\times 10^3$, $1\times10^3$ and $1\times10^3$, respectively. The maximization epochs, learning rate and batch size in training the proposed architectures are set as $1\times10^5$, 0.0001 and 32, respectively. When the value of a loss function on the validation set does not change for five \lsccc{epochs}, the training process is stopped and the corresponding parameters of the auto-encoder are saved.\footnote{\nwqf{When $N=100$, the training times for NN, AMP-NN, GROUP LASSO-NN and MAP-NN are about 10 minutes, 10 minutes, 30 minutes and 30 minutes, respectively. When $N=1000$, the training time for AMP-NN is about 1 hour. Note that we use a computer node with one 16-core AMD Ryzen processor, one Nvidia RTX 2080 GPU and 32 GB of memory. \lscc{In grant-free random access for static devices, the \lscv{distributions} of sparse signals and noise do not vary much. Hence, the obtained designs are effective for a relatively long time duration, and the training cost is neglegible.} }}

\subsection{Channel Estimation}
\label{resultsig}
In this subsection, we present numerical results on channel estimation in MIMO-based grant-free random access. We evaluate the proposed model-driven approach with the GROUP LASSO-based decoder, i.e., GROUP LASSO-NN, and with the AMP-based decoder, i.e., AMP-NN, and three baseline schemes, namely GROUP LASSO \cite{qin2013efficient,LS2019}, AMP \cite{8323218} and ML-MMSE \cite{8437359}. GROUP LASSO conducts channel estimation using the \nwqa{block coordinate descent} algorithm \cite{qin2013efficient,LS2019}; AMP conducts channel estimation using the AMP algorithm with MMSE denoiser based on the known (average) access probability $p$ \cite{8323218}; ML-MMSE uses the coordinate descent algorithm for the ML estimation in \cite{8437359} to detect device activities, and then uses the standard MMSE method to estimate the channel states of the devices that have been detected to be active.
\nwqa{The underlying assumptions for the baseline schemes have been illustrated in Section~\ref{intro}}.
\nwqa{The computational complexities for \nwqa{GROUP LASSO}, AMP and ML-MMSE are $\mathcal{O}(LNM)$, $\mathcal{O}(LNM)$ and $\mathcal{O}(NL^2)$, respectively.}
Considering computational complexity, we evaluate GROUP LASSO-NN only at $N=100$, but \lscb{evaluate AMP-NN} at $N=100$ and $N=1000$.\footnote{\lscc{When $N$ is large, it takes a large number of iterations for Algorithm~1 to \lscv{converge to a \lscb{reasonable}} estimate. Thus, we do not adopt GROUP LASSO-NN at $N=1000$}.} We set \nwqd{$U=200$} and $V=3$ for the proposed GROUP LASSO-based decoder, and  set $U=50$ and $V=3$ for the proposed AMP-based decoder. Those choices are based on a large number of experiments and the tradeoff between recovery accuracy and computation time. \lscn{For \lscb{ease} of comparison, the total numbers of iterations for the block coordinate descent algorithm for GROUP LASSO, the AMP algorithm and the block coordiante descent algorithm for ML are chosen as $200$, $50$ and $55$, respectively.}\footnote{\lscb{Note that when increasing the numbers of iterations for GROUP LASSO and AMP, the corresponding computation times will increase, but the MSEs are still larger than those of GROUP LASSO-NN and AMP-NN, respectively. That is, GROUP LASSO-NN and AMP-NN can achieve smaller MSEs with shorter computation times than GROUP LASSO and AMP, respectively.}}  We numerically evaluate the MSE
$\frac{1}{NT}\sum_{t=1}^{T}\||\mathbf{X}^{[t]}-\hat{\mathbf{X}}^{[t]}\||_F^2$ and computation time (on the same server) of each scheme using the same set of $T=10^3$ testing samples.

\begin{figure}[tp]
\begin{center}
 \subfigure[\scriptsize{MSE versus $L/N$ at $p=0.1$, $M=4$, $p_1/p_2=3$.}]
 {\resizebox{4.3cm}{!}{\includegraphics{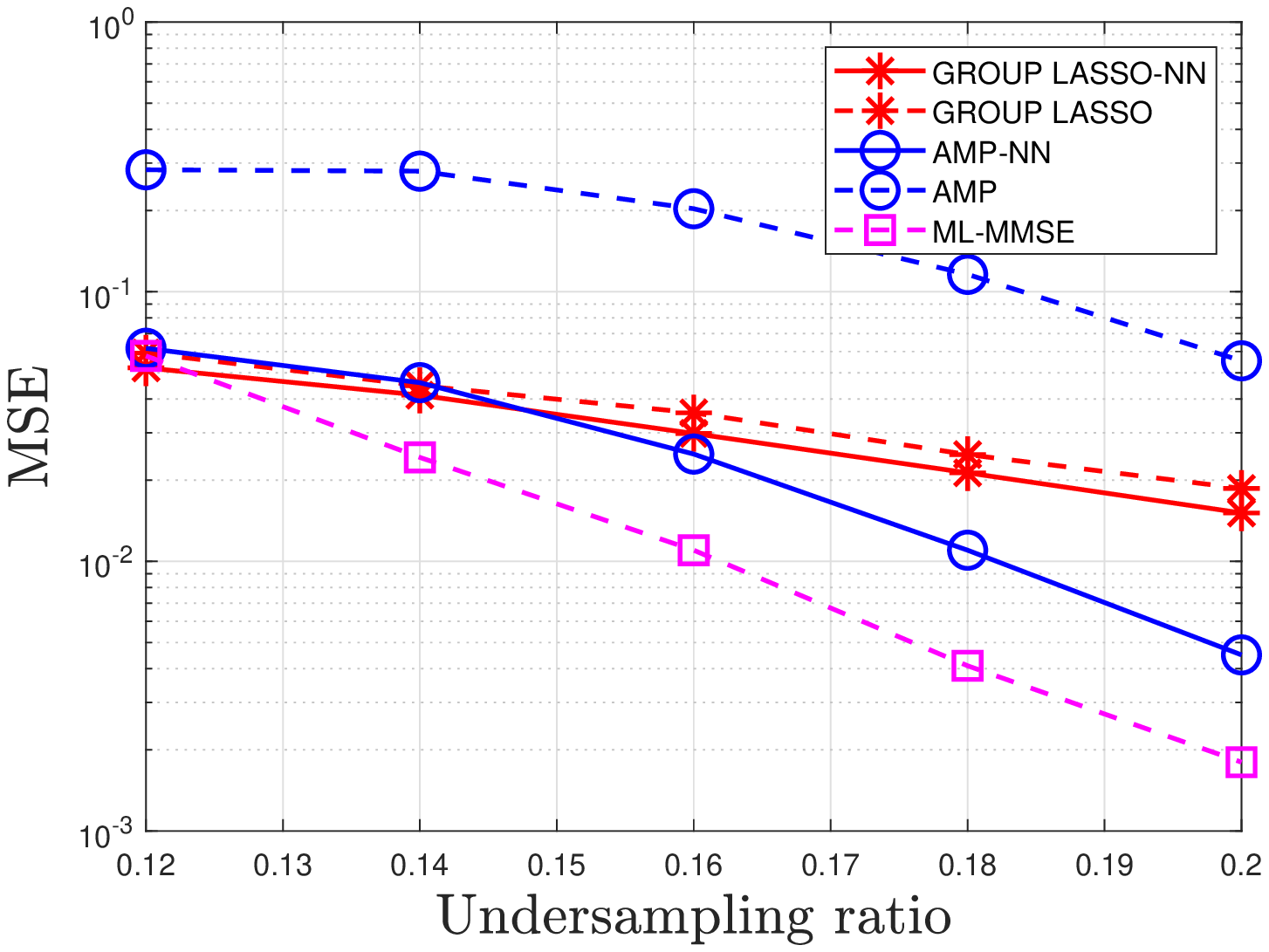}}}
 \subfigure[\scriptsize{MSE versus $p$ at $L/N=0.2$, $M=4$, $p_1/p_2=3$.}]
 {\resizebox{4.3cm}{!}{\includegraphics{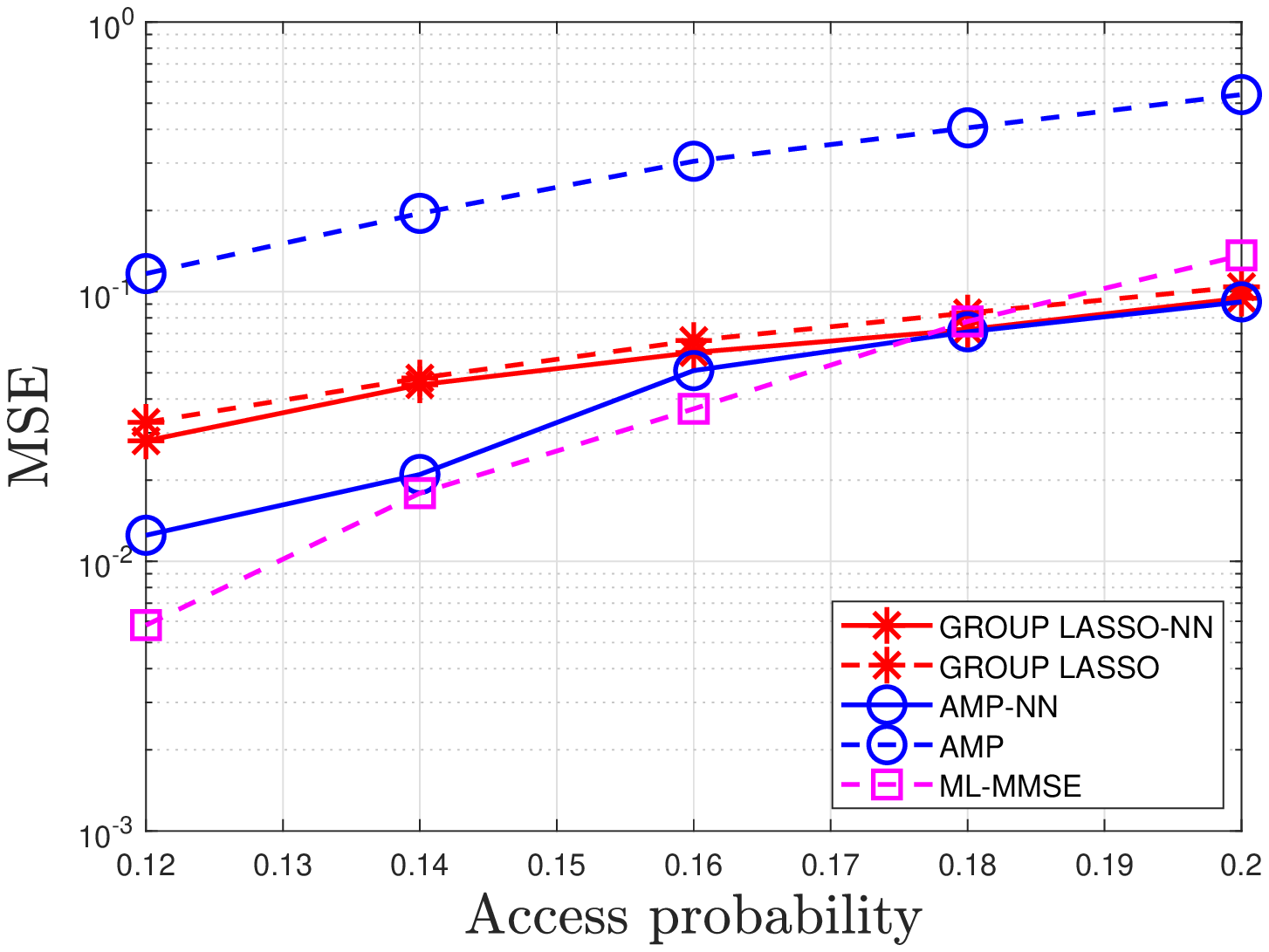}}}
 \subfigure[\scriptsize{MSE versus $M$ at $L/N=0.12$, $p=0.1$, $p_1/p_2=3$.}]
 {\resizebox{4.3cm}{!}{\includegraphics{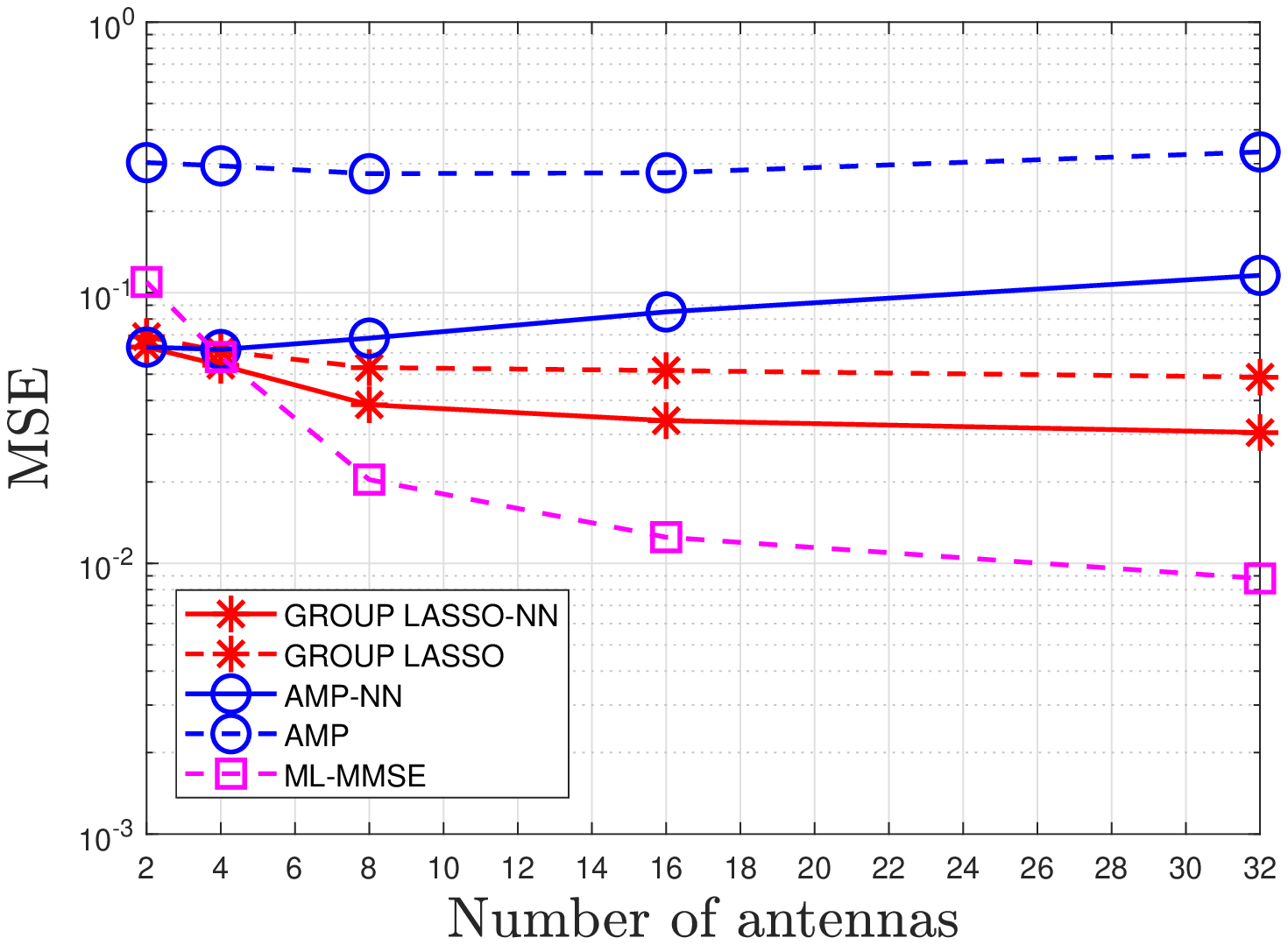}}}
 \subfigure[\scriptsize{MSE versus $p_1/p_2$ at $L/N=0.12$, $M=4$, $p=0.1$.}]
 {\resizebox{4.3cm}{!}{\includegraphics{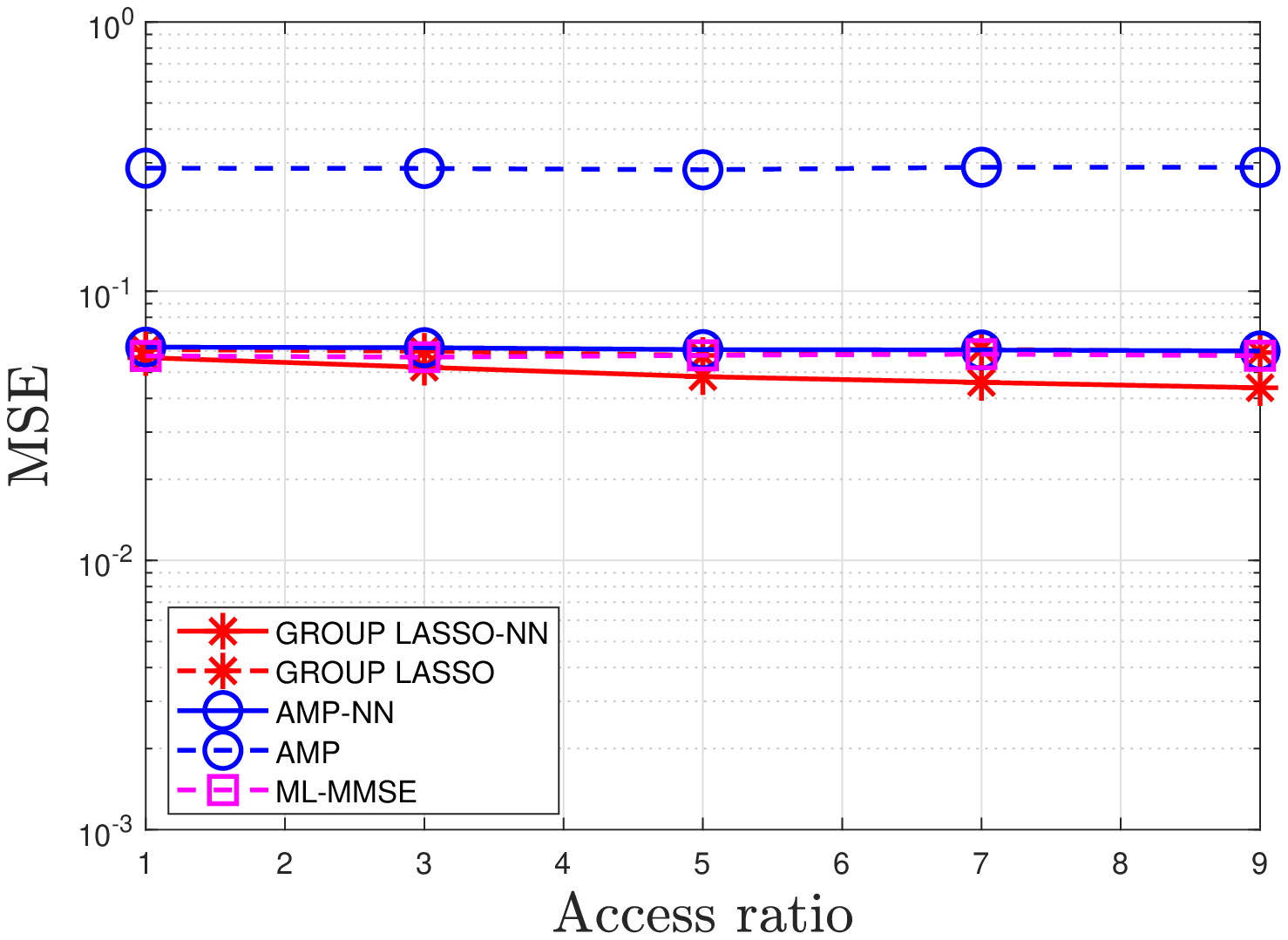}}}
 \subfigure[\scriptsize{$\epsilon_1/\epsilon_2$ versus $p_1/p_2$ at $L/N=0.12$, $M=4$, $p=0.1$.}]
 {\resizebox{4.3cm}{!}{\includegraphics{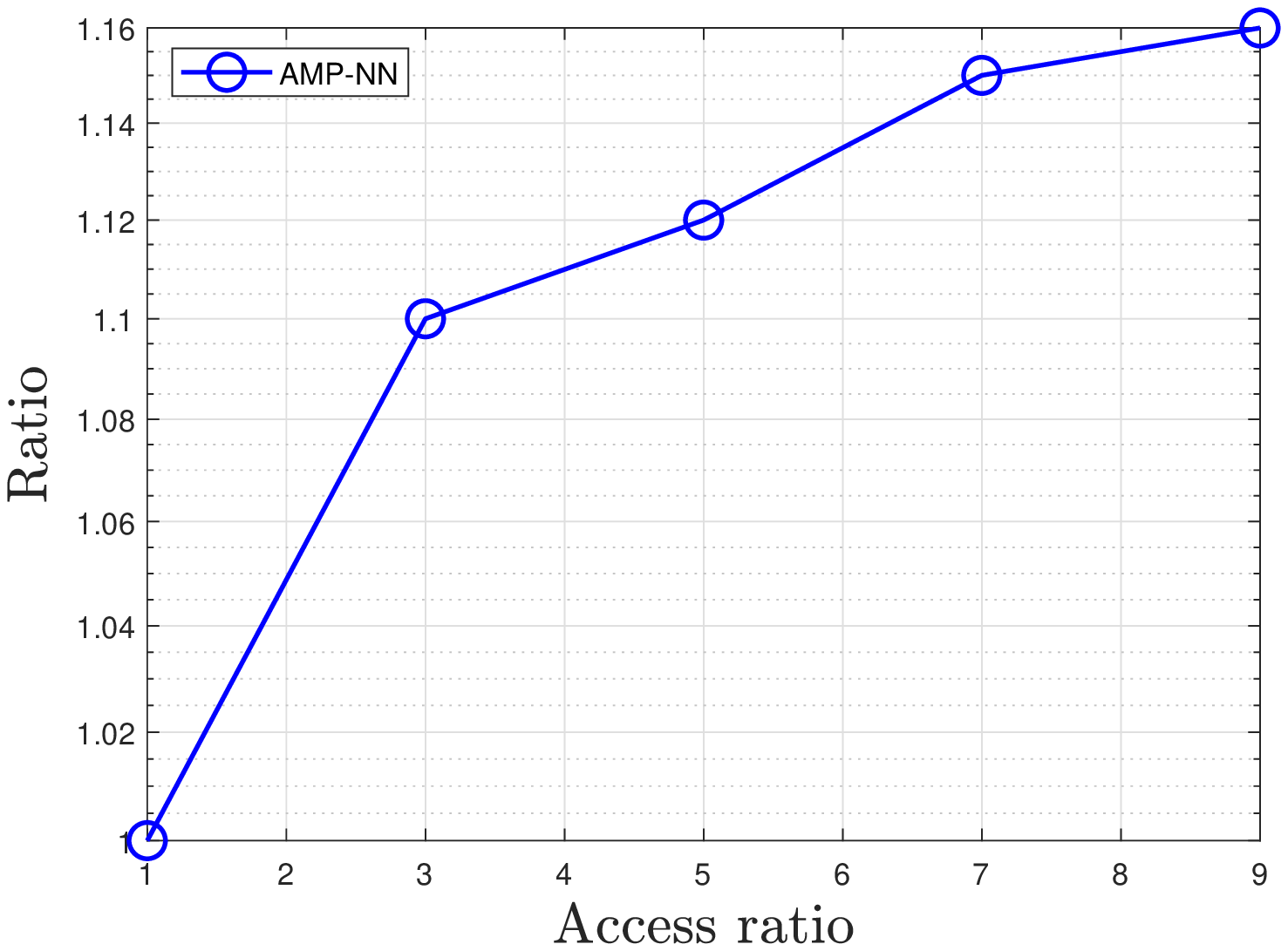}}}
  \end{center}
 \vspace*{-0.3cm}
   \caption{\small{\lsccc{Channel estimation in the} independent case at $N=100$. $\epsilon_1=\frac{2}{N}\sum_{n\in\mathcal{N}_1}\epsilon(n)$ and $\epsilon_2=\frac{2}{N}\sum_{n\in\mathcal{N}_2}\epsilon(n)$, where $\epsilon(n),n\in \mathbb{N}$, play the role of the access probabilities in Algorithm~2, and are extracted from the trained AMP-based decoder.}}
   \label{signal100}
\end{figure}

\begin{figure}[tp]
\begin{center}
 \subfigure[\scriptsize{MSE versus $L/N$ at $p=0.1$, $M=4$.}]
 {\resizebox{4.3cm}{!}{\includegraphics{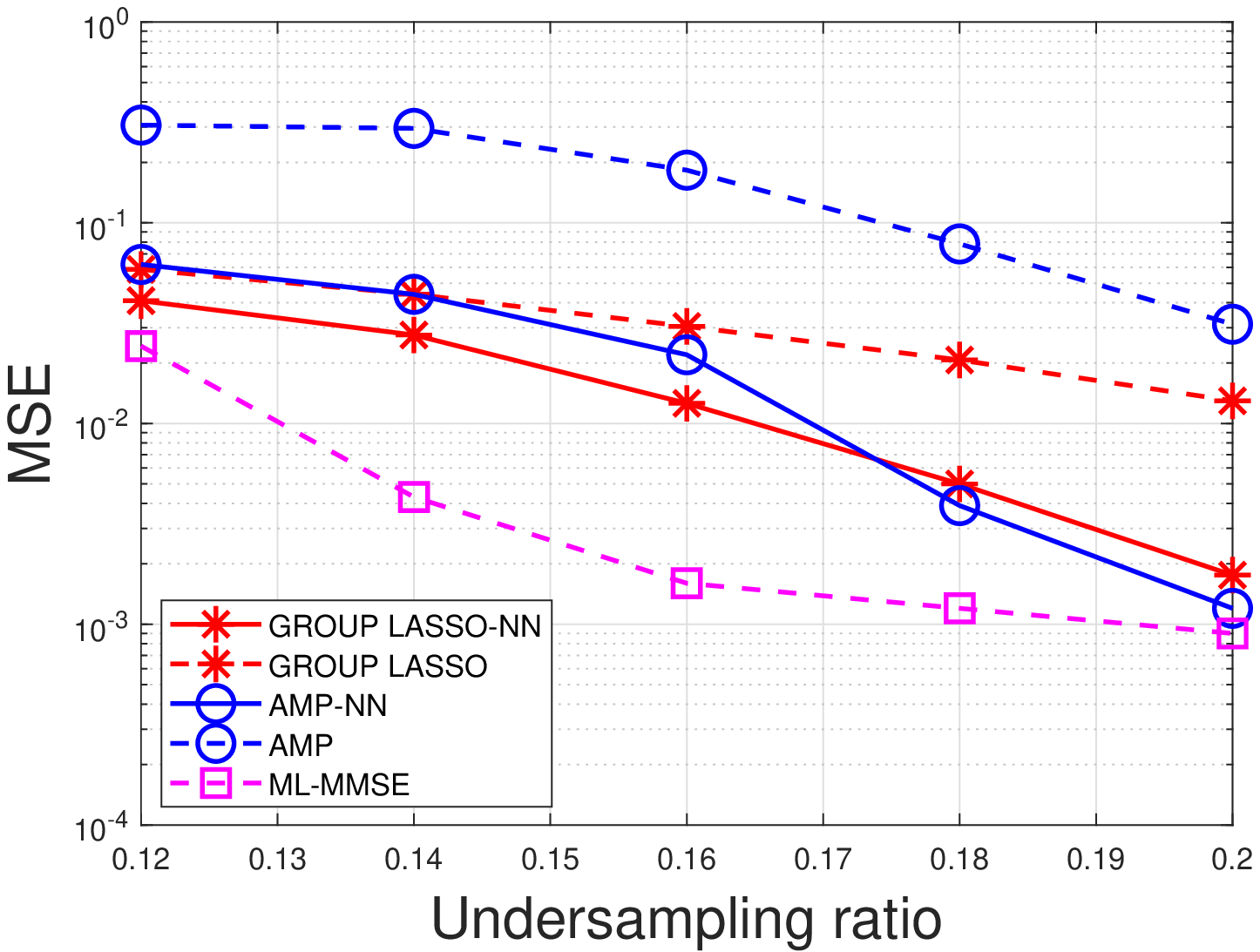}}}
 \subfigure[\scriptsize{MSE versus $p$ at $L/N=0.2$, $M=4$.}]
 {\resizebox{4.3cm}{!}{\includegraphics{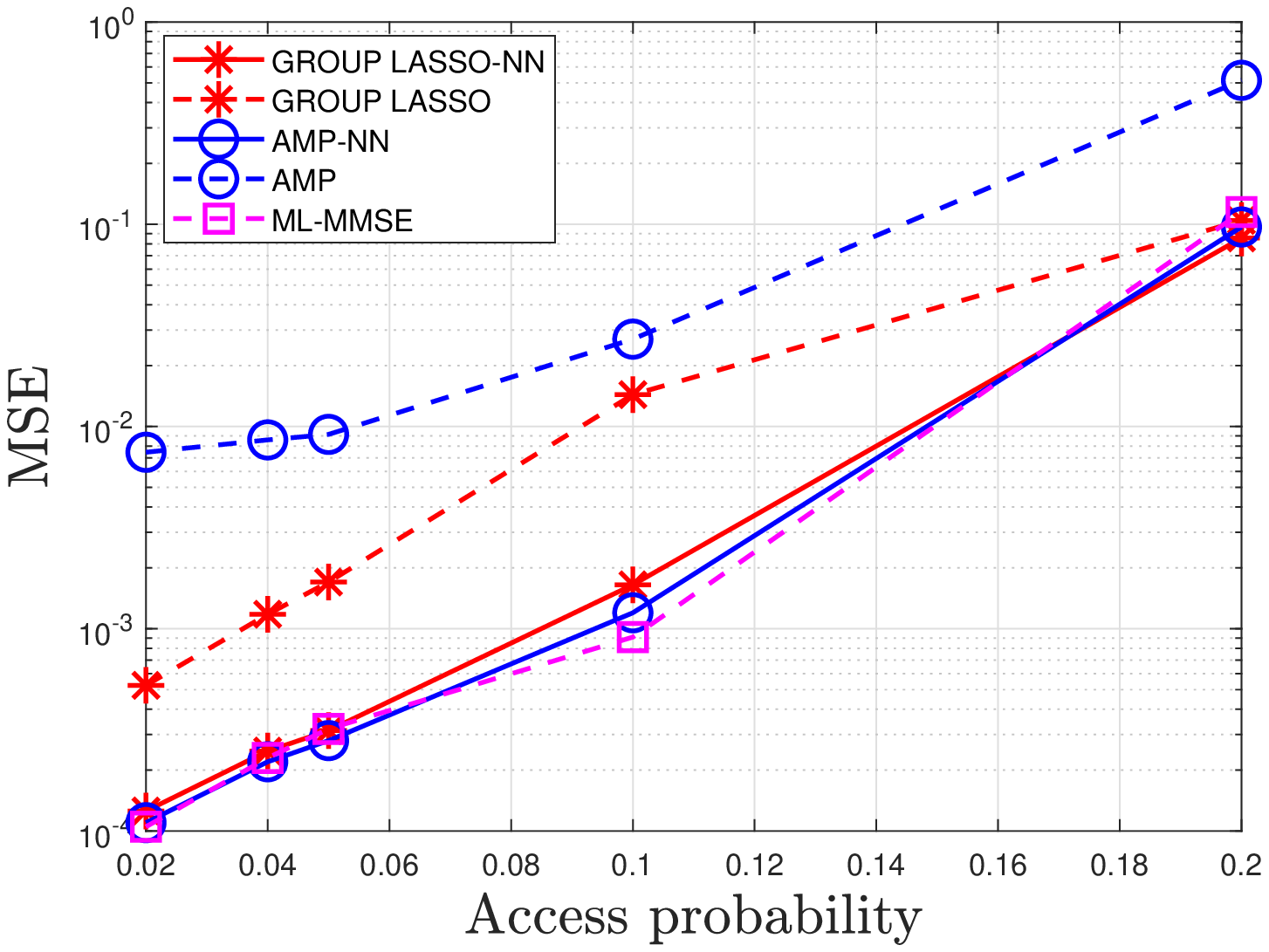}}}
 \subfigure[\scriptsize{MSE versus $M$ at $L/N=0.12$, $p=0.1$.}]
 {\resizebox{4.3cm}{!}{\includegraphics{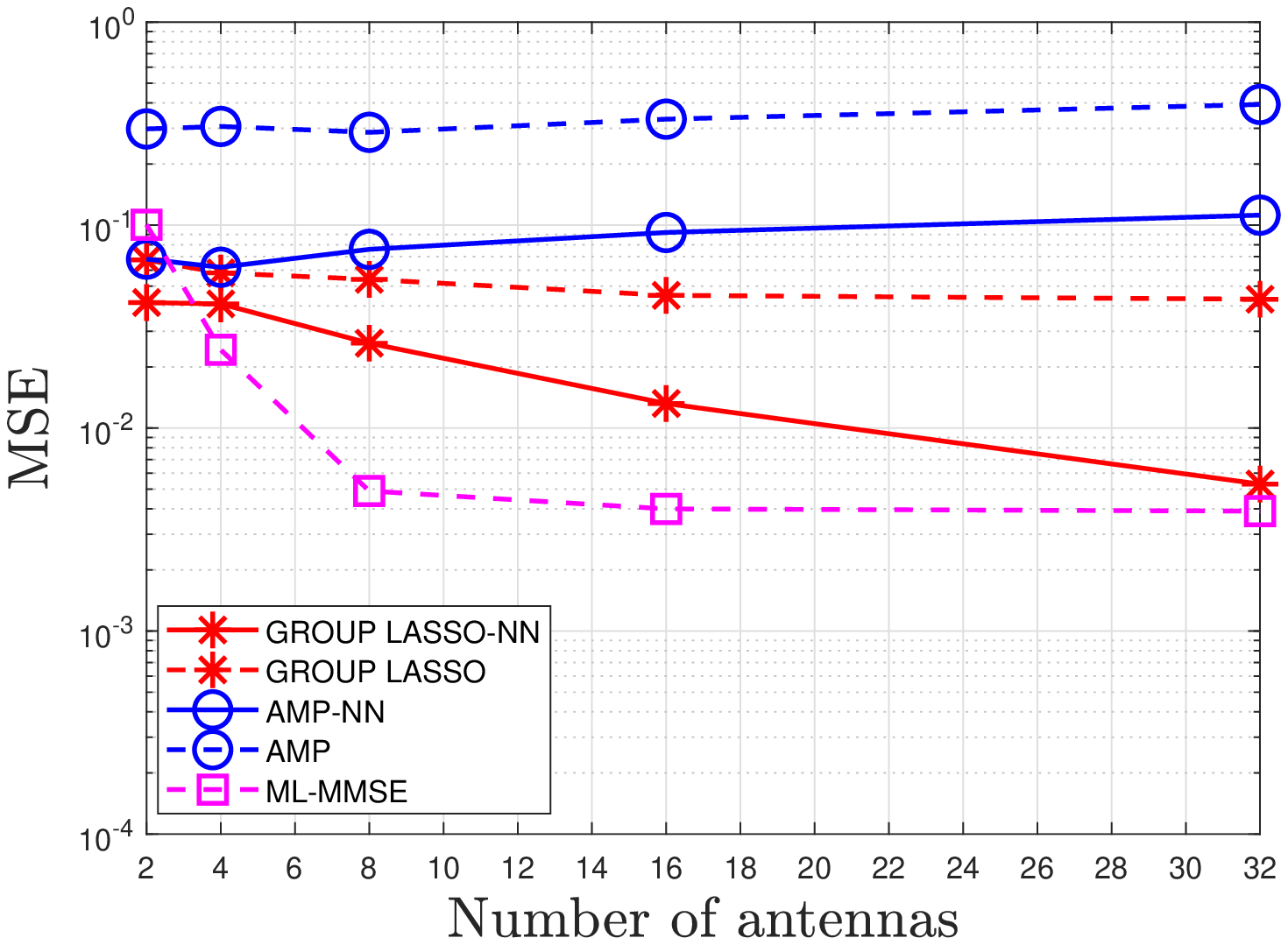}}}
  \end{center}
 \vspace*{-0.3cm}
   \caption{\small{\lsccc{Channel estimation in the} correlated case with a single active group at $N=100$.}}
   \label{signalco100}
\end{figure}

\begin{figure}[tp]
\begin{center}
 \subfigure[\scriptsize{MSE versus $L/N$ at $p=0.1$, $M=16$, $p_1/p_2=3$.}]
 {\resizebox{4.3cm}{!}{\includegraphics{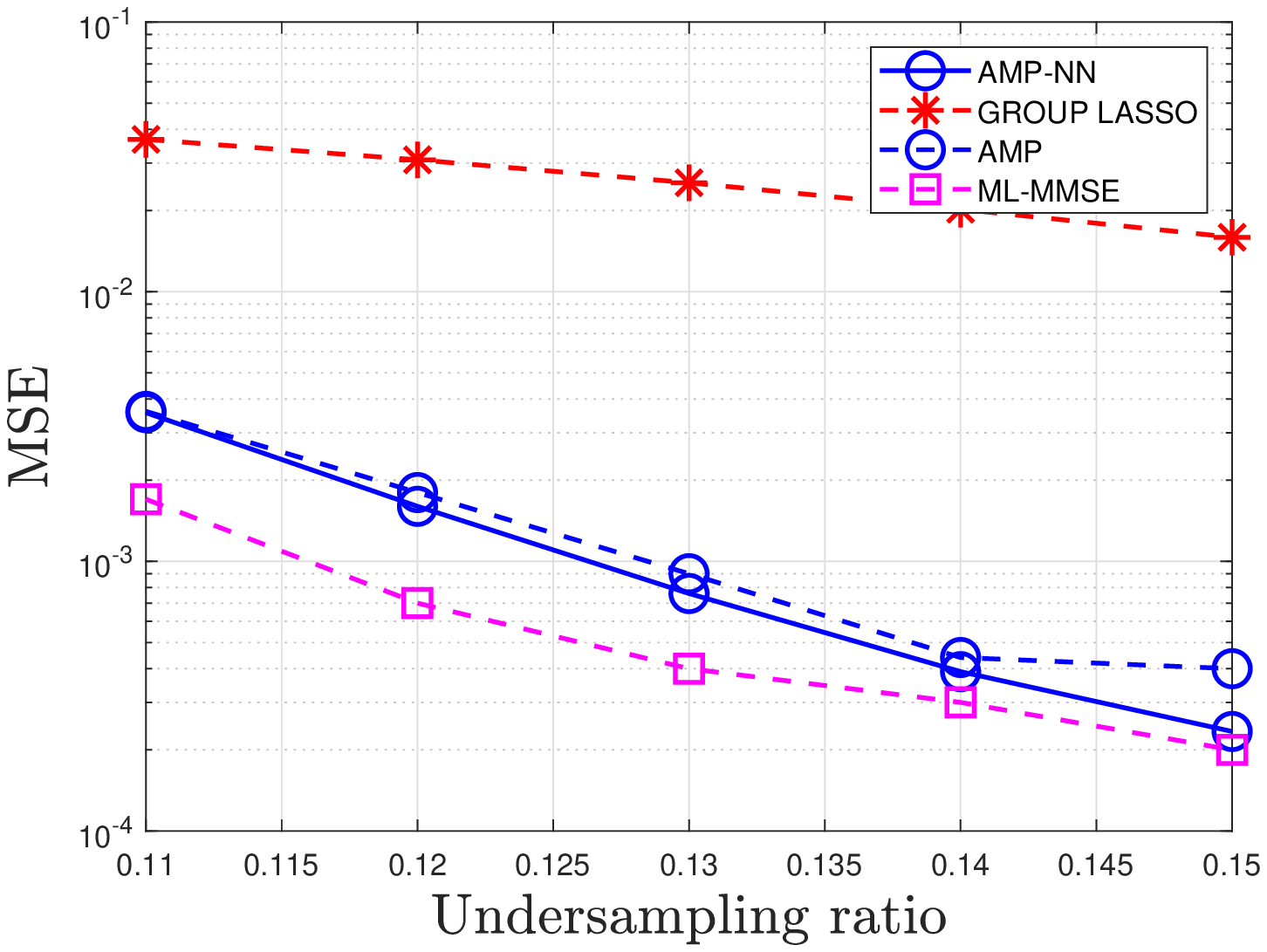}}}
 \subfigure[\scriptsize{MSE versus $p$ at $L/N=0.15$, $M=16$, $p_1/p_2=3$.}]
 {\resizebox{4.3cm}{!}{\includegraphics{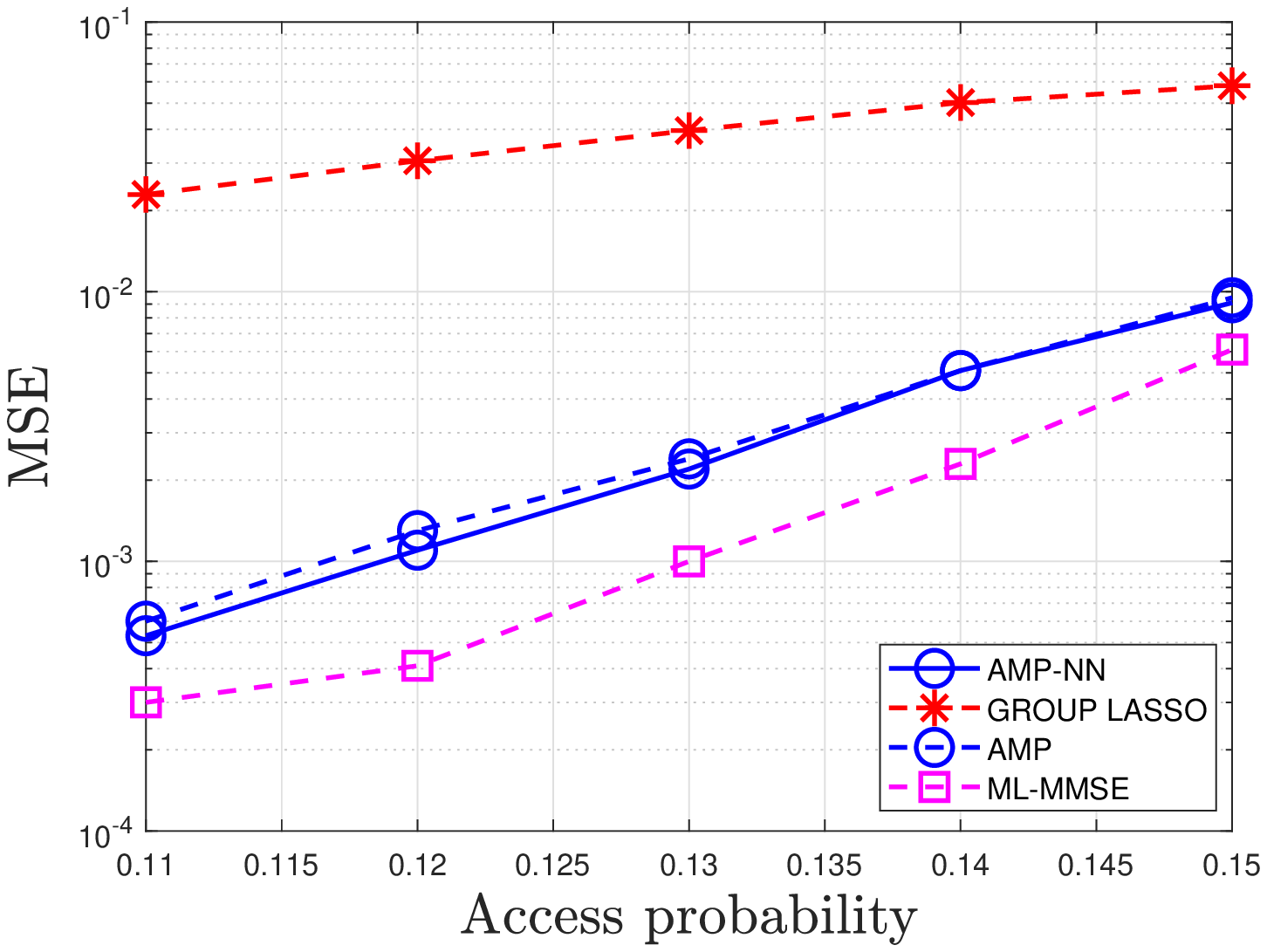}}}
 \subfigure[\scriptsize{\lscr{MSE versus $M$ at $L/N=0.11$, $p=0.1$, $p_1/p_2=3$.}}]
 {\resizebox{4.3cm}{!}{\includegraphics{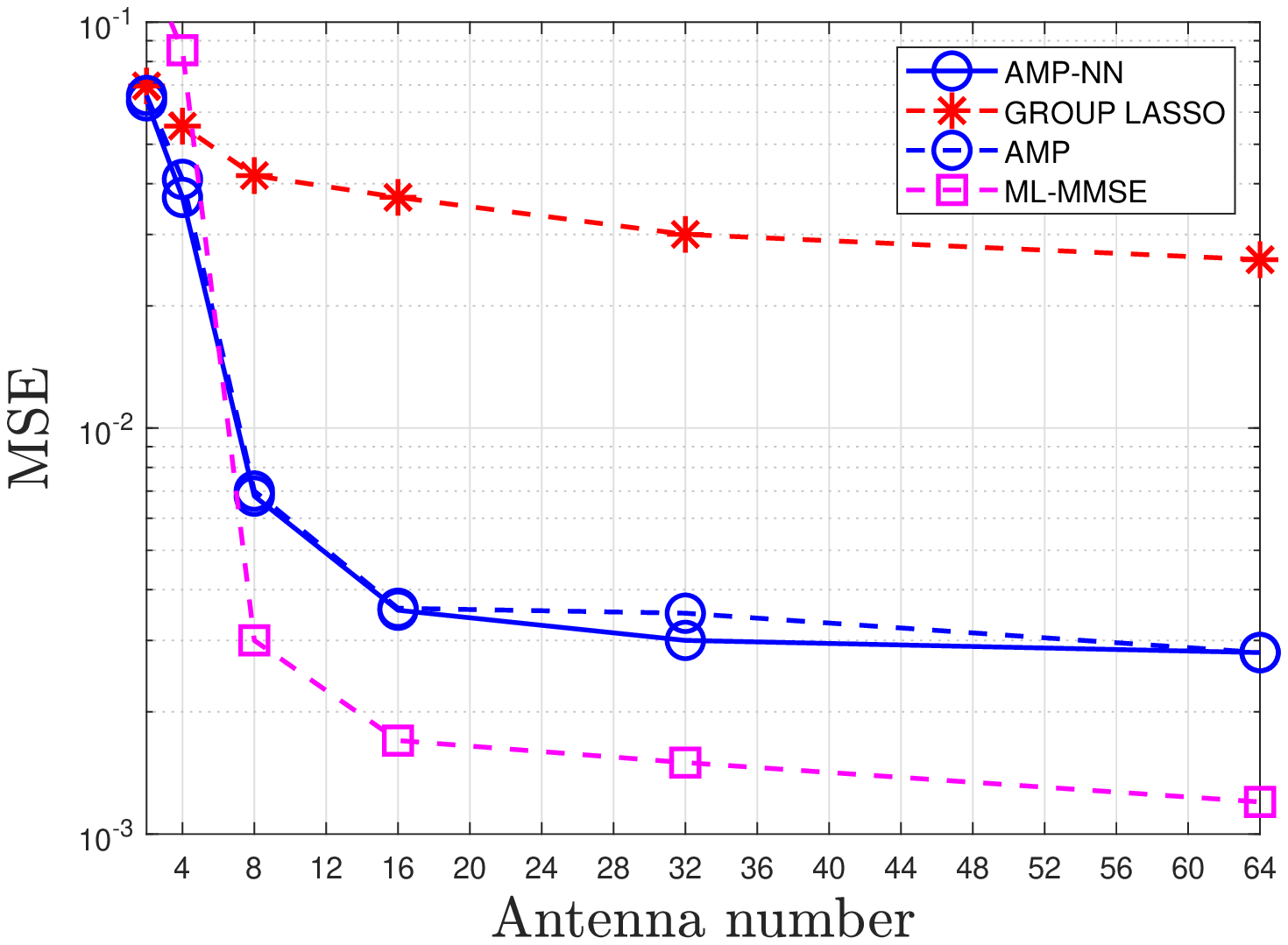}}}
 \subfigure[\scriptsize{MSE versus $p_1/p_2$ at $L/N=0.11$, $M=16$, $p=0.1$.}]
 {\resizebox{4.3cm}{!}{\includegraphics{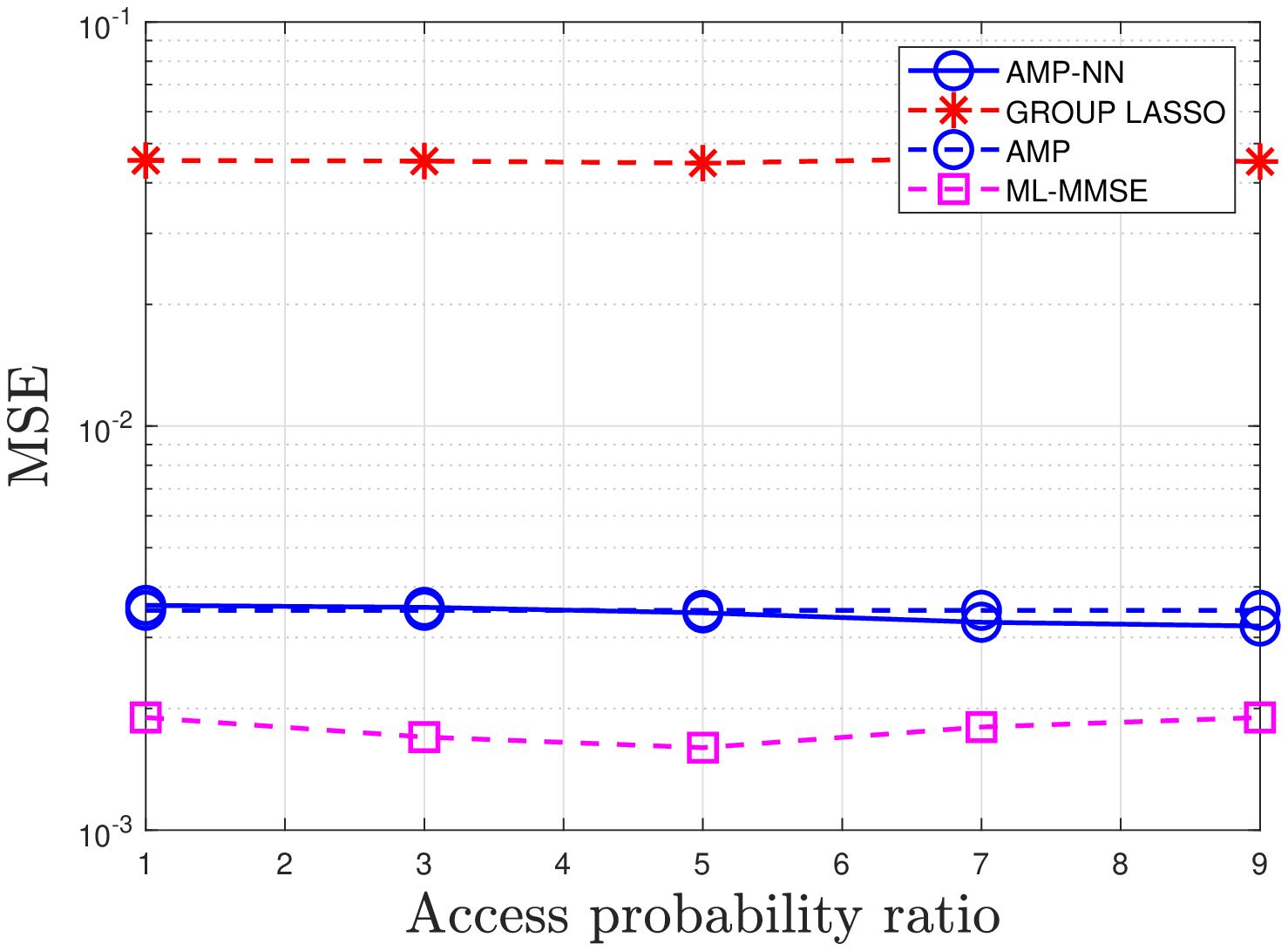}}}
  \subfigure[\scriptsize{Coherence versus $L/N$ \lscca{at $p=0.1$, $M=16$, $p_1/p_2=3$}.}]
 {\resizebox{4.4cm}{!}{\includegraphics{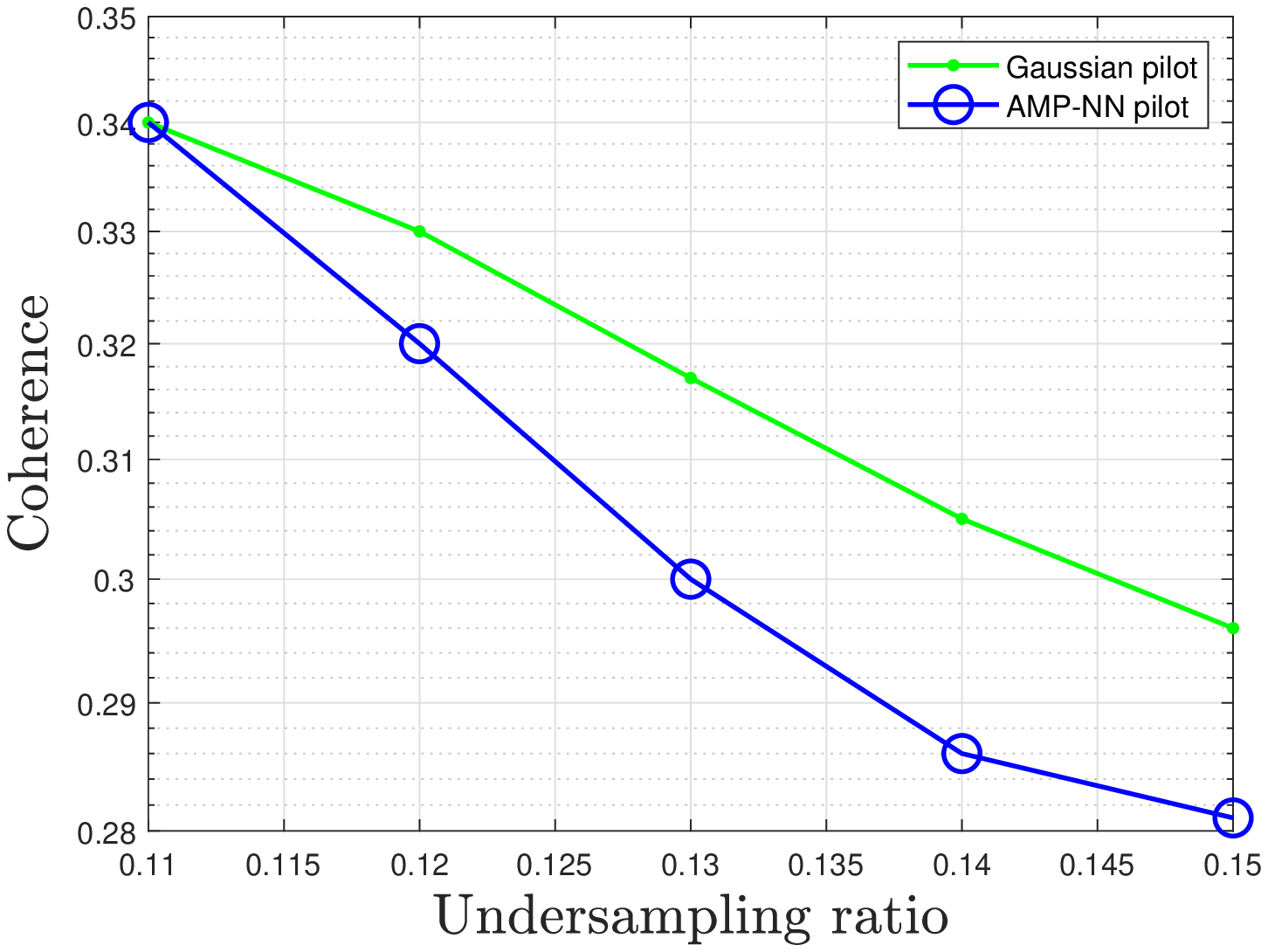}}}
  \end{center}
 \vspace*{-0.3cm}
   \caption{\small{\lsccc{Channel estimation in the} independent case at $N=1000$.}}
   \label{signal1000}
\end{figure}
\begin{figure}[tp]
\begin{center}
 \subfigure[\scriptsize{MSE versus $L/N$ at $p=0.1$, $M=16$.}]
 {\resizebox{4.3cm}{!}{\includegraphics{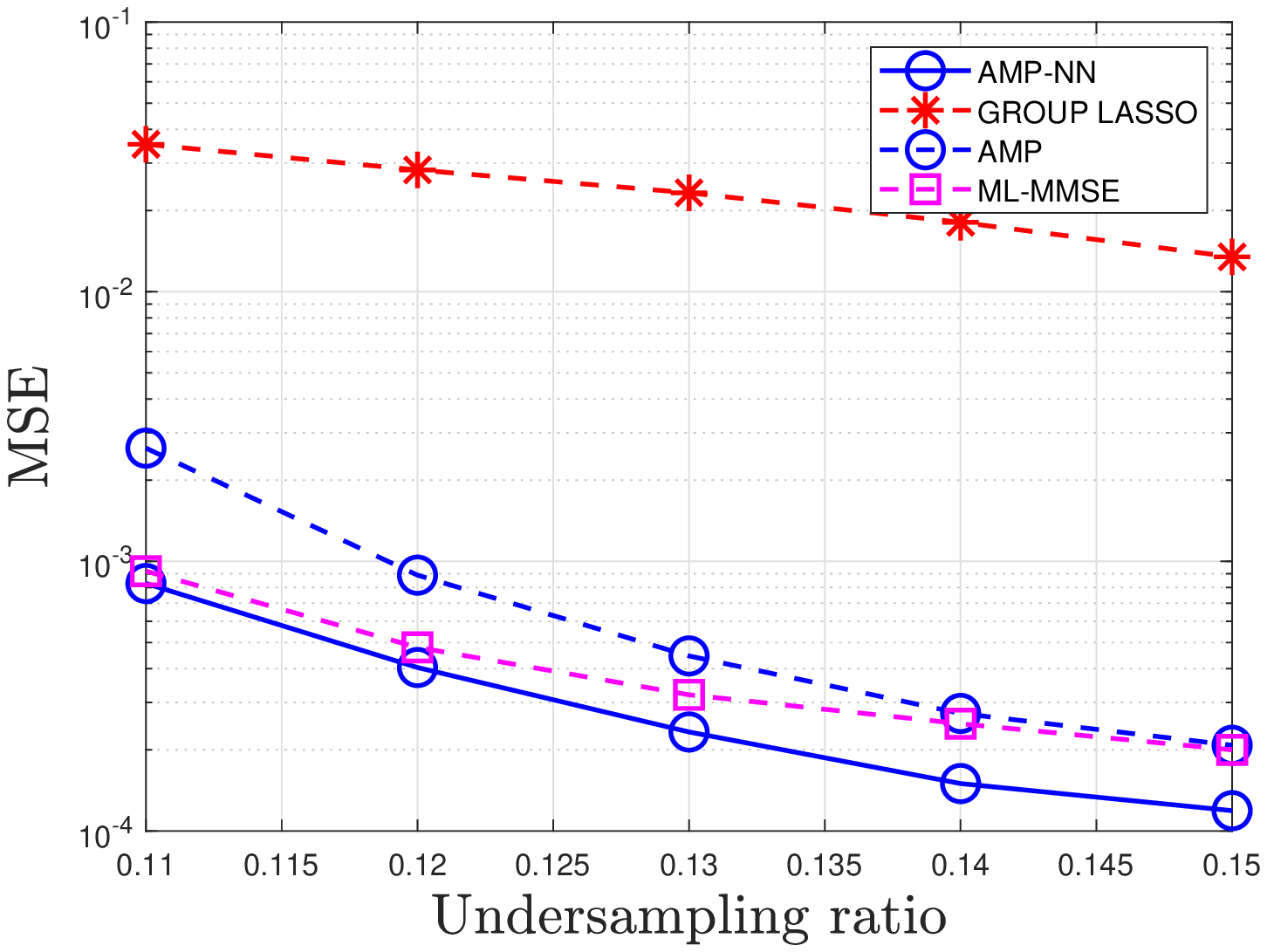}}}
 \subfigure[\scriptsize{MSE versus $p$ at $L/N=0.15$, $M=16$.}]
 {\resizebox{4.3cm}{!}{\includegraphics{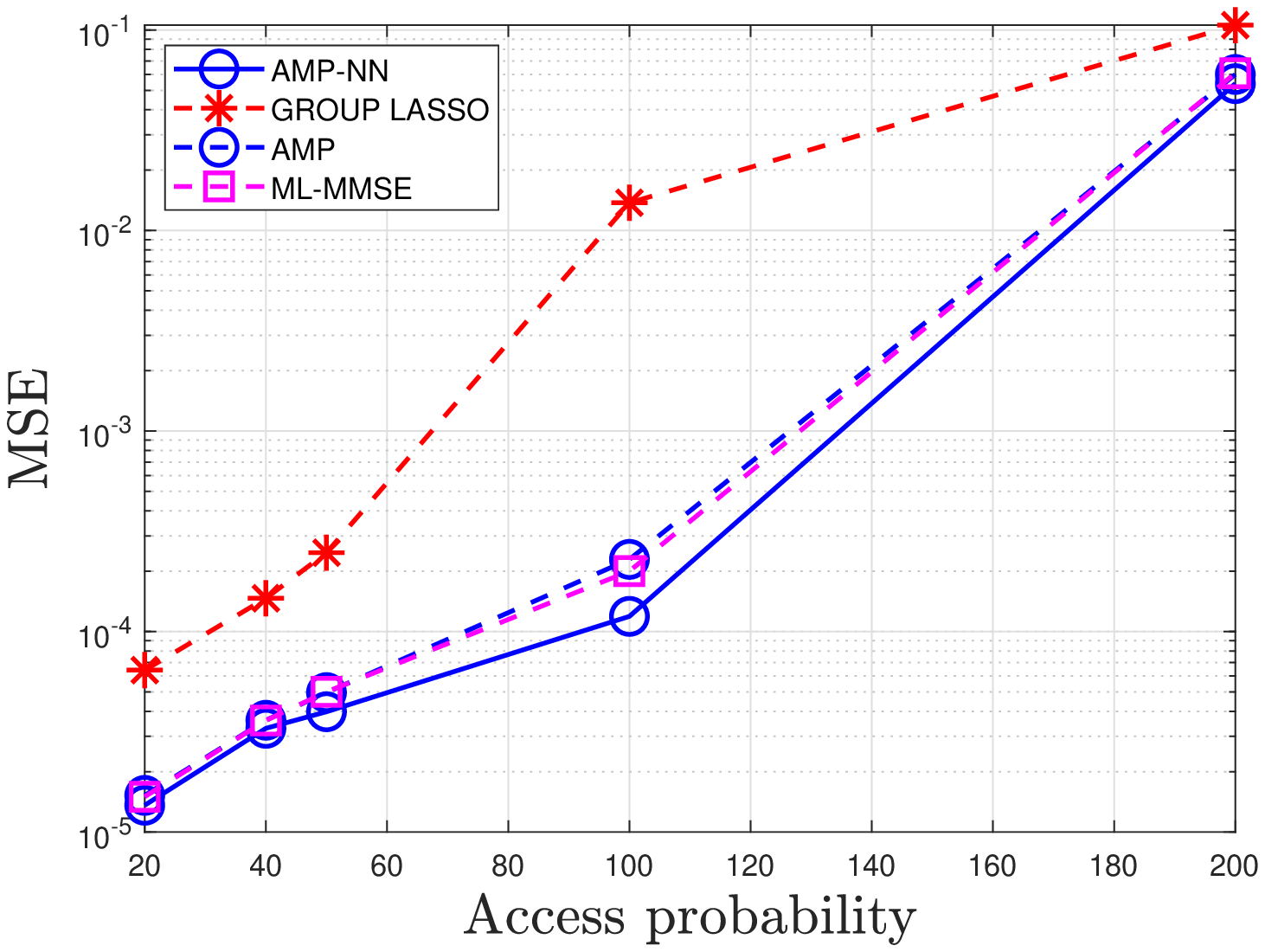}}}
 \subfigure[\scriptsize{\lscr{MSE versus $M$ at $L/N=0.11$, $p=0.1$.}}]
 {\resizebox{4.3cm}{!}{\includegraphics{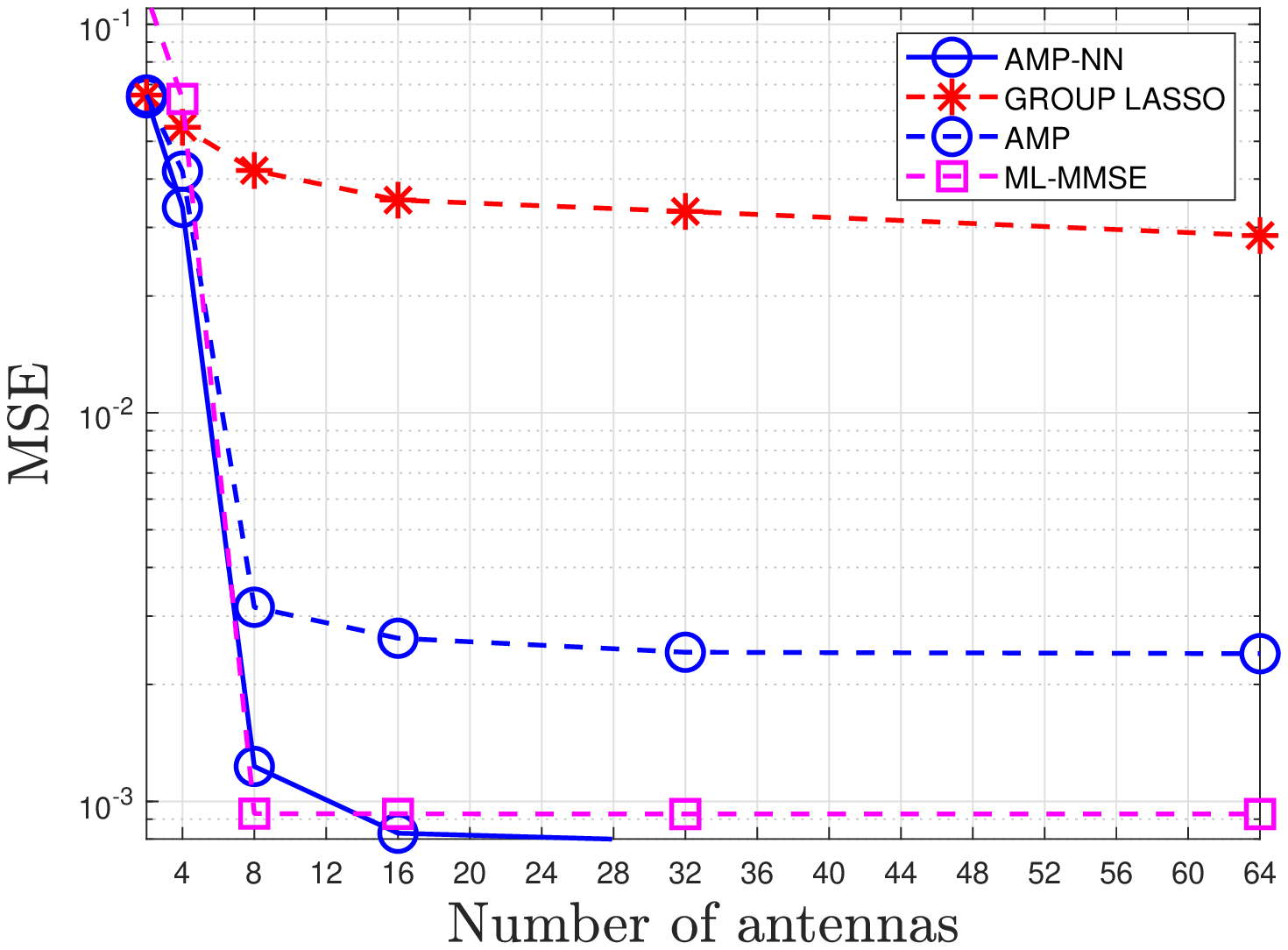}}}
  \subfigure[\scriptsize{Block-coherence across groups versus $L/N$ \lscca{at $p=0.1$, $M=16$.}}]
 {\resizebox{4.3cm}{!}{\includegraphics{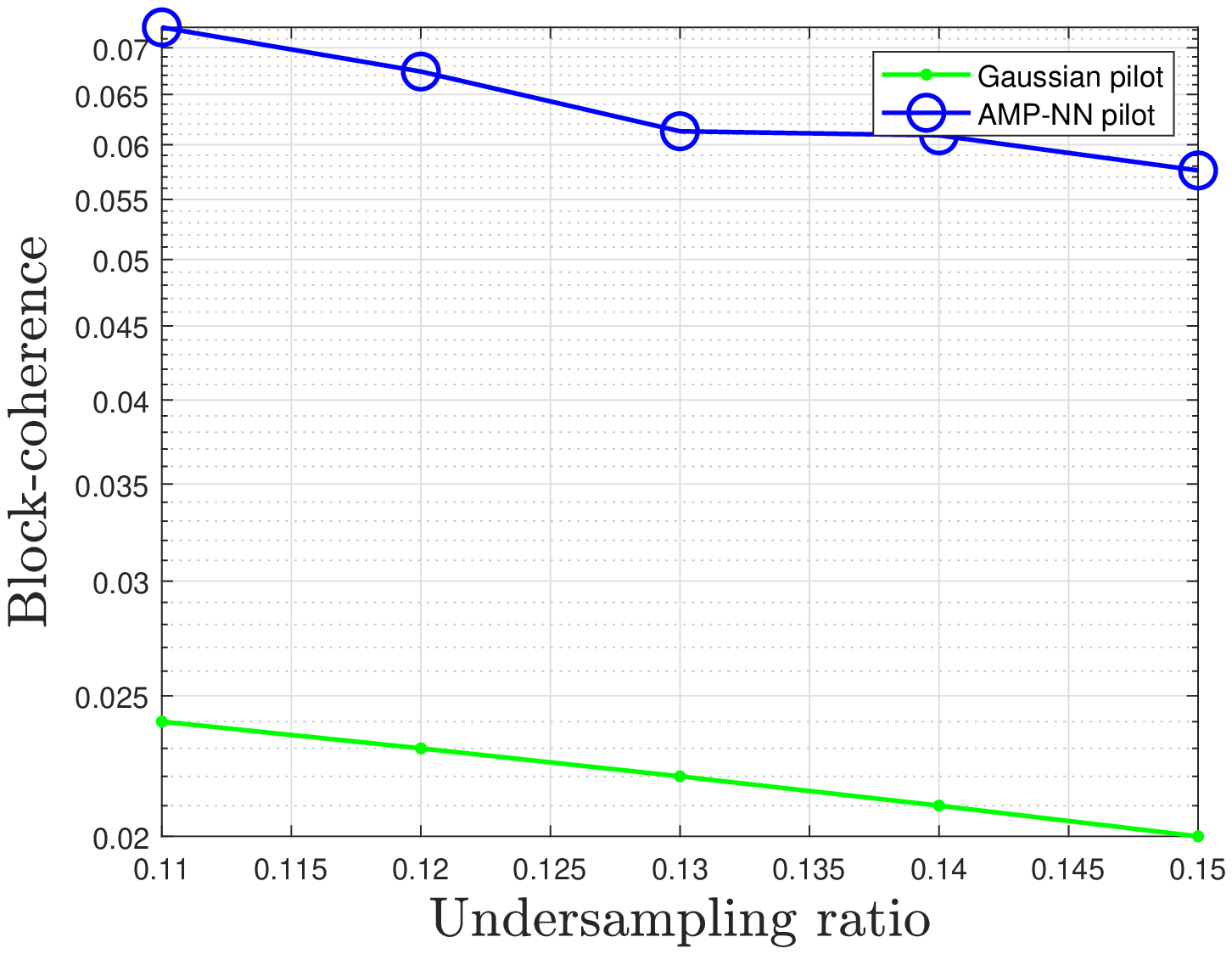}}}
 \subfigure[\scriptsize{Sub-coherence within one group versus $L/N$ \lscca{at $p=0.1$, $M=16$}.}]
 {\resizebox{4.3cm}{!}{\includegraphics{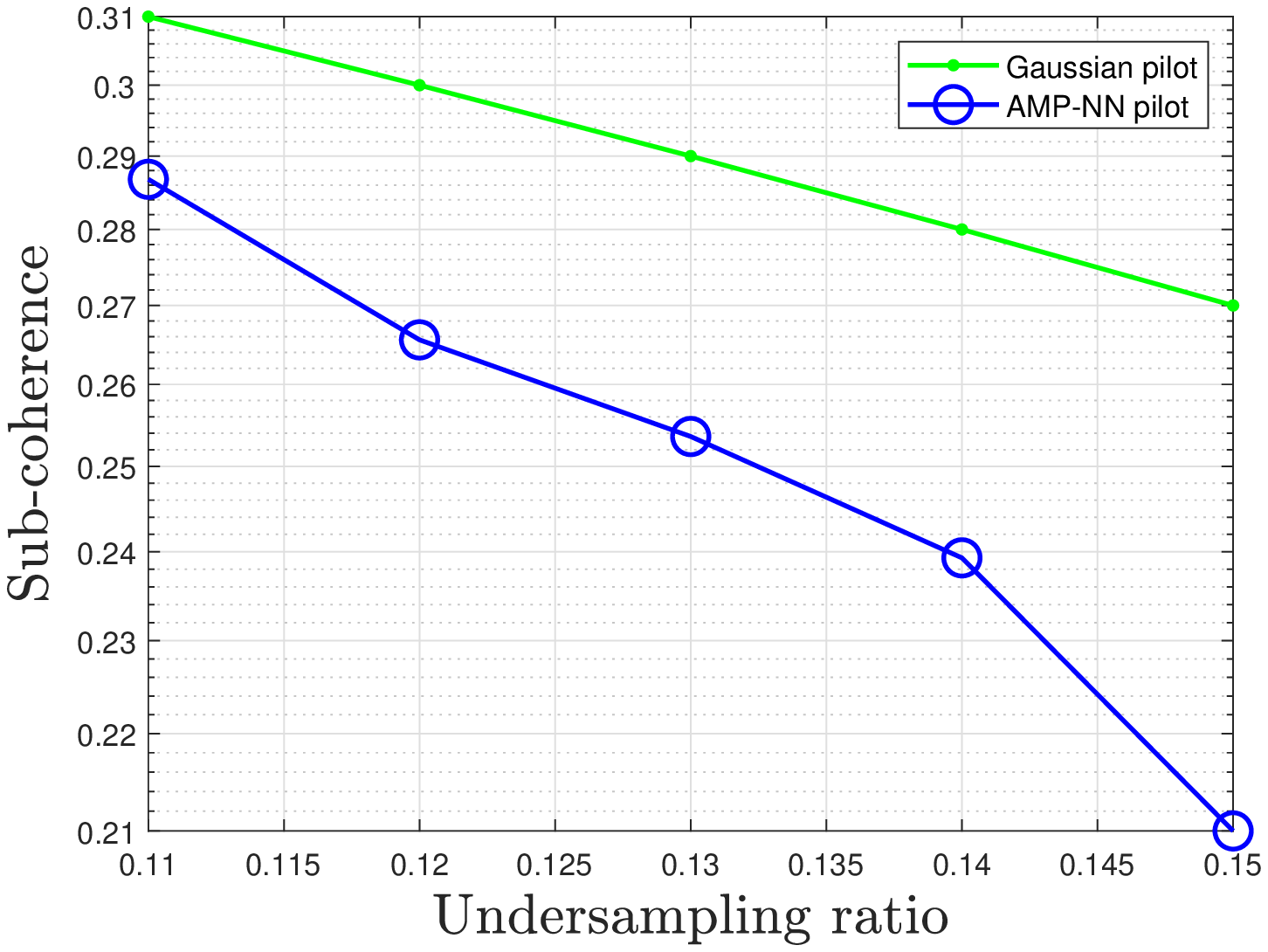}}}
  \end{center}
 \vspace*{-0.3cm}
   \caption{\small{\lsccc{Channel estimation in the} correlated case with a single active group at $N=1000$.}}
   \label{signalco1000}
\end{figure}

Fig.~\ref{signal100} (a), (b), (c), (d) and Fig.~\ref{signal1000} (a), (b), (c), (d) illustrate the MSE versus the undersampling ratio $L/N$, access probability $p$, \lscv{\lsccc{number of antennas} $M$ and access ratio $p_1/p_2$} in the independent case at $N=100$ and $N=1000$, respectively. Fig.~\ref{signalco100} and Fig.~\ref{signalco1000} (a), (b), (c) illustrate the MSE versus the undersampling ratio $L/N$, access probability $p$ and number of antennas $M$ in the correlated case with a single active group at $N=100$ and $N=1000$, respectively. From all these figures, we make the following observations. \lscc{The MSE of each scheme always decreases with $L/N$ and increases with $p$. In each case, the MSEs of all schemes \lscv{except} for AMP \lsccc{and AMP-NN} decrease with $M$ at $N=100$, \lscv{and the MSE of each scheme decreases with $M$ at $N=1000$}. The trend for AMP \lscv{at $N=100$} is abnormal, as AMP is not \lscv{suitable} for small $N$.  In the independent case, the MSEs of GROUP LASSO-NN and AMP-NN \lscv{decrease} with $p_1/p_2$, while the MSEs of the baseline schemes almost do not change with $p_1/p_2$. This phenomenon indicates that our proposed model-driven approach can exploit the difference in device activity probabilities to improve signal recovery accuracy, as shown in Fig.~\ref{signal100} (e).} \lscr{Our proposed GROUP LASSO-NN \nwqa{and AMP-NN always} outperform \lscv{the underlying} GROUP LASSO \nwqa{and AMP}, \wq{respectively}, demonstrating the competitive advantage of the proposed model-driven approach in designing effective measurement matrices and adjusting correction layers for improving signal recovery accuracy.}
\nwqa{The gains in the correlated case with a single active group are larger than those in the independent case.
In the independent case, the gains are larger \lscv{at larger} $p_1/p_2$.
\lscv{Note that larger} gains \lscv{appear} when more features of sparsity pattens exist and can be utilized by the proposed model-driven approach.}
\lscr{In the independent case with $N=1000$, the MSEs of AMP-NN and AMP are \lscv{close.} The reason is that \lsccc{not many features of} sparsity patterns \lscv{can be utilized for AMP-NN,} and AMP already achieves very excellent recovery accuracy \lscv{at large $N$}.} GROUP LASSO-NN outperforms AMP-NN when $N$ and $L/N$ are small, while AMP-NN outperforms GROUP LASSO-NN when $N$ and $L/N$ are large. \lscv{In most cases,} ML-MMSE has the smallest MSE, at the cost of \nwqa{computation time increase (\lscb{due to} ML}) \cite{8437359}. \lscv{In the correlated case with a single active group, when $N=1000$,} AMP-NN \lscv{achieves the smallest MSE. This is because the underlying} AMP functions well for large $N$ and the encoder and correction layers \nwqa{of AMP-NN} effectively utilize the sparsity patterns for further improving signal recovery accuracy.

Fig.~\ref{signal1000} (e) shows the coherence of our learned measurement matrix in the independent case and Gaussian matrix versus the undersampling ratio $L/N$ at $N=1000$. Fig.~\ref{signalco1000} (d), (e) show the block-coherence across groups and sub-coherence within one group of our learned measurement matrix in the correlated case with a single active group and Gaussian matrix versus the undersampling ratio $L/N$ at $N=1000$. The definitions of coherence, block-coherence and sub-coherence of a matrix can be found in \cite{8264817}. Roughly speaking, the coherence and sub-coherence of a matrix reflect the orthogonality of all the columns \lscb{and orthogonality of all the} columns within one group, \lscb{respectively}, and block-coherence reflects the overall orthogonality between two groups of columns. Note that the measures for Gaussian matrix are obtained by averaging over $10000$ realizations.  The measures for the learned measurement matrix in the correlated case with a single active group are averaged over all groups. From Fig.~\ref{signal1000} (e) and Fig.~\ref{signalco1000} (e), we can observe that the learned measurement matrix of our proposed AMP-NN has smaller coherence and sub-coherence than Gaussian matrix, \lscca{to more effectively differentiate devices that may be active simultaneously and devices that} are always active simultaneously, \lscca{respectively}. From Fig.~\ref{signalco1000} (d), we can see that the learned measurement matrix of our proposed AMP-NN has larger block-coherence than Gaussian matrix. This is to give higher priority to the differentiability of the devices that are always active simultaneously, at the sacrifice of the differentiability of the devices that never activate at the same time.

\begin{figure}[tp]
\begin{center}
 \subfigure[\scriptsize{\lscr{Computation time versus $L/N$ at $M=4$ and $N=100$.}}]
 {\resizebox{4.3cm}{!}{\includegraphics{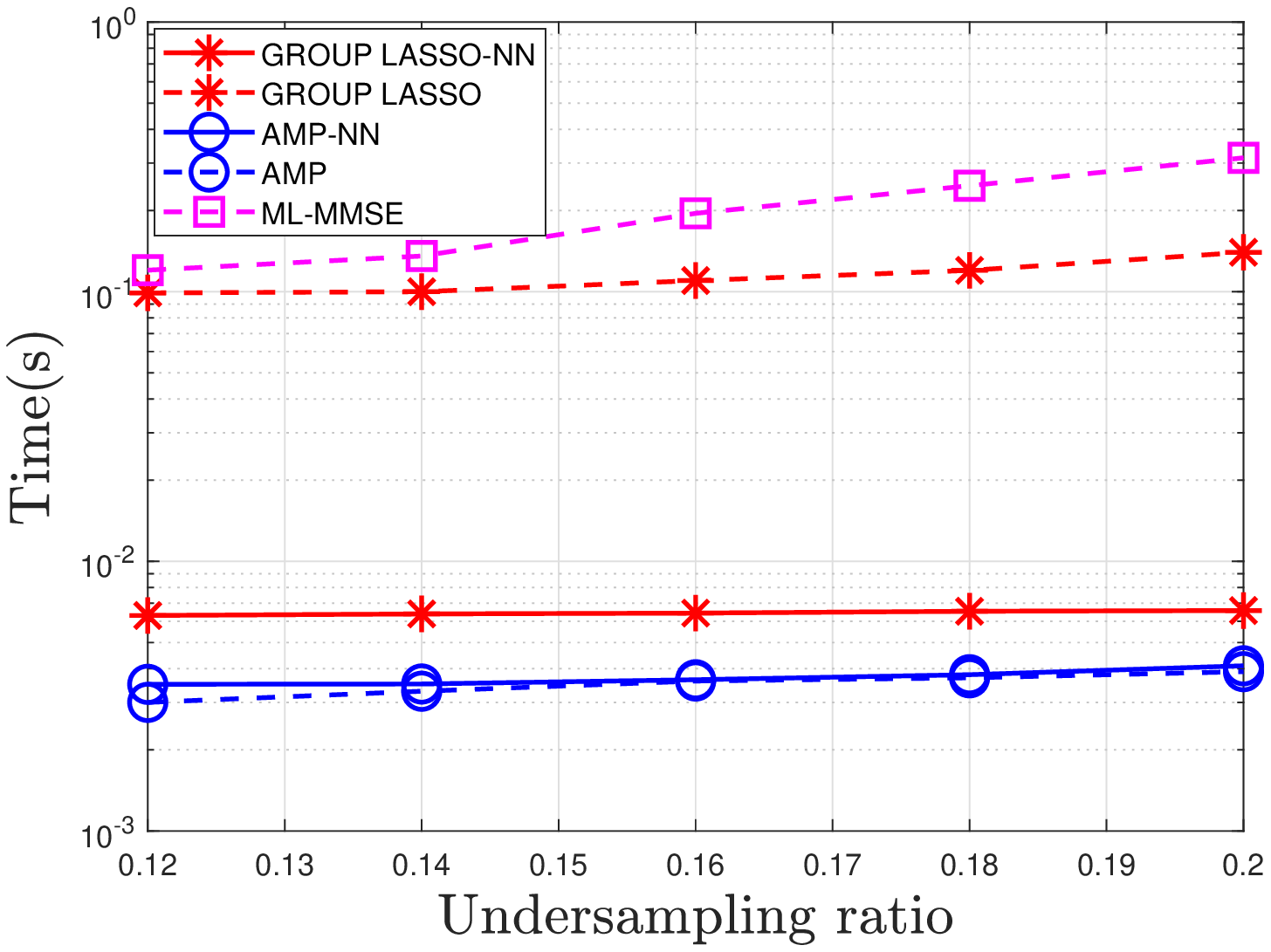}}}
 \subfigure[\scriptsize{\lscr{Computation time versus $L/N$ at $M=16$ and $N=1000$.}}]
 {\resizebox{4.3cm}{!}{\includegraphics{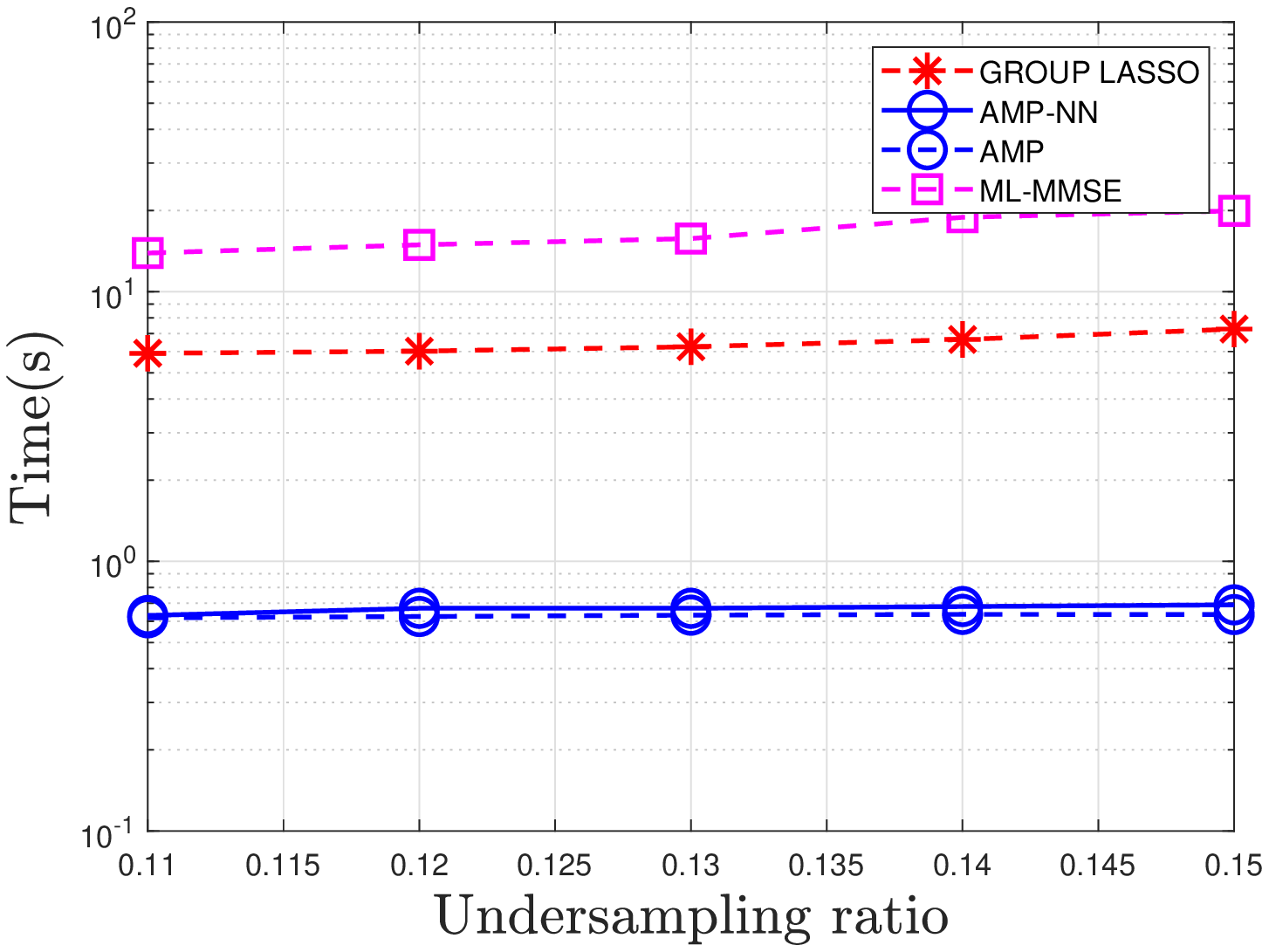}}}
  \end{center}
 \vspace*{-0.3cm}
   \caption{\small{Computation time \lsccc{for channel estimation} in the independent case \lsccc{at $p=0.1$ and $p_1/p_2=3$}.}}
   \label{timesig}
\end{figure}

Fig.~\ref{timesig} illustrates the computation time versus the undersampling ratio $L/N$ in the independent case at $N=100$ and $N=1000$.\footnote{\lscb{Note that in our setup, the computation time of a scheme in the independent case is identical to that in \lscca{a} correlated case.}} \lscc{From Fig.~\ref{timesig}, we can make the following observations. The computation time of each scheme increases with $L/N$. GROUP LASSO-NN has shorter computation time than GROUP LASSO. \lscn{The gain derives from the fact that Algorithm~1 allows parallel computation}. GROUP LASSO and ML-MMSE have much longer computation time than the other schemes, especially at large $N$. Thus, \lscv{they may not be applicable for channel estimation in} practical mMTC (with large $N$)}. AMP-NN and AMP have the \lscb{similar} computation time \lscb{(which is the shortest), as they have the same number of iterations and $V=3$ is \lsccc{quite} small}.
%Fig.~\ref{timesig} illustrates the computation time versus the undersampling ratio $L/N$ in the independent case at $N=100$ and $N=1000$. From Fig.~\ref{timesig}, we note that our proposed \lscr{AMP-NN} has the smallest computation time in most cases. From Fig.~\ref{timesig} (a), we observe that the computation times of the proposed AMP-NN and GROUP LASSO-NN are smaller than AMP and GROUP LASSO, respectively, in most cases, owing to the smaller numbers of iterations adopted in the approximation parts and the parallelizable neural network architecture. From Fig.~\ref{timesig} (b), we observe that the computation time of ML-MMSE \lscr{has much longer computation time} than AMP and the proposed AMP-NN \lscr{at large $N$. This makes ML-MMSE} not quite practical for channel estimation in grant-free random access, even though it has very excellent signal recovery performance.

\subsection{Device Activity Detection}
\label{sim}
In this subsection, we present numerical results on \lscr{device activity detection} in MIMO-based grant free random access. We evaluate the proposed model-driven approaches with \lscr{the} covariance-based decoder, i.e., NN, \lscr{and the MAP-based decoder, i.e., MAP-NN, both} at $N = 100$\lscr{,} and \lscr{with the} AMP-based decoder, i.e., AMP-NN, at $N = 100$ and $N=1000$.\footnote{When $N$ is large, training the proposed NN needs a large number of samples, and it takes a large number of iterations for Algorithm~3 to \lscv{converge to a reasonable} estimate. Thus, we do not adopt NN and MAP-NN at $N=1000$. \lscn{From numerical results, we find that i.i.d. Gaussian pilots are quite suitable for Algorithm~3 and the proposed approach cannot find better measurement matrix for the MAP-based decoder}.} \lscr{We consider four baseline schemes,} namely AMP \cite{8323218}, ML \cite{8437359}, \lscr{GROUP LASSO \lscb{(which adopts the block coordinate descent algorithm)} \cite{qin2013efficient}} and covariance-based LASSO \lscb{(which adopts the coordinate descent algorithm)} \cite{6994860}. \lscr{The computational complexities for AMP, ML, GROUP LASSO and covariance-based LASSO are $\mathcal{O}(LNM)$, $\mathcal{O}(NL^2)$, $\mathcal{O}(LNM)$ and $\mathcal{O}(LNM)$, respectively.} We set $U=0$ and $V=3$ for NN, set $U=55$ and $V=3$ for MAP-NN\lscr{,} and set $U=50$ and $V=3$ for AMP-NN. The choices are based on a large number of experiments and the tradeoff between  recovery accuracy and computation time. \lscn{For \lscb{ease} of comparison,  the total numbers of iterations for the block coordiante descent algorithm for ML, the AMP algorithm, the block coordinate descent algorithm for GROUP LASSO and the block coordiante descent algorithm for covariance-based LASSO are chosen as $55$, $50$, $200$ and $200$, respectively.} We numerically evaluate the average error rate of device activity detection \lscb{$\frac{1}{NT} \sum_{t=1}^T\|\boldsymbol{\alpha}^{[t]}-\hat{\boldsymbol{\alpha}}^{[t]} \|_1$} and computation time (on the same server) of each scheme over $T=10^3$ testing samples.

\begin{figure}[tp]
\begin{center}
 \subfigure[\scriptsize{Error rate versus $L/N$ at $p=0.1$, $M=4$, $p_1/p_2=3$.}]
 {\resizebox{4.3cm}{!}{\includegraphics{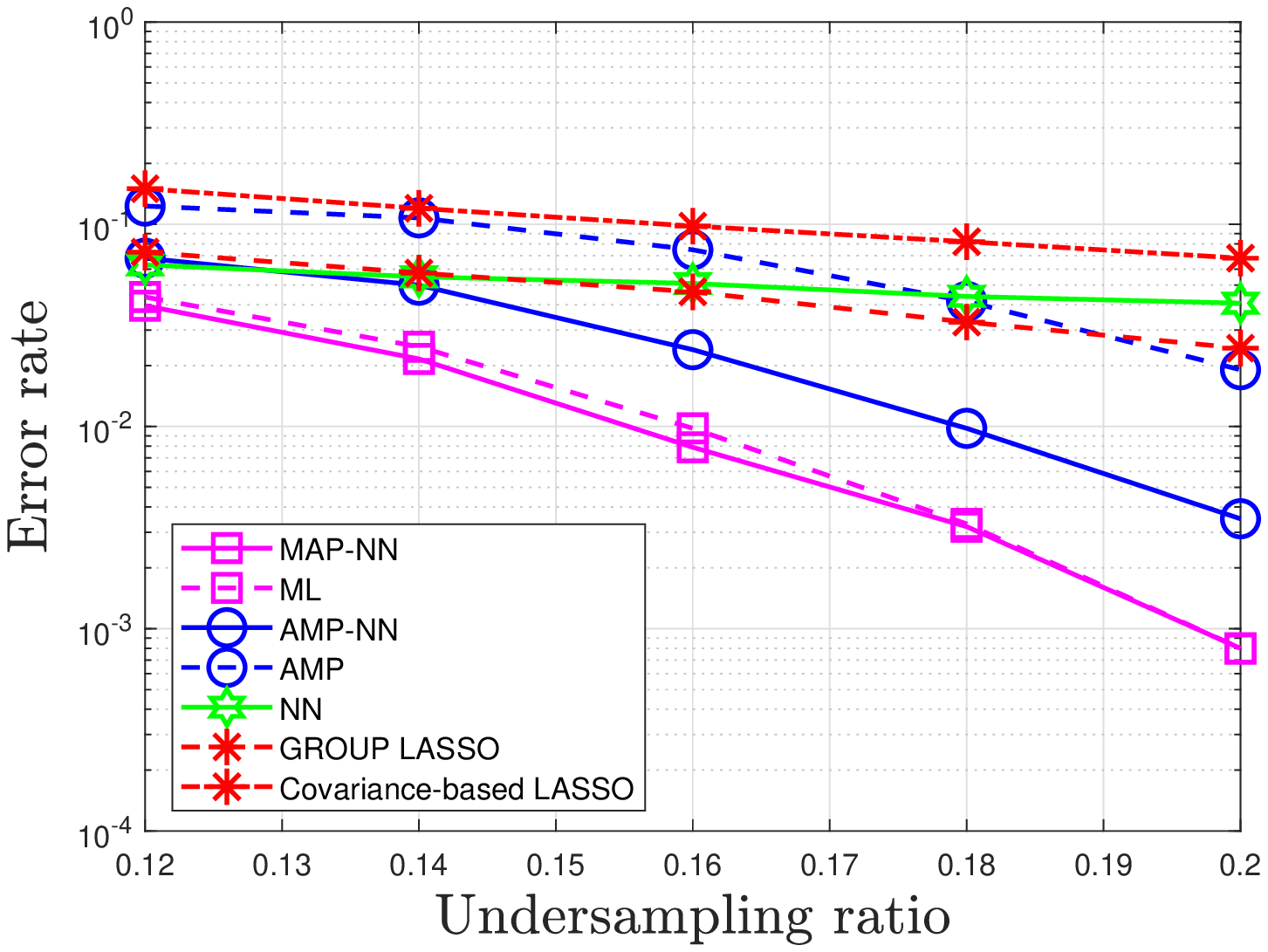}}}
 \subfigure[\scriptsize{Error rate versus $p$ at $L/N=0.2$, $M=4$, $p_1/p_2=3$.}]
 {\resizebox{4.3cm}{!}{\includegraphics{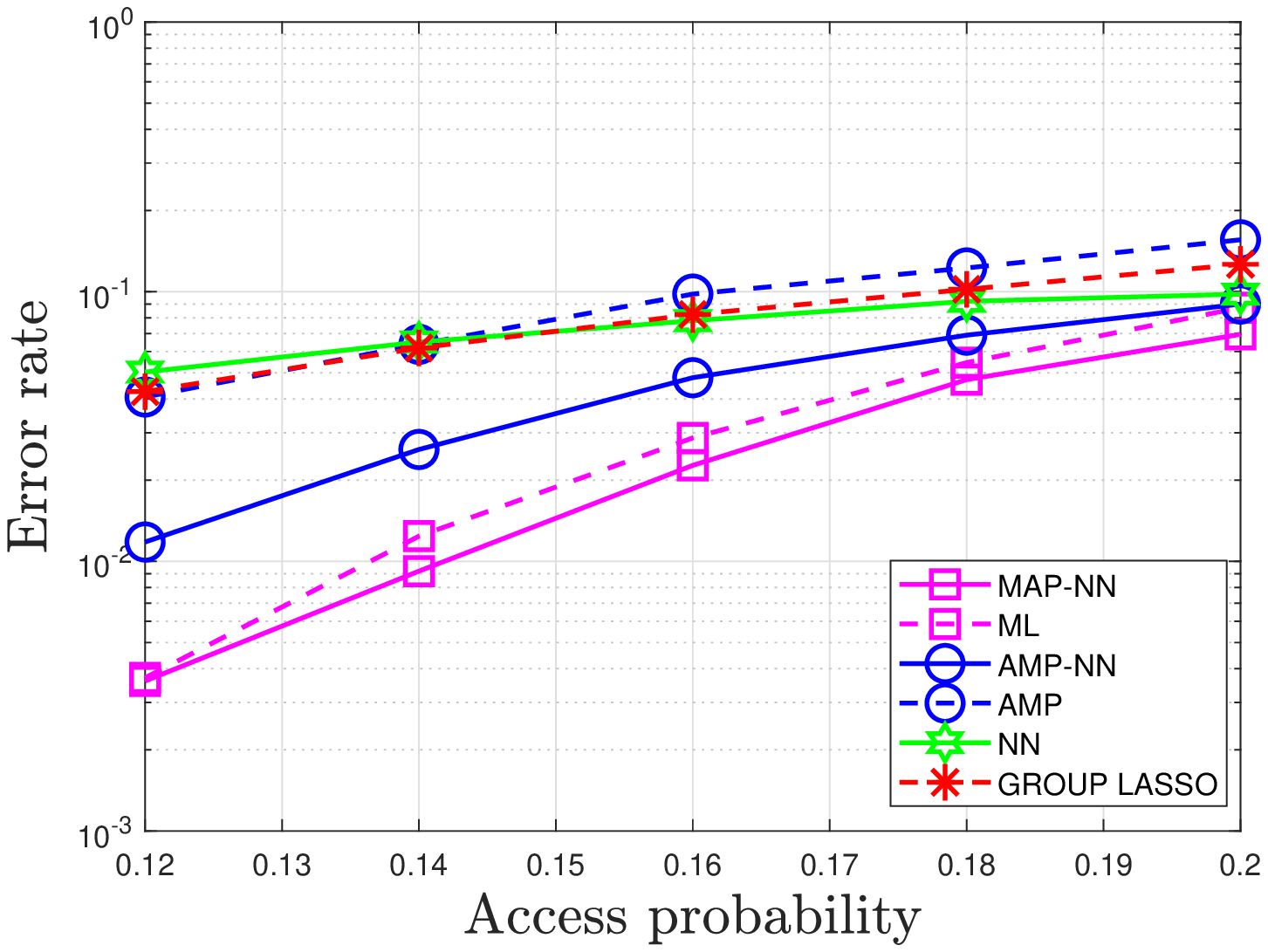}}}
 \subfigure[\scriptsize{Error rate versus $M$ at $L/N=0.12$, $p=0.1$, $p_1/p_2=3$.}]
 {\resizebox{4.3cm}{!}{\includegraphics{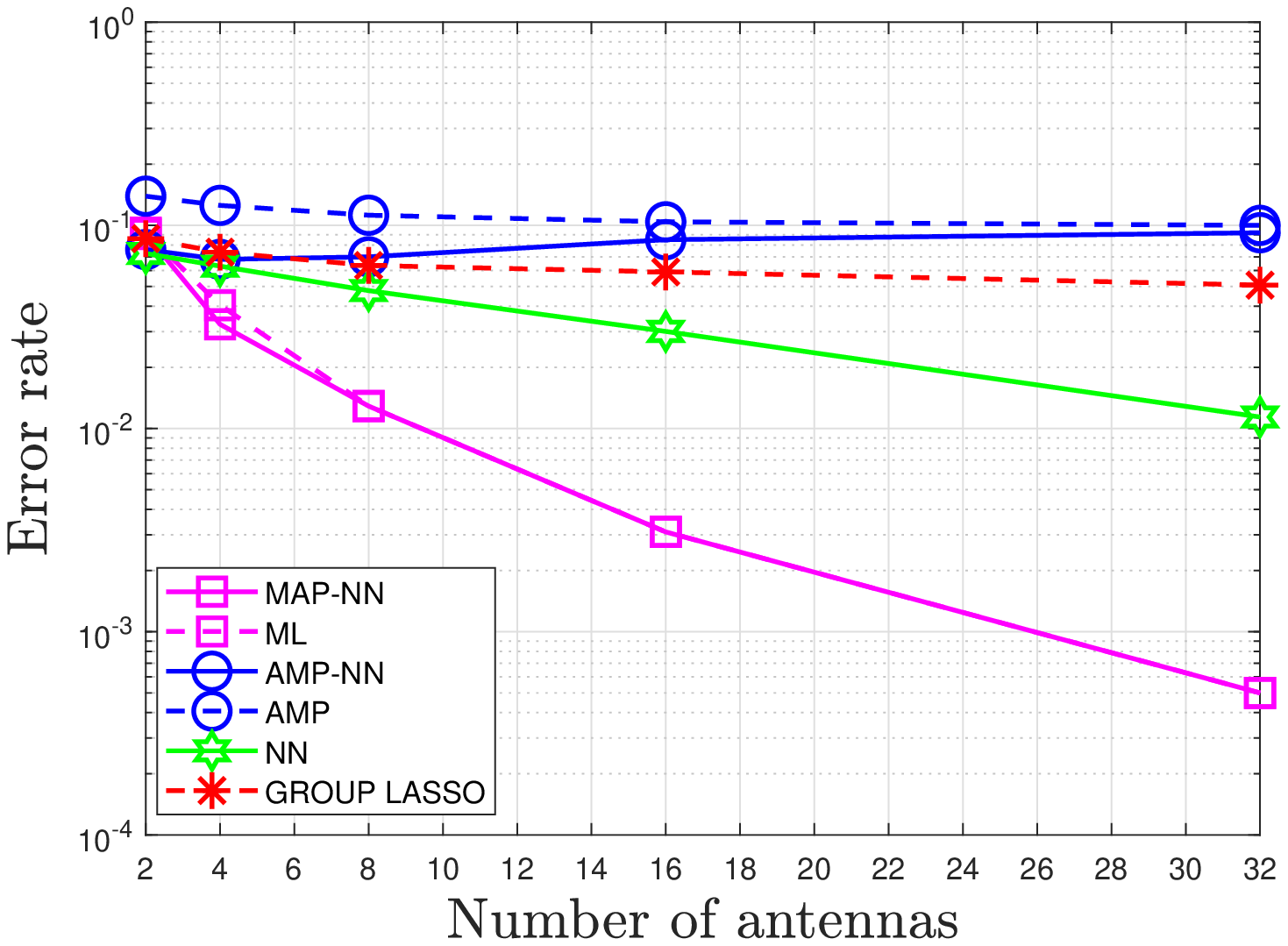}}}
 \subfigure[\scriptsize{Error rate versus $p_1/p_2$ at $L/N=0.12$, $M=4$, $p=0.1$.}]
 {\resizebox{4.4cm}{!}{\includegraphics{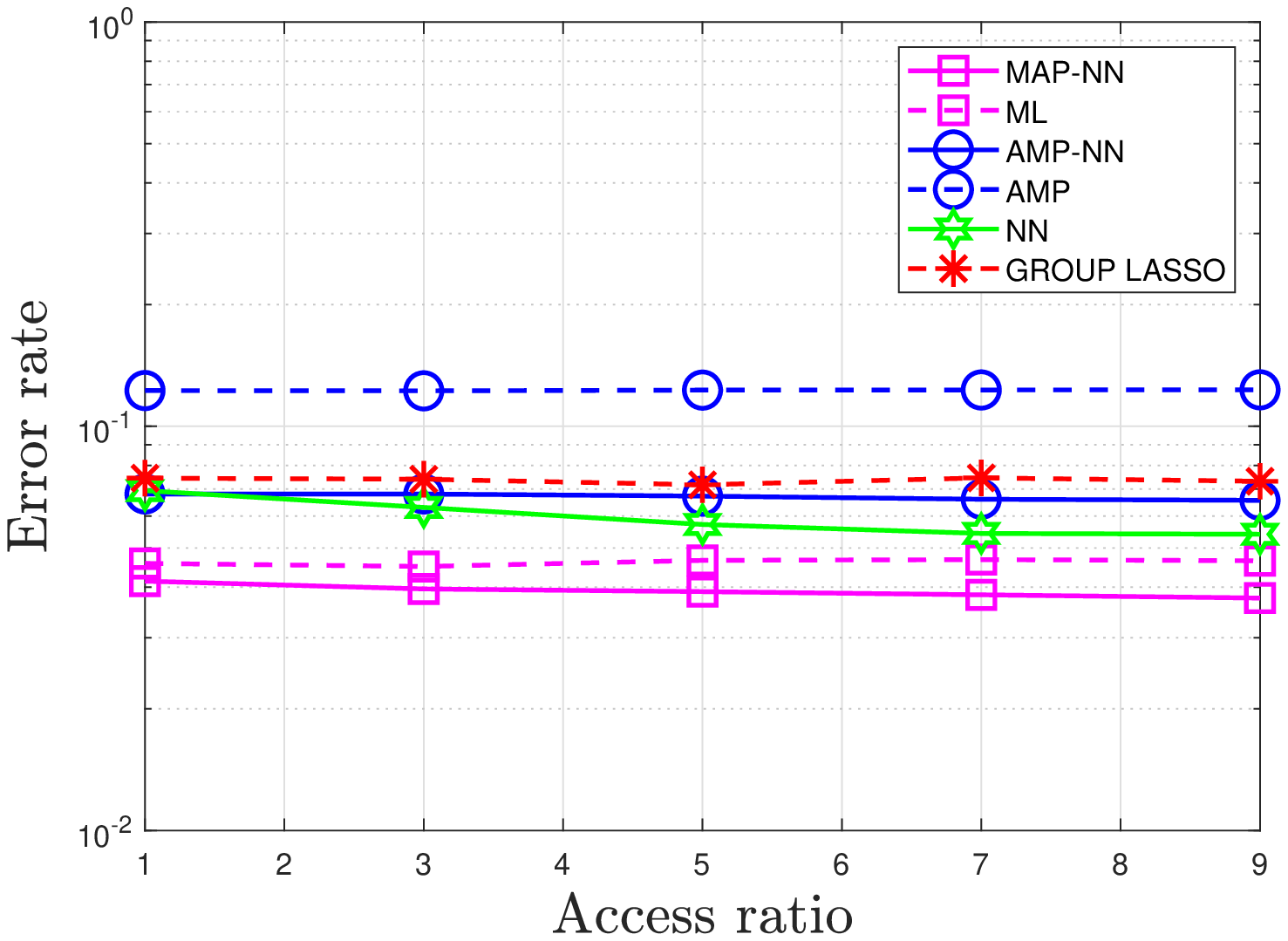}}}
   \subfigure[\scriptsize{$\epsilon_1/\epsilon_2$ versus $p_1/p_2$ at $L/N=0.12$, $M=4$, $p=0.1$.}]
 {\resizebox{4.3cm}{!}{\includegraphics{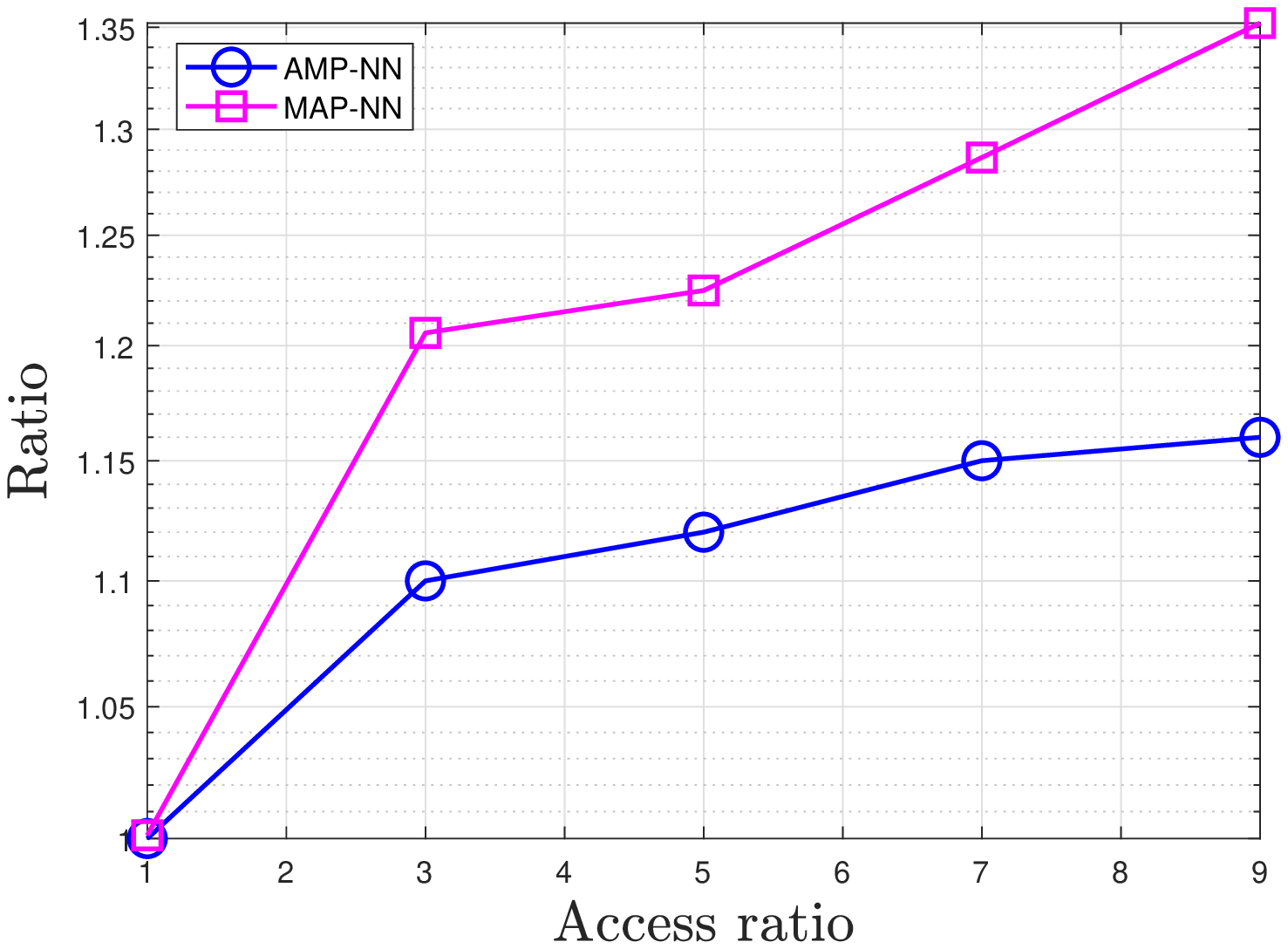}}}
  \end{center}
 \vspace*{-0.3cm}
   \caption{\small{\lsccc{Device activity detection in the} independent case at $N=100$, $\epsilon_1=\frac{2}{N}\sum_{n\in\mathbb{N}_1}\epsilon(n)$ and $\epsilon_2=\frac{2}{N}\sum_{n\in\mathcal{N}_2}\epsilon(n)$, where $\epsilon(n),n\in \mathcal{N}$, playing the role of the access probabilities in Algorithm~3, and are extracted from the trained AMP-based decoder.}}
   \label{support100}
\end{figure}

\begin{figure}[tp]
\begin{center}
 \subfigure[\scriptsize{Error rate versus $L/N$ at $p=0.1$, $M=4$, $G=20$.}]
 {\resizebox{4.3cm}{!}{\includegraphics{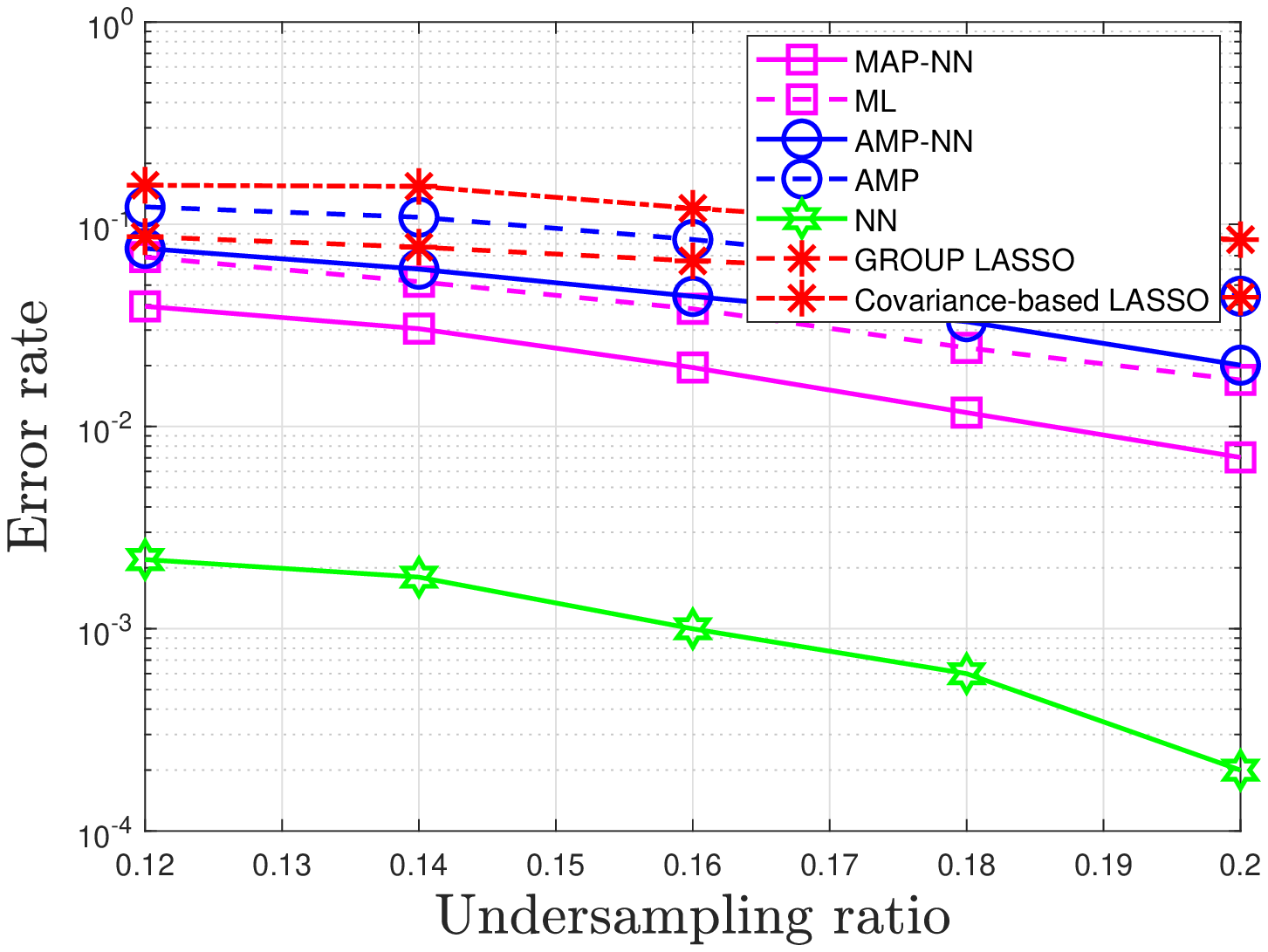}}}
 \subfigure[\scriptsize{Error rate versus $p$ at $L/N=0.2$, $M=4$, $G=20$.}]
 {\resizebox{4.3cm}{!}{\includegraphics{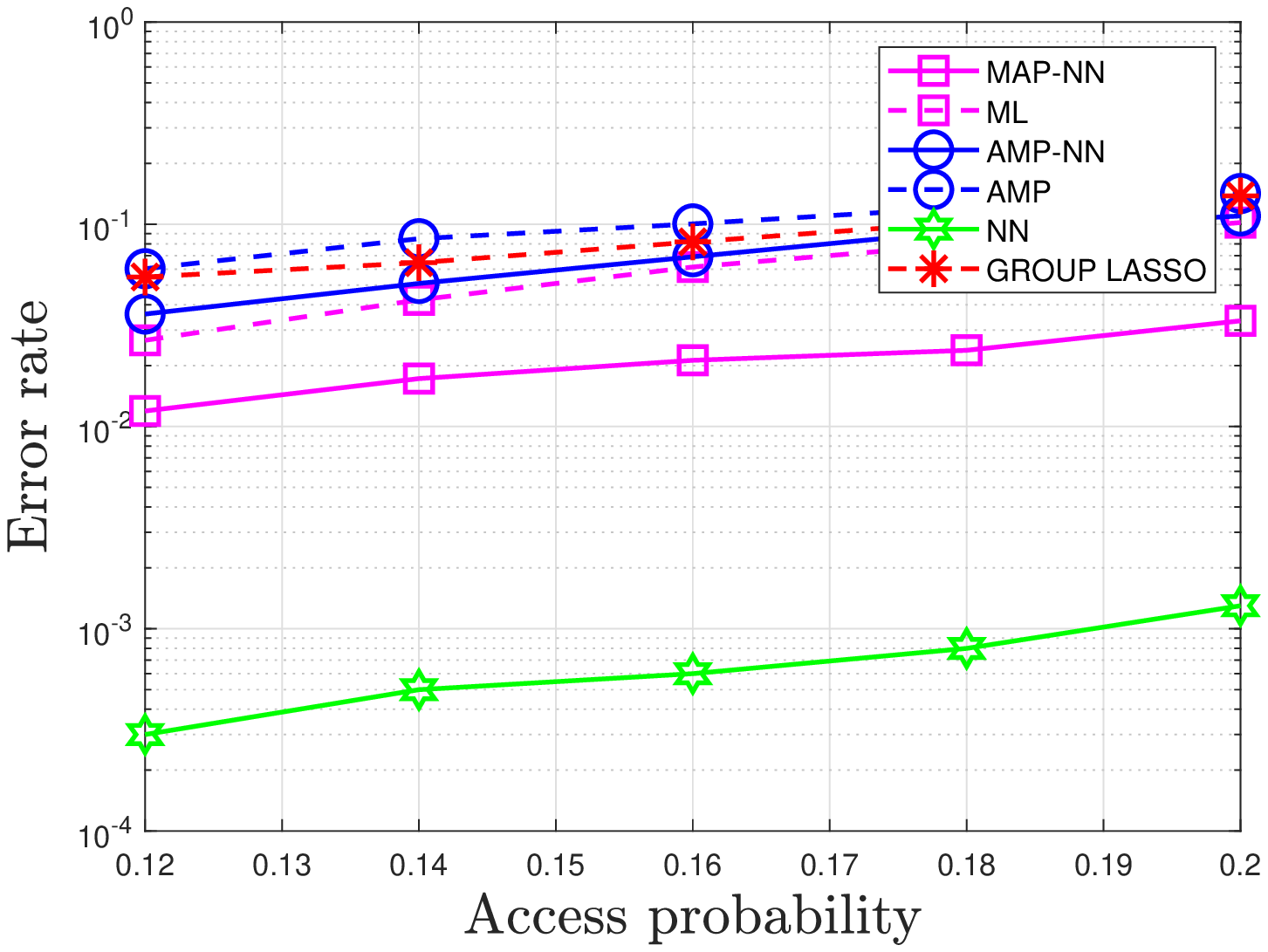}}}
 \subfigure[\scriptsize{Error rate versus $M$ at $L/N=0.12$, $p=0.1$, $G=20$.}]
 {\resizebox{4.3cm}{!}{\includegraphics{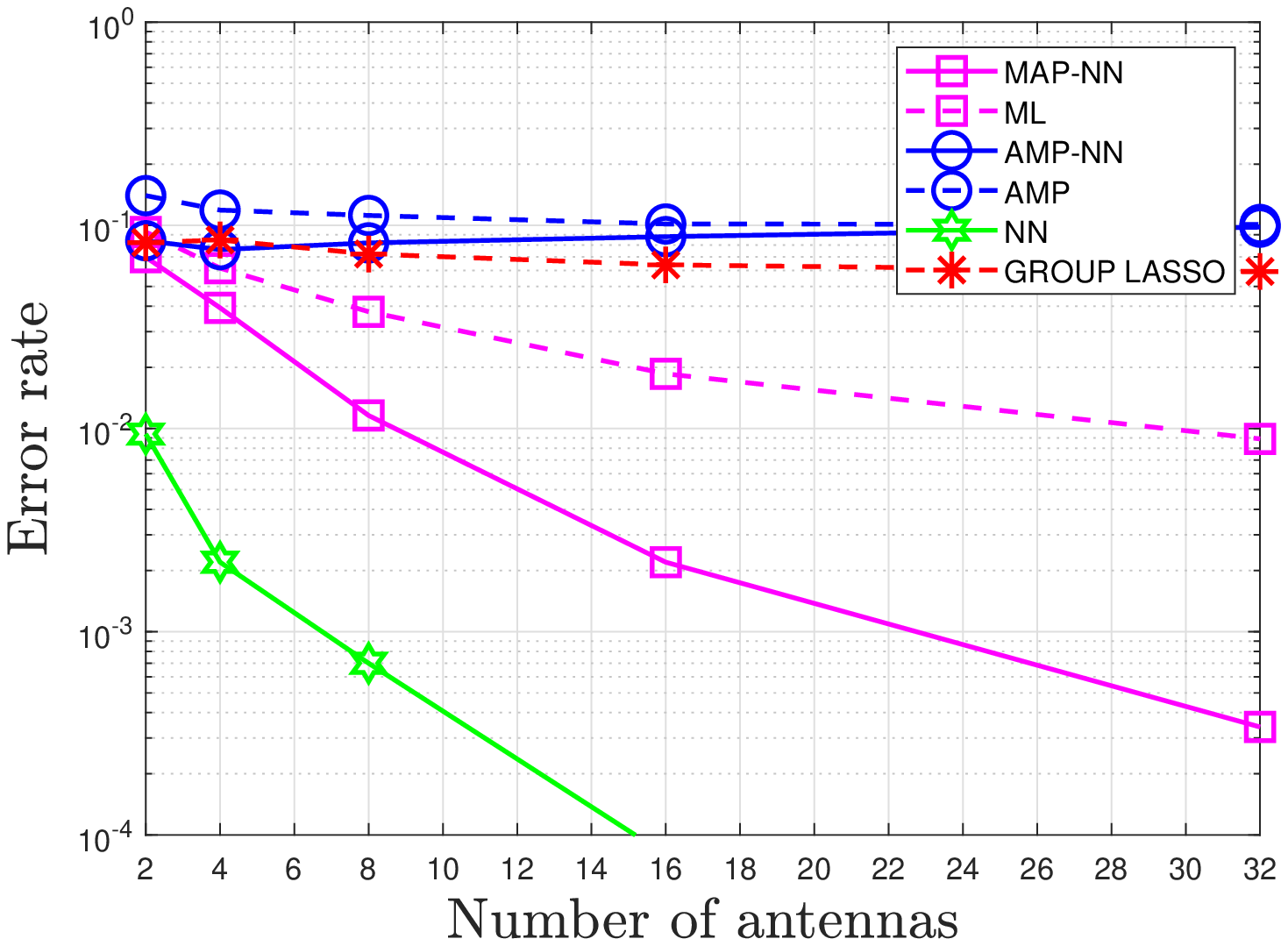}}}
 \subfigure[\scriptsize{Error rate versus $G$ at $L/N=0.12$, $M=4$, $p=0.1$.}]
 {\resizebox{4.3cm}{!}{\includegraphics{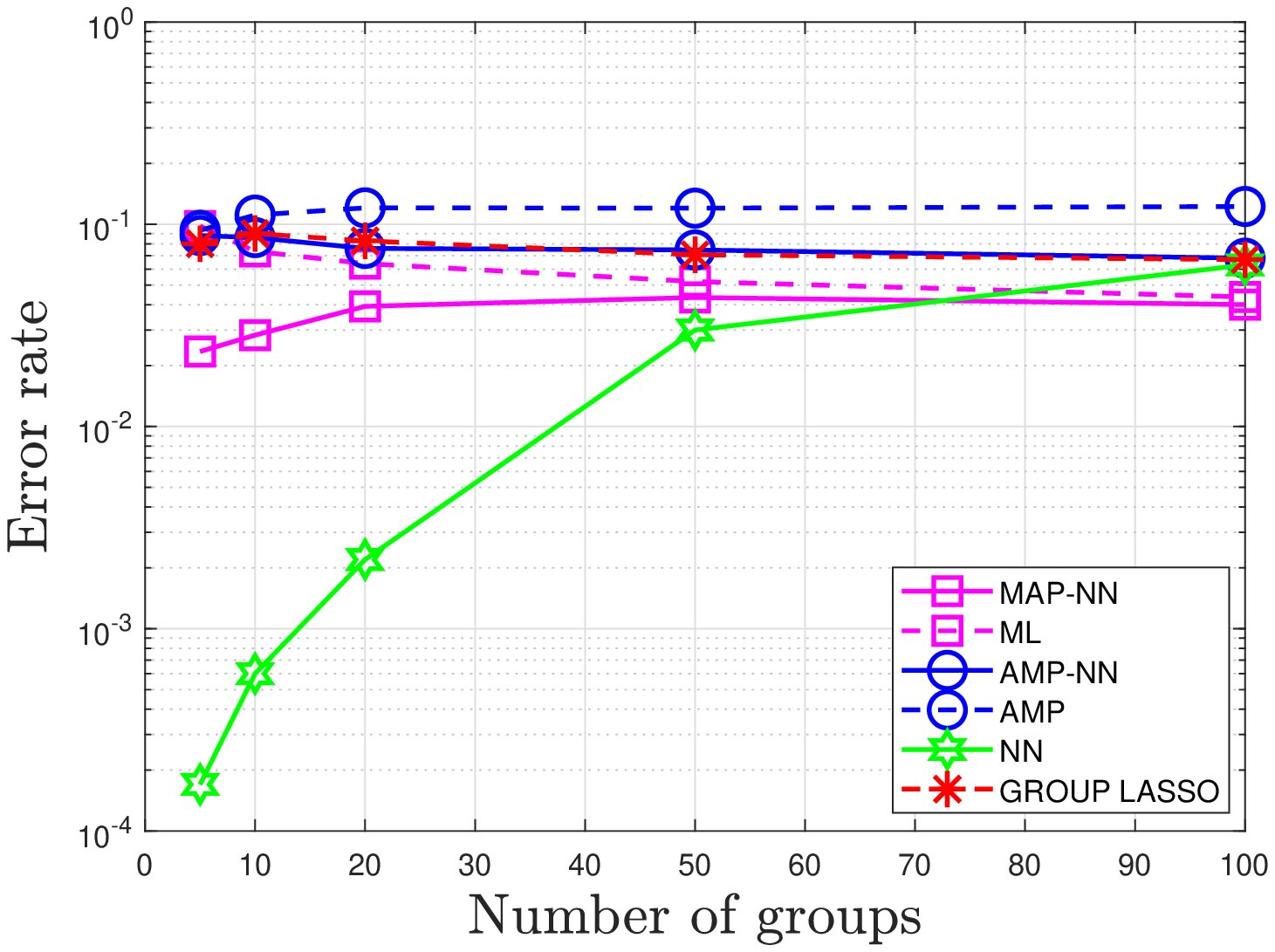}}}
  \end{center}
 \vspace*{-0.3cm}
   \caption{\small{\lsccc{Device activity detection in the} correlated case with i.i.d. group activity at $N=100$.}}
   \label{supportco100}
\end{figure}

\begin{figure}[tp]
\begin{center}
 \subfigure[\scriptsize{Error rate versus $L/N$ at $p=0.1$, $M=16$, $p_1/p_2=3$.}]
 {\resizebox{4.3cm}{!}{\includegraphics{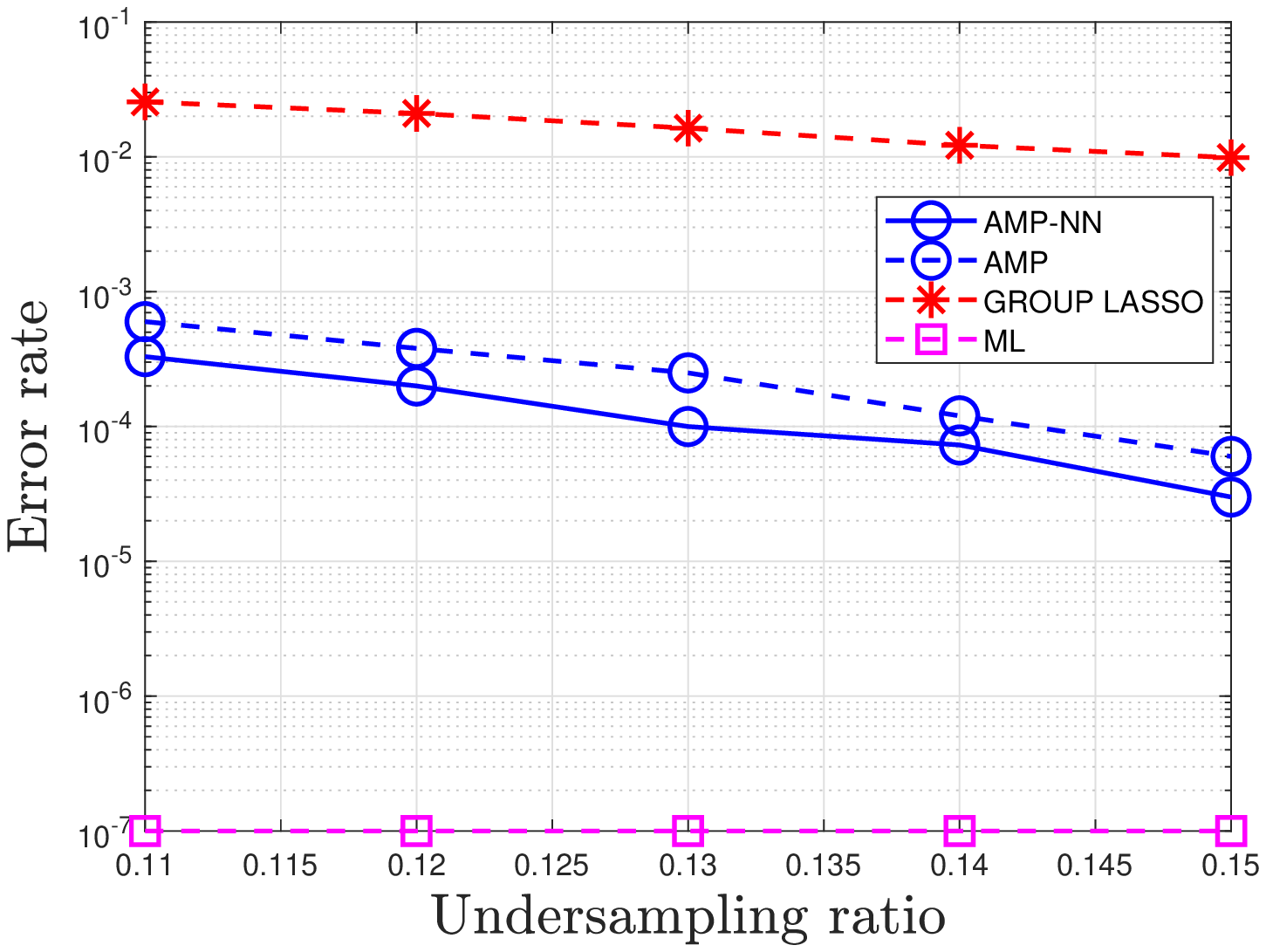}}}
 \subfigure[\scriptsize{Error rate versus $p$ at $L/N=0.15$, $M=16$, $p_1/p_2=3$.}]
 {\resizebox{4.3cm}{!}{\includegraphics{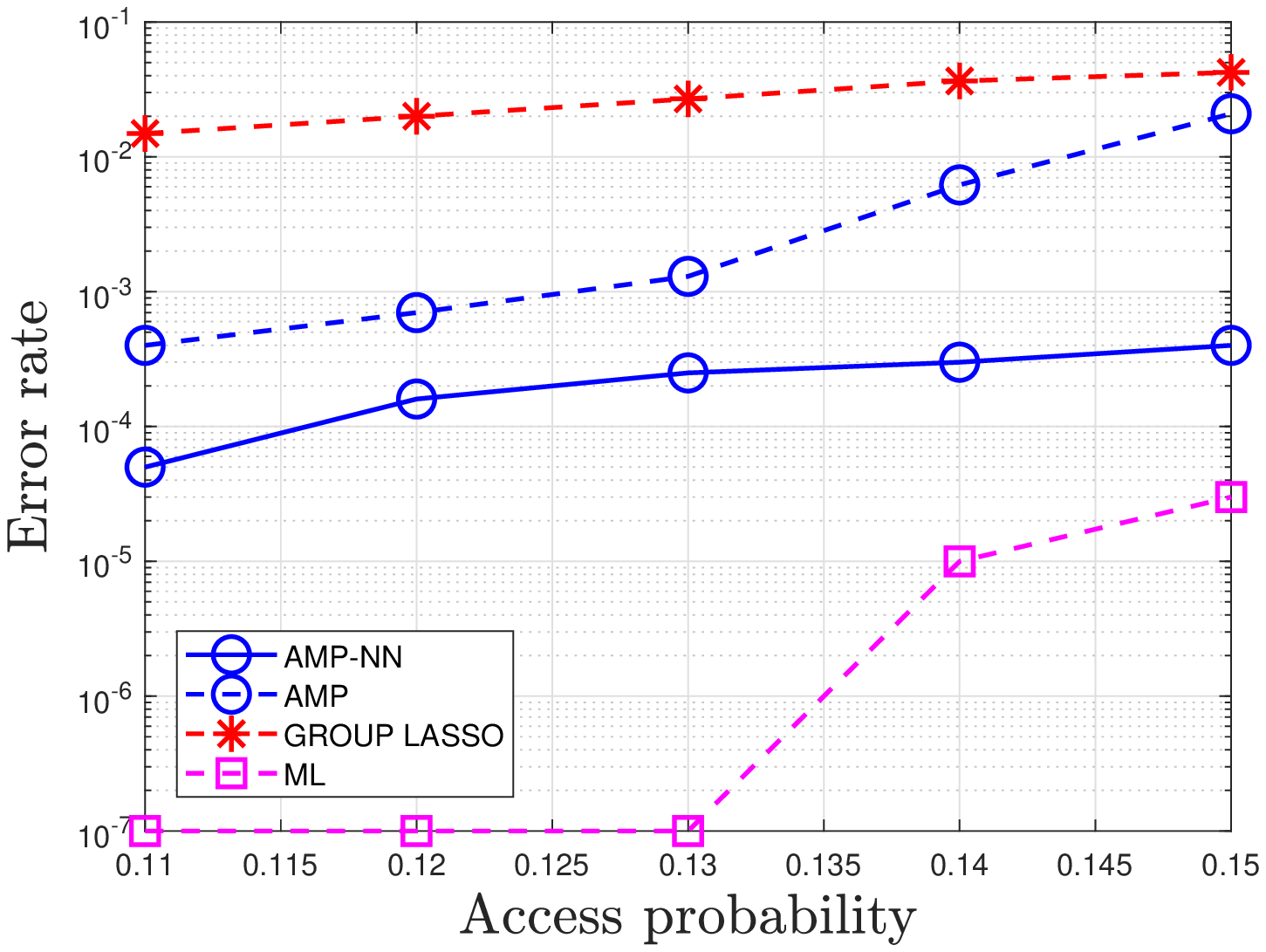}}}
 \subfigure[\scriptsize{\lscr{Error rate versus $M$ at $L/N=0.11$, $p=0.1$, $p_1/p_2=3$.}}]
 {\resizebox{4.3cm}{!}{\includegraphics{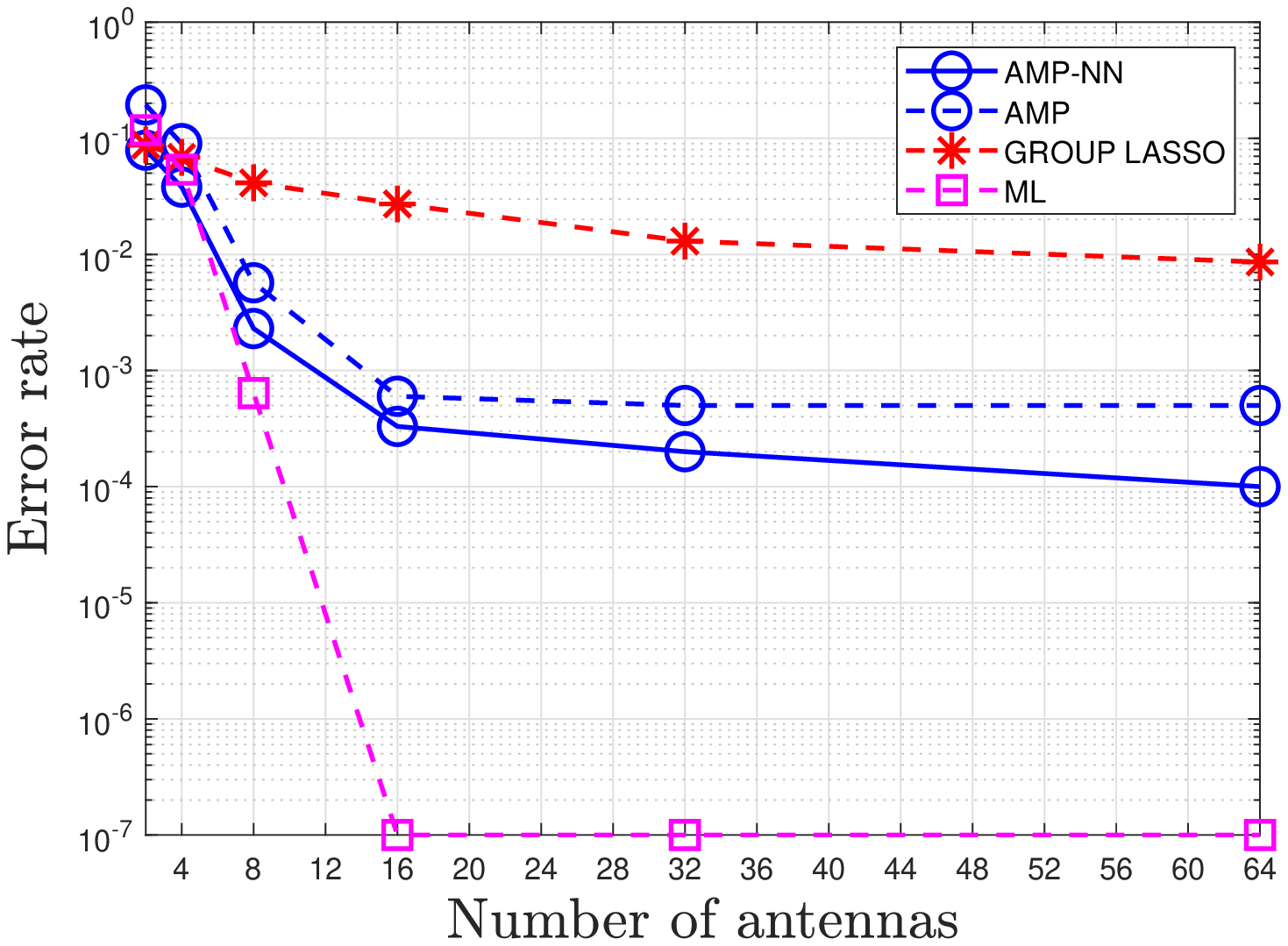}}}
 \subfigure[\scriptsize{Error rate versus $p_1/p_2$ at $L/N=0.11$, $M=16$, $p=0.1$.}]
 {\resizebox{4.3cm}{!}{\includegraphics{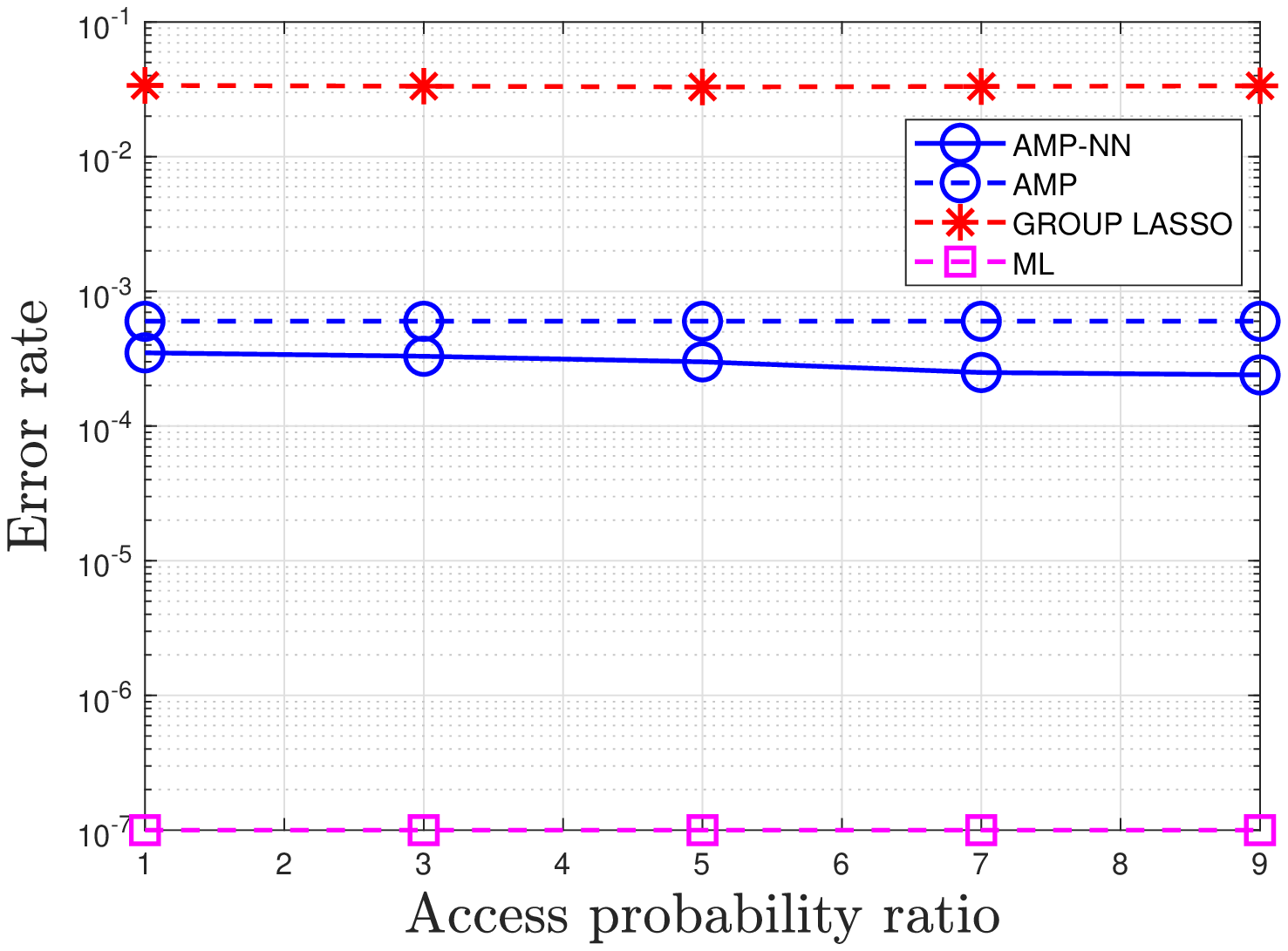}}}
  \subfigure[\scriptsize{\lsccc{Coherence versus $L/N$ \lscca{at $p=0.1$, $M=16$, $p_1/p_2=3$.}.}}]
 {\resizebox{4.4cm}{!}{\includegraphics{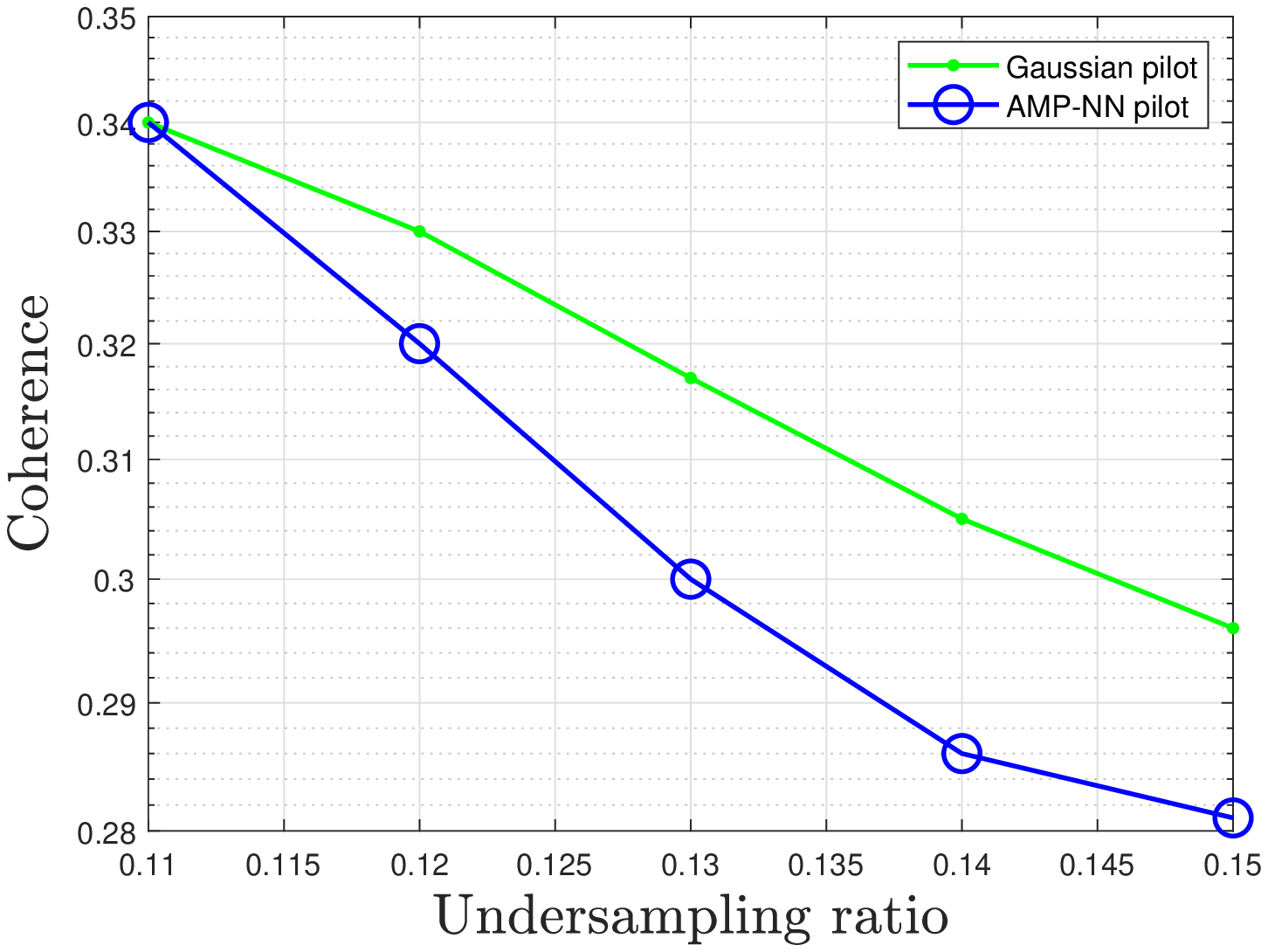}}}
  \end{center}
 \vspace*{-0.3cm}
   \caption{\small{\lsccc{Device activity detection in the} independent case at $N=1000$.}}
   \label{support1000}
\end{figure}

\begin{figure}[tp]
\begin{center}
 \subfigure[\scriptsize{Error rate versus $L/N$ at $p=0.1$, $M=16$, $G=200$.}]
 {\resizebox{4.3cm}{!}{\includegraphics{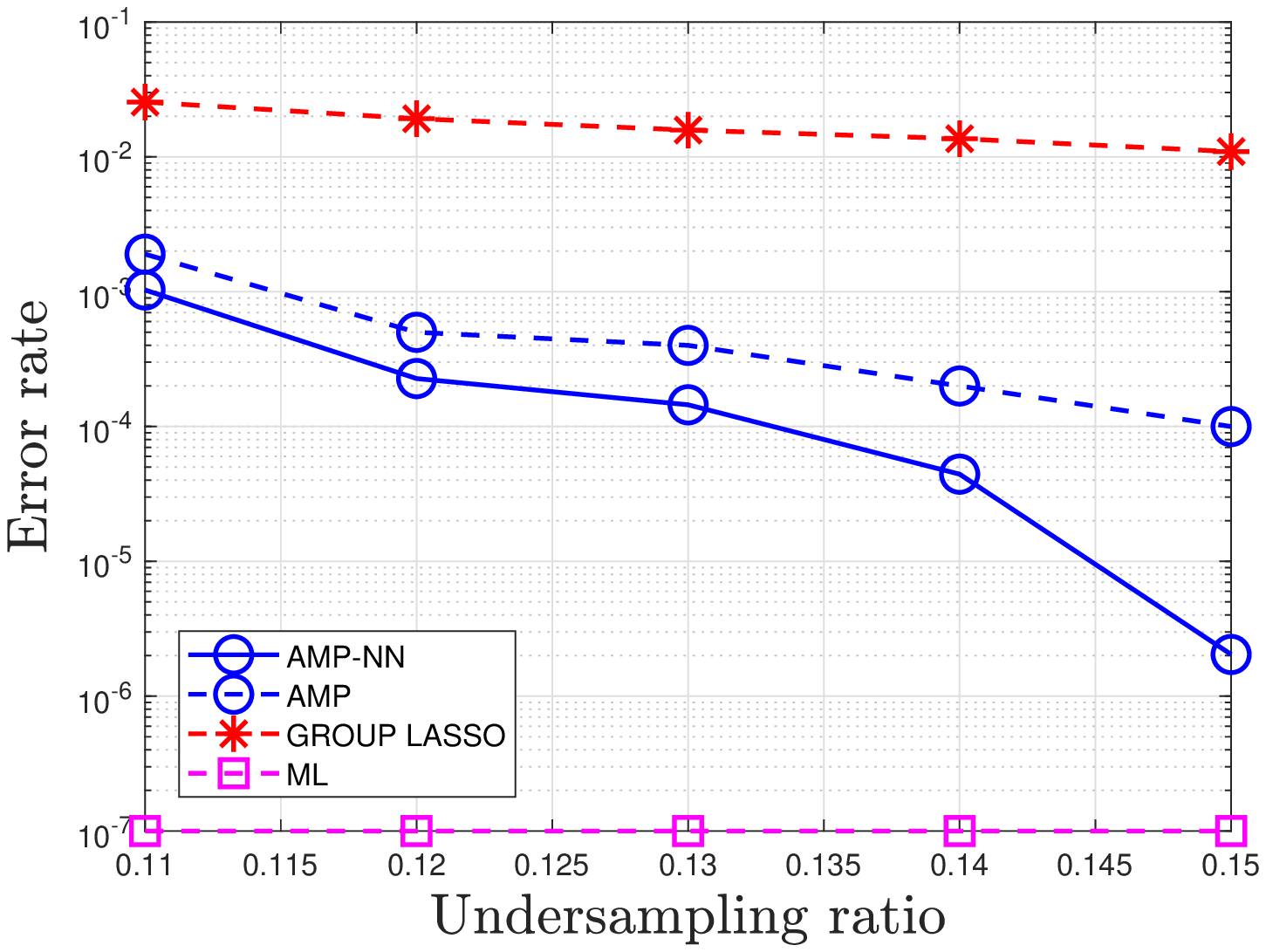}}}
 \subfigure[\scriptsize{Error rate versus $p$ at $L/N=0.15$, $M=16$, $G=200$.}]
 {\resizebox{4.3cm}{!}{\includegraphics{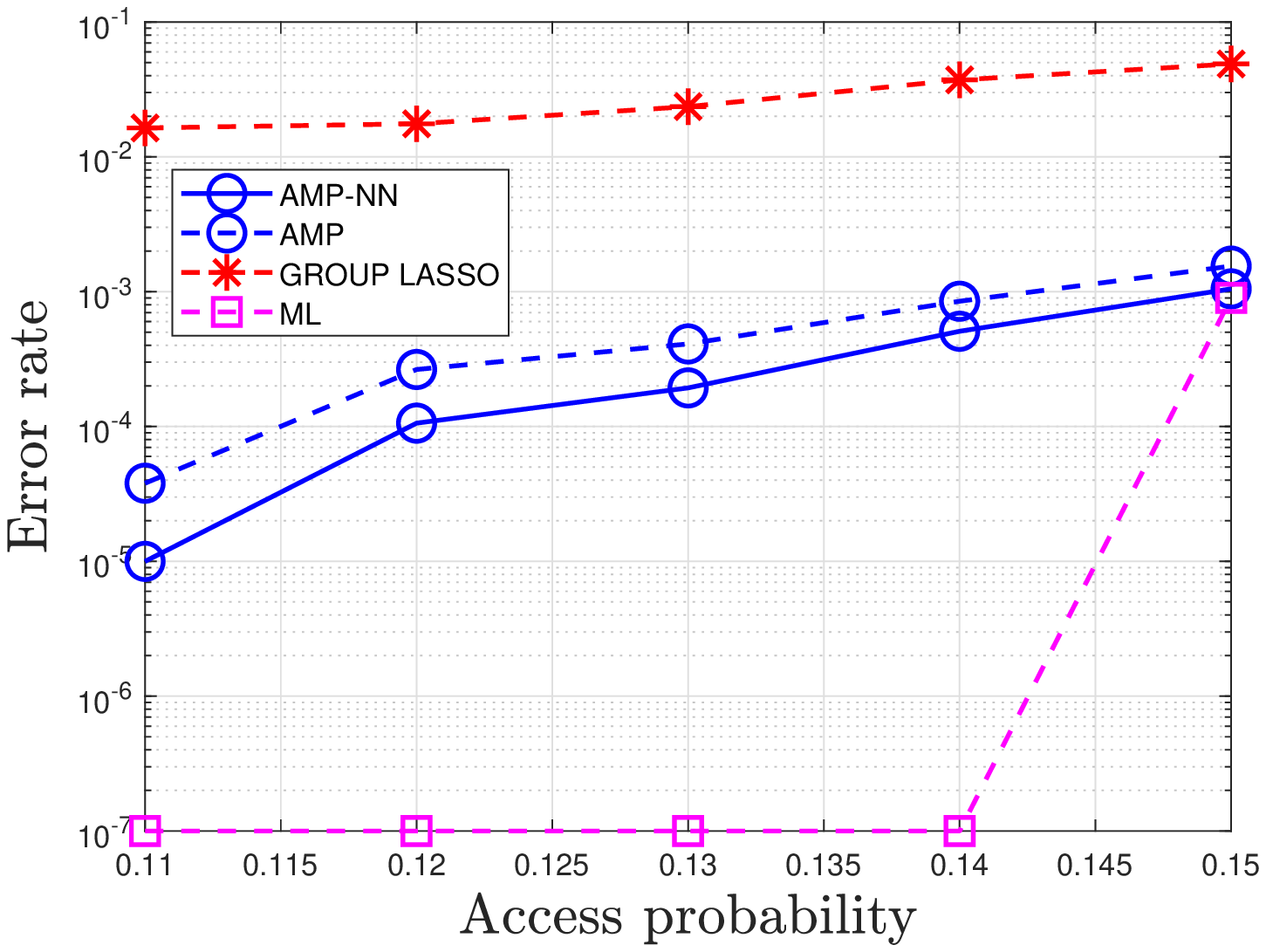}}}
 \subfigure[\scriptsize{\lscr{Error rate versus $M$ at $L/N=0.11$, $p=0.1$, $G=200$.}}]
 {\resizebox{4.3cm}{!}{\includegraphics{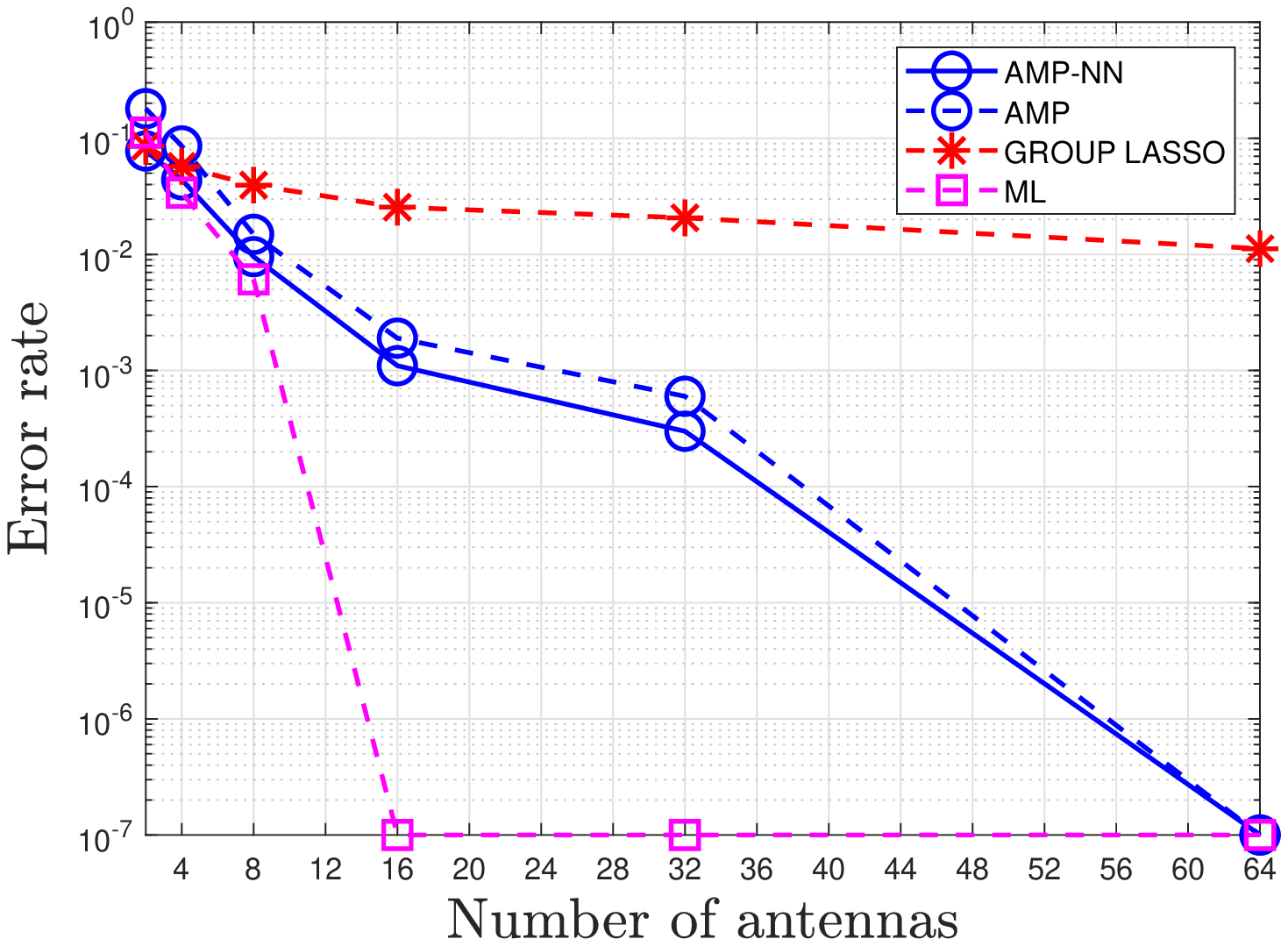}}}
 \subfigure[\scriptsize{Error rate versus $G$ at $L/N=0.11$, $M=16$, $p=0.1$.}]
 {\resizebox{4.3cm}{!}{\includegraphics{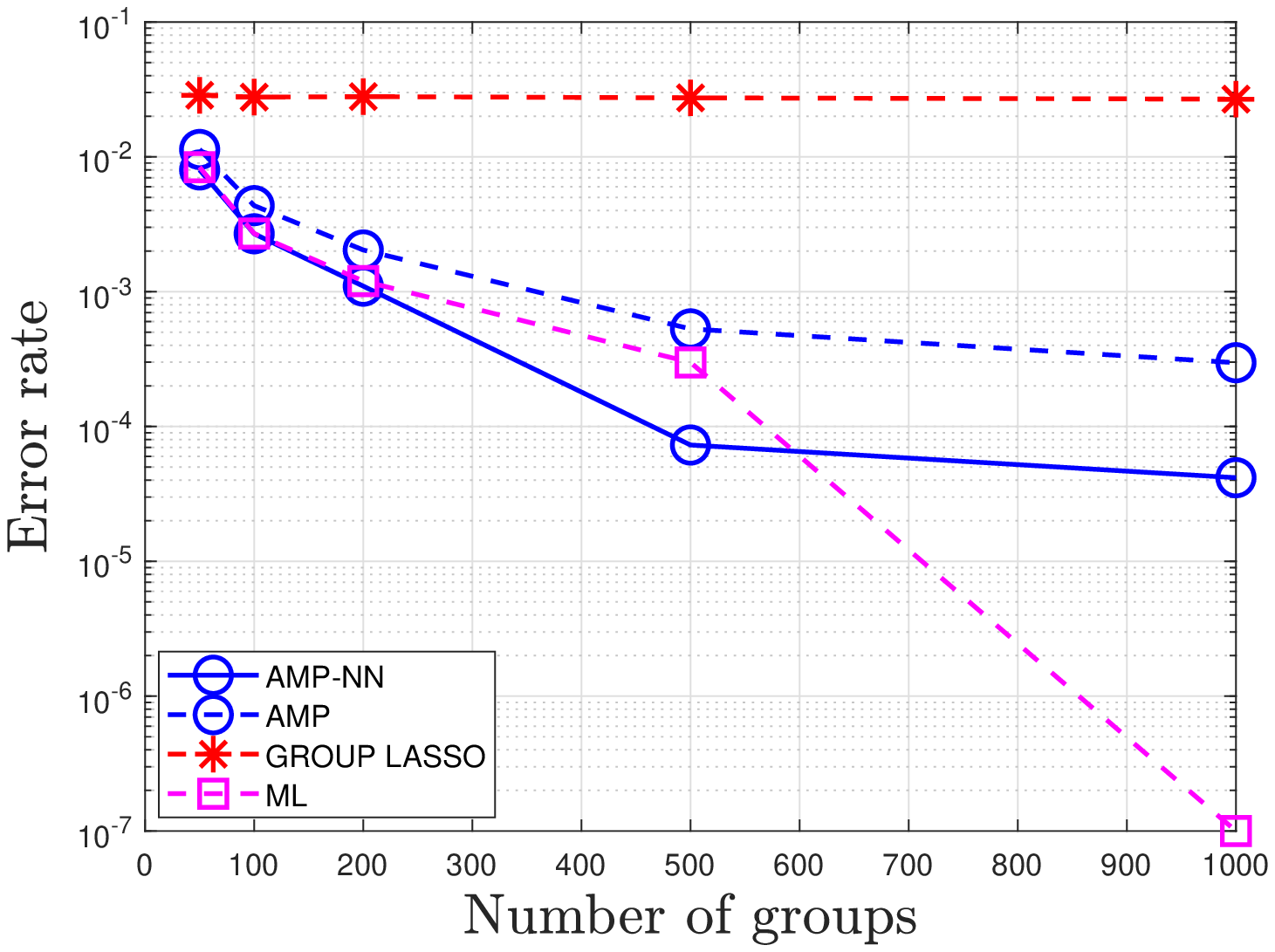}}}
  \subfigure[\scriptsize{Block-coherence across groups versus $L/N$ \lscca{at $p=0.1$, $M=16$, $G=200$}.}]
 {\resizebox{4.3cm}{!}{\includegraphics{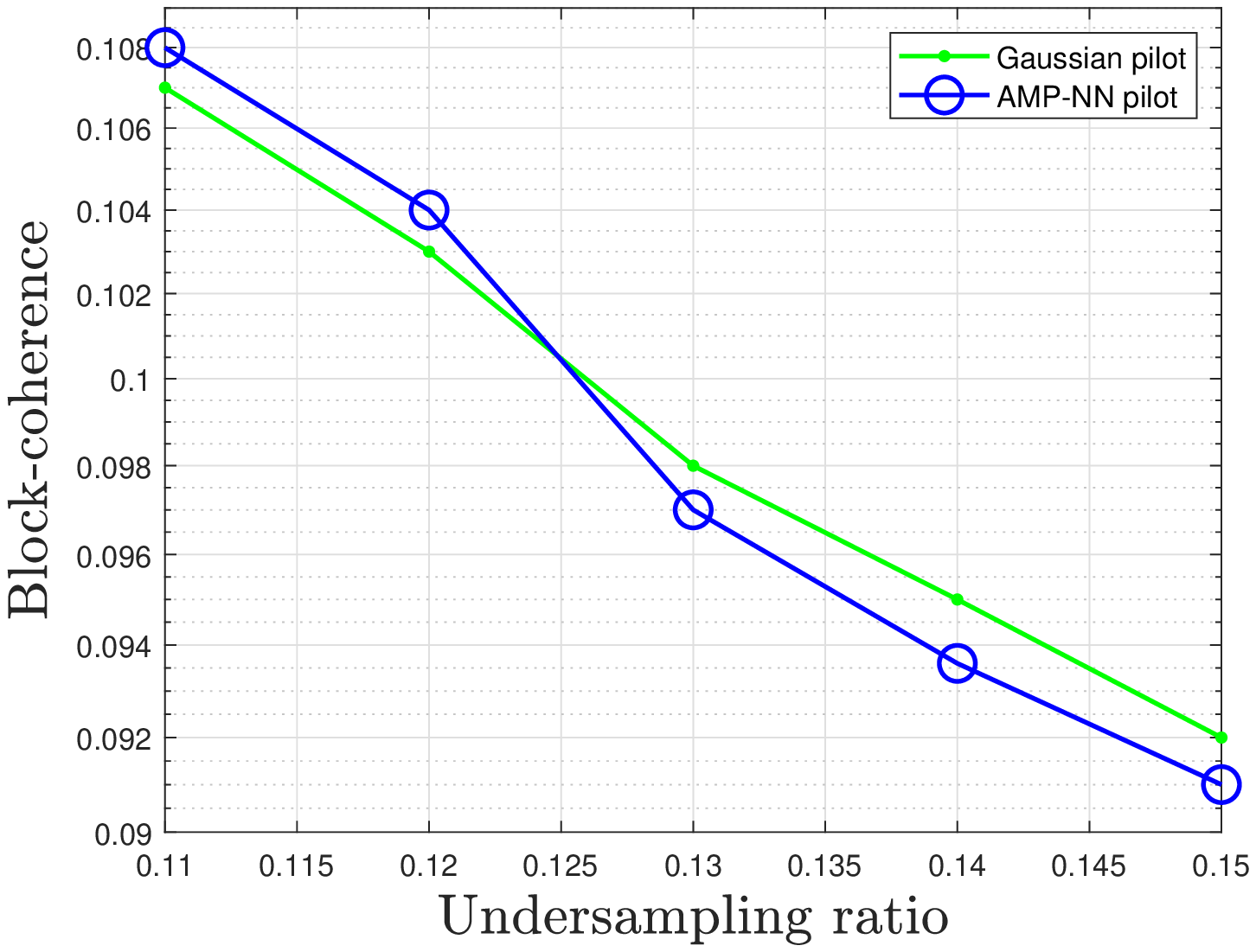}}}
 \subfigure[\scriptsize{Sub-coherence within one group versus $L/N$ \lscca{at $p=0.1$, $M=16$, $G=200$}.}]
 {\resizebox{4.3cm}{!}{\includegraphics{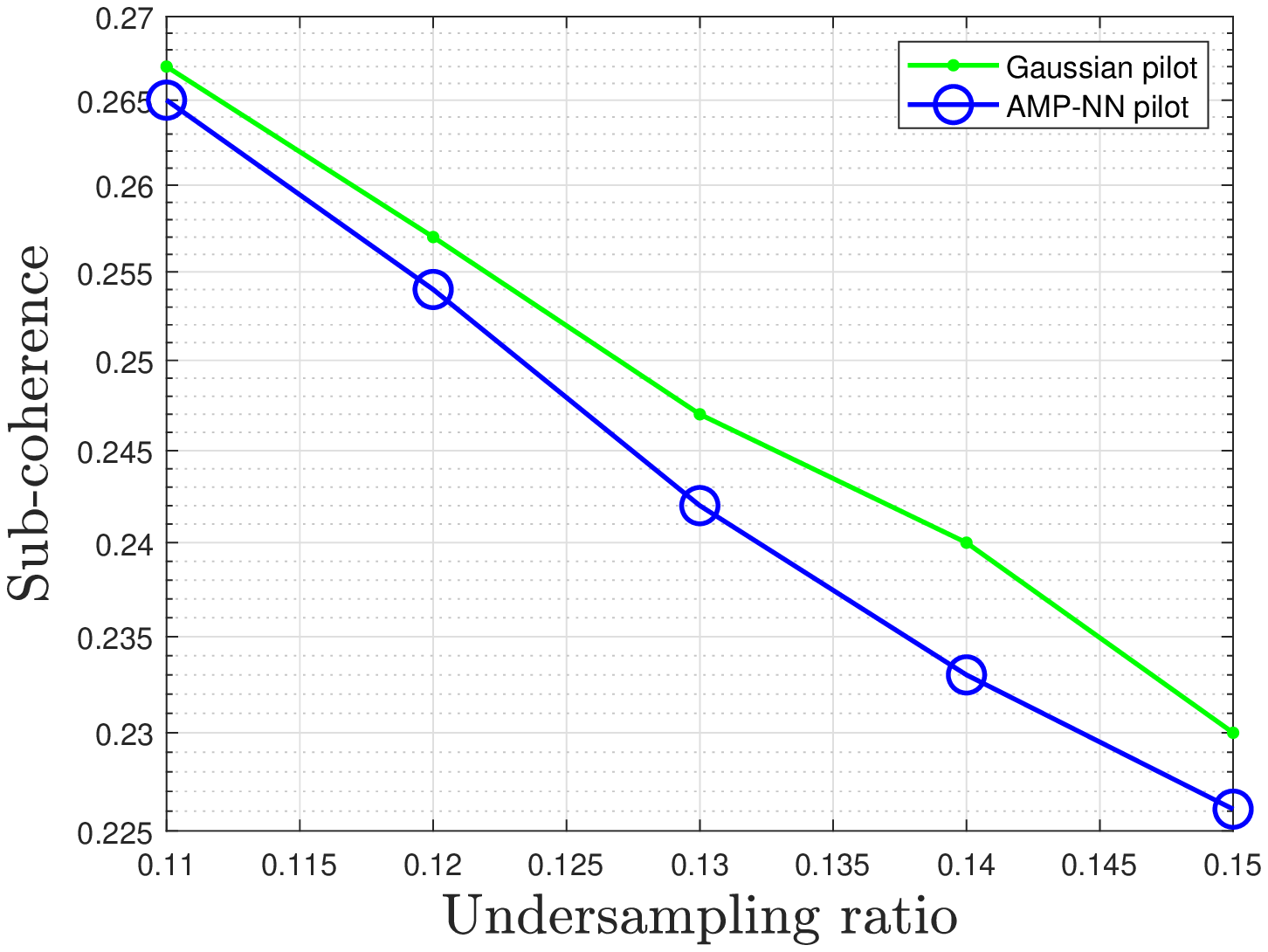}}}
  \end{center}
 \vspace*{-0.3cm}
   \caption{\small{\lsccc{Device activity detection in the} correlated case with i.i.d. group activity at $N=1000$.}}
   \label{supportco1000}
\end{figure}

Fig.~\ref{support100} (a), (b), (c), (d) and Fig.~\ref{support1000} \lscc{(a), (b), (c), (d)} illustrate the error rate versus the undersampling ratio $L/N$, access probability $p$, \lscc{\lsccc{number of antennas} $M$ and access ratio $p_1/p_2$} in the independent case at $N=100$ and $N=1000$, respectively. Fig.~\ref{supportco100} \lscc{(a), (b), (c), (d)} and Fig.~\ref{supportco1000} (a), (b), (c), (d) illustrate the error rate versus the undersampling ratio $L/N$, access probability $p$,  \lscc{\lsccc{number of antennas} $M$ and number of groups $G$} in the correlated case with i.i.d. group activity at $N=100$ and $N=1000$, respectively. From Fig.~\ref{support100} (a) and Fig.~\ref{supportco100} (a), we see that covariance-based LASSO performs much worse than Group LASSO and AMP at small $M$, as explained in Section~\ref{supco}. Given its unsatisfactory recovery accuracy, we no longer compare with covariance-based LASSO in the remaining figures. From all these figures, we make the following observations. \lscc{For each scheme, we observe similar trends with respect to $L/N$, $p$ and $M$}, \lscv{as in Section~\ref{resultsig}. Here we focus on the trends with respect to $G$. Note that correlation decreases with $G$. When $N=100$ and $N=1000$, the error rate of GROUP LASSO almost does not change with $G$, and the error rate of ML always decreases with $G$. When $N=100$, the error rates of NN and MAP-NN increase with $G$, as not much correlation can be utilized by the proposed model-driven approach for improving support recovery accuracy when $G$ is large.}  Our proposed  MAP-NN and AMP-NN outperform ML and AMP, respectively, under all considered parameters, demonstrating the competitive advantage of the proposed model-driven approach in designing effective measurement matrix and correction layers for improving support recovery accuracy. \lscv{When $N=100$, MAP-NN outperforms NN in the independent case, as Algorithm~3, which is designed under the independent assumption, is very effective in such case. When $N=100$, NN outperforms MAP-NN in the correlated case with  i.i.d. group activity, as NN can better utilize the correlation information in such case. When $N=100$,} our proposed MAP-NN and NN have the smallest error rate in the independent case and in the correlated case with i.i.d. group activity, respectively. \lscv{When $N=1000$, ML achieves the smallest error rate in the independent case, and AMP-NN achieves the smallest error in the correlated case with i.i.d. group activity. The reasons are as follows. ML achieves higher \lscb{support} recovery accuracy than AMP at the cost of computation time increase. However, AMP-NN can utilize correlation information to improve \lscb{support} recovery accuracy in the correlated case with i.i.d. group activity.}

\lscca{Fig.~\ref{support1000} (e) shows the coherence of our learned measurement matrix in the independent case and Gaussian matrix versus the
undersampling ratio $L/N$ at $N = 1000$. Note that Fig.~\ref{support1000} (e) for device activity detection in the independent case is the same as Fig.~\ref{signal1000} (e) for channel estimation in the independent case. This is because the proposed approaches for jointly sparse support recovery and signal recovery with the AMP-based decoder share the same auto-encoder structure and training process.} Fig.~\ref{supportco1000} (e), (f) show the block-coherence across groups and sub-coherence within one group for all devices of our learned measurement matrix in the correlated case with i.i.d. group activity and Gaussian matrix versus the undersampling ratio $L/N$ at $N=1000$.
As from Fig.~\ref{signalco1000} (e), from Fig.~\ref{supportco1000} (f), we can see that the learned measurement matrix in our proposed AMP-NN has smaller sub-coherence than i.i.d. Gaussian matrix, to more effectively differentiate devices that are always active simultaneously. From Fig.~\ref{supportco1000} (e), we can see that when $L/N$ is small, the learned measurement matrix of our proposed AMP-NN has larger block-coherence than Gaussian matrix, as the differentiability of the devices that have a smaller chance of being  active at the same time is sacrificed at small $L/N$; when $L/N$ is large, the learned measurement matrix of our proposed AMP-NN has smaller block-coherence than Gaussian matrix, as the differentiability of the devices that have a smaller chance of being  active at the same time can be taken into account at large $L/N$.
%\lscca{Fig.~\ref{support1000} (e) and} Fig.~\ref{supportco1000} (e), (f) demonstrate the contribution of the designed measurement matrix to the gain of AMP-NN over AMP.
%The Fig.~\ref{support100} (e) shows that the learned average activity ratios for the approaches with covariance-based decoder and with AMP-based decoder increase with access ratio, demonstrating that the proposed the approaches with these two decoders can capture prior access ratio to a certain extent.

\begin{figure}[tp]
\begin{center}
 \subfigure[\scriptsize{\lscr{Computation time versus $L/N$ at $M=4$ and $N=100$.}}]
 {\resizebox{4.3cm}{!}{\includegraphics{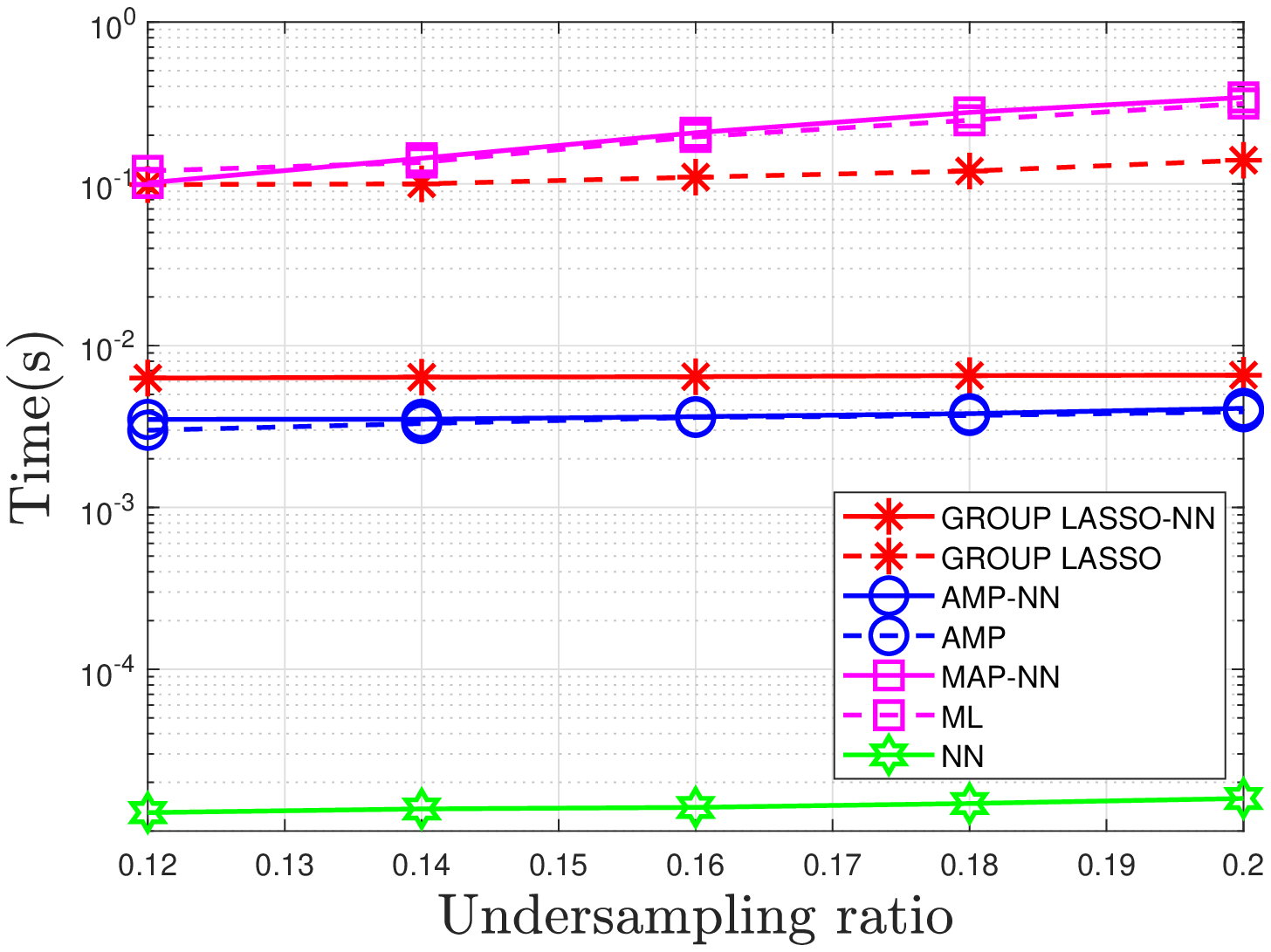}}}
 \subfigure[\scriptsize{\lscr{Computation time versus $L/N$ at $M=16$ and $N=1000$.}}]
 {\resizebox{4.3cm}{!}{\includegraphics{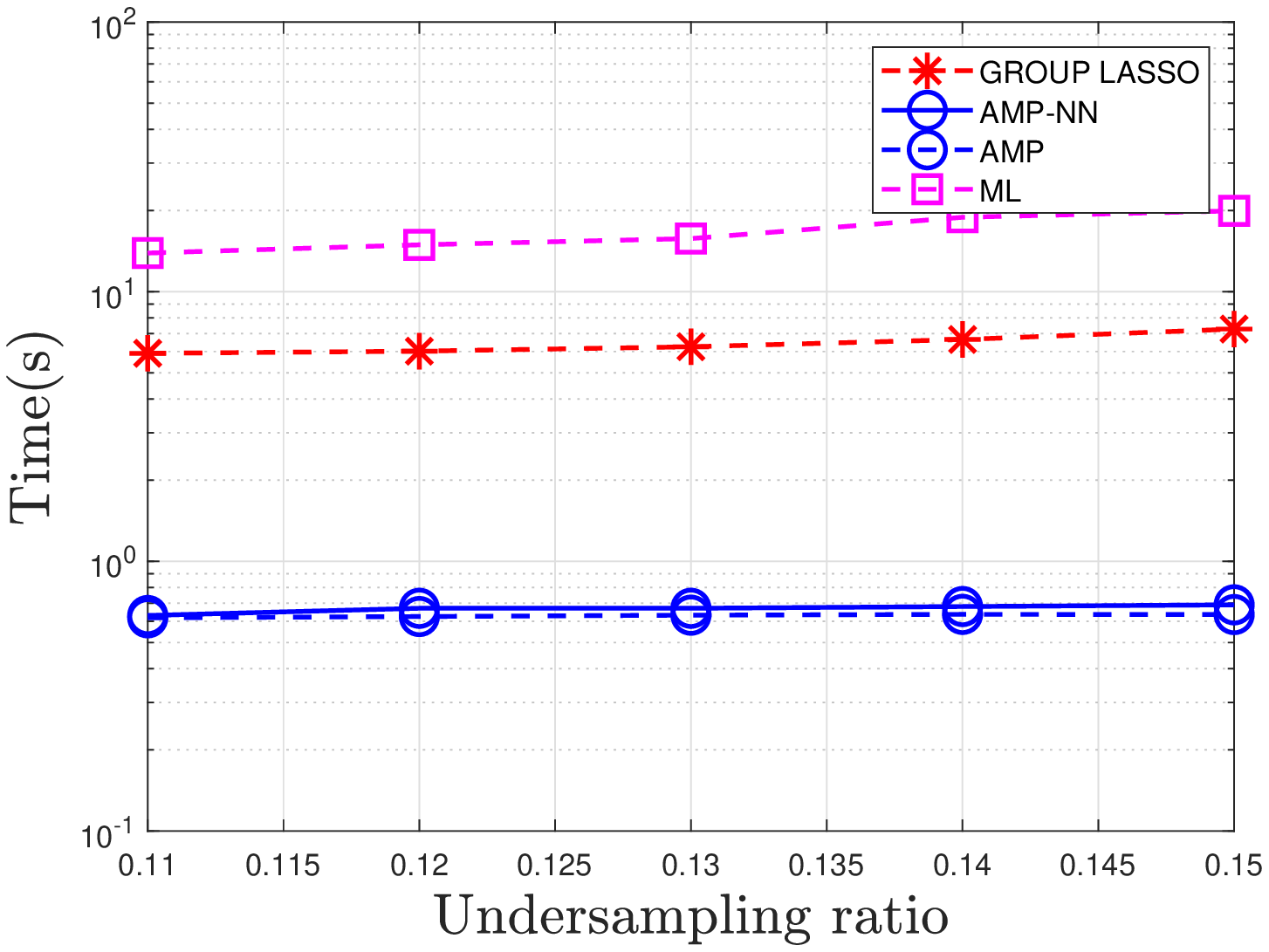}}}
  \end{center}
 \vspace*{-0.3cm}
   \caption{\small{Computation time \lsccc{for device activity detection} in the independent case \lsccc{at $p=0.1$ and $p_1/p_2=3$}.}}
   \label{timesup}
   \vspace{-0.1cm}
\end{figure}

Fig.~\ref{timesup} \lscr{illustrates} the computation \lscr{time} versus the undersampling ratio $L/N$ in independent case at $N=100$ and $N=1000$. From Fig.~\ref{timesup}, \lscr{we can make the following observations.} \lscn{The proposed AMP-NN and MAP-NN have similar computation time to AMP and ML, respectively}. \lscb{Because AMP-NN and AMP have the same number of iterations, MAP-NN and ML have the same number of iterations, and $V=3$ in AMP-NN and MAP-NN is \lsccc{quite} small.} \lscv{When $N=100$}, our proposed NN has the \lscc{shortest} computation time, \lscv{as the covariance-based decoder does not require an approximation part}. \lscn{When $N=1000$, AMP and our proposed AMP-NN have the shortest computation time.} \lscc{GROUP LASSO and ML have much longer computation time than the other schemes, especially at large $N$. Thus, they \lscv{may not be suitable for device activity detection in} practical mMTC (with large $N$).}

\section{Conclusion}
In this paper, we propose two-model driven approaches, each consisting of an encoder and a model-driven decoder, using the standard auto-encoder structure for real numbers in deep learning. One aims to jointly design the common measurement matrix and jointly sparse signal recovery methods for complex signals, and can be used in the joint design of pilot sequences and channel estimation methods in MIMO-based grant-free random access. The other is to jointly design the common measurement matrix and jointly sparse support recovery methods for complex signals, and can be applied to the joint design of pilot sequences and device activity detection in MIMO-based grant-free random access. We propose the Group LASSO-based decoder and AMP-based decoder, as \lscb{key} instances for the model-driven decoder for jointly sparse signal recovery. In addition, we propose the covariance-based decoder, MAP-based decoder and AMP-based decoder, as important instances for the model-driven decoder for jointly sparse support recovery. These decoders are all based on the state-of-the-art recovery methods. The proposed model-driven approaches can greatly benefit from the \lscr{underlying advanced} methods with theoretical performance guarantee via the approximation \lscr{parts} of the model-driven decoders. They can also effectively utilize \lscr{features} of sparsity patterns in designing the encoders for \lscr{obtaining} effective common measure matrix and adjusting the correction \lscr{parts} of the model-driven decoders. Last but not the least, they \lscb{can} provide higher recovery accuracy with \lscb{shorter} computation time than the underlying advanced recovery methods. We conduct extensive numerical results on channel estimation and device activity detection in MIMO-based grant-free random access\lscr{. The numerical results} show that the proposed approaches can achieve better detection and estimation accuracy with shorter computation time than the state-of-the-art recovery methods, including Group LASSO, ML and AMP, etc. Furthermore, the numerical results explain how such gains are achieved via the proposed approaches. The obtained results are of critical importance for achieving massive access in mMTC.
%\clearpage

%
\bibliographystyle{IEEEtran}
%\bibliography{BIBFILE}
	\begin{IEEEbiography}[{\includegraphics[height=1.25in]{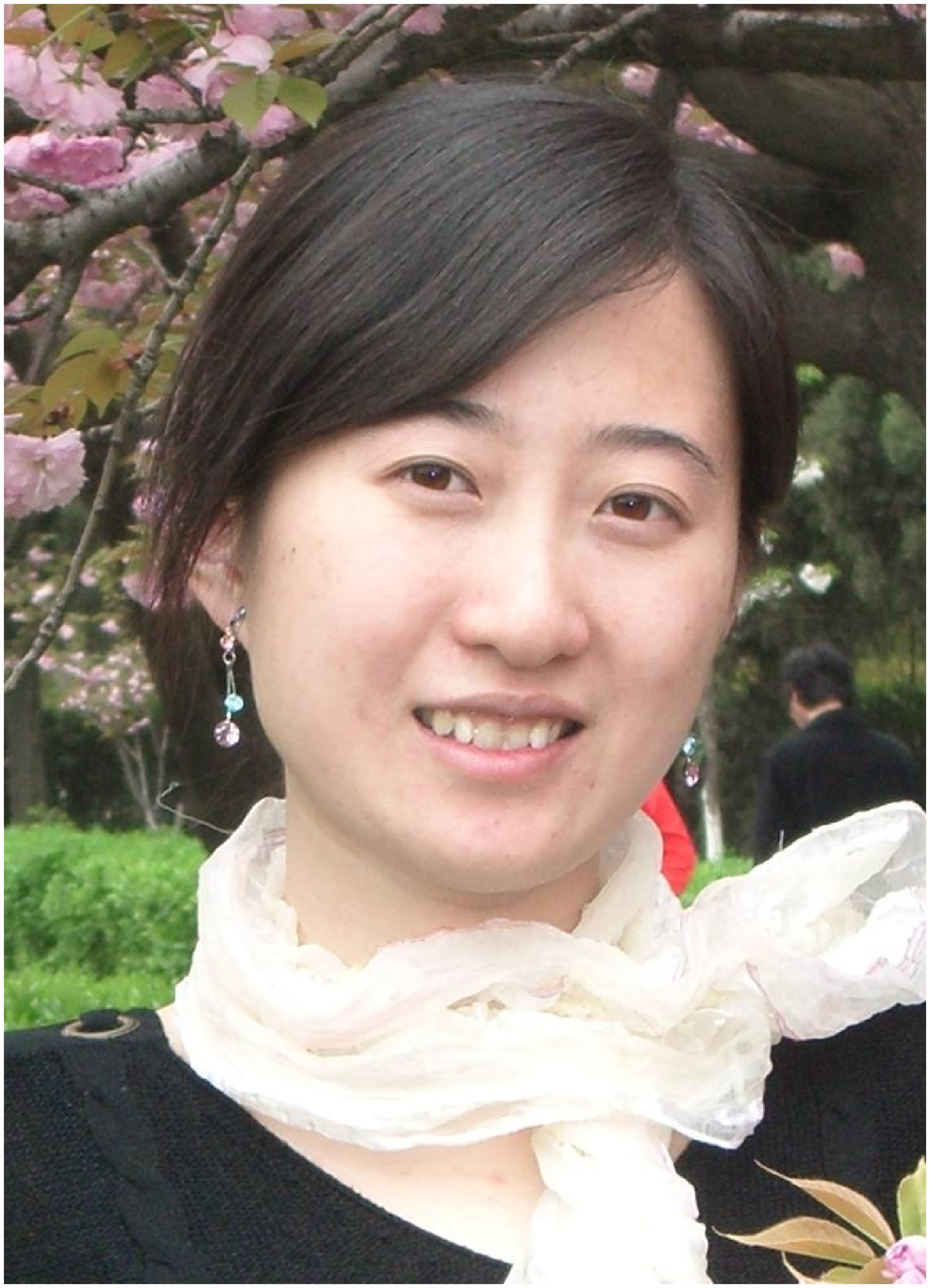}}]{Ying Cui}
		received the B.E. degree in electronic and information engineering from Xi’an Jiao Tong University, China, in 2007, and the Ph.D. degree in electronic and computer engineering from the Hong Kong University of Science and Technology (HKUST), Hong Kong, in 2011. From 2012 to 2013, she was a Post-Doctoral Research Associate with the Department of Electrical and Computer Engineering, Northeastern University, Boston, MA, USA. From 2013 to 2014, she was a Post-Doctoral Research Associate with the Department of Electrical Engineering and Computer Science, Massachusetts Institute of Technology (MIT), Cambridge, MA. Since 2015, she has been an Associate Professor with the Department of Electronic Engineering, Shanghai Jiao Tong University, China. Her current research interests include optimization, cache-enabled wireless networks, mobile edge computing, and delay-sensitive cross-layer control. She was selected to the Thousand Talents Plan for Young Professionals of China in 2013. She was a recipient of the Best Paper Award at IEEE ICC, London, U.K., June 2015. She serves as an Editor for IEEE TRANSACTIONS ON WIRELESS COMMUNICATIONS.
	\end{IEEEbiography}

	\begin{IEEEbiography}[{\includegraphics[height=1.25in]{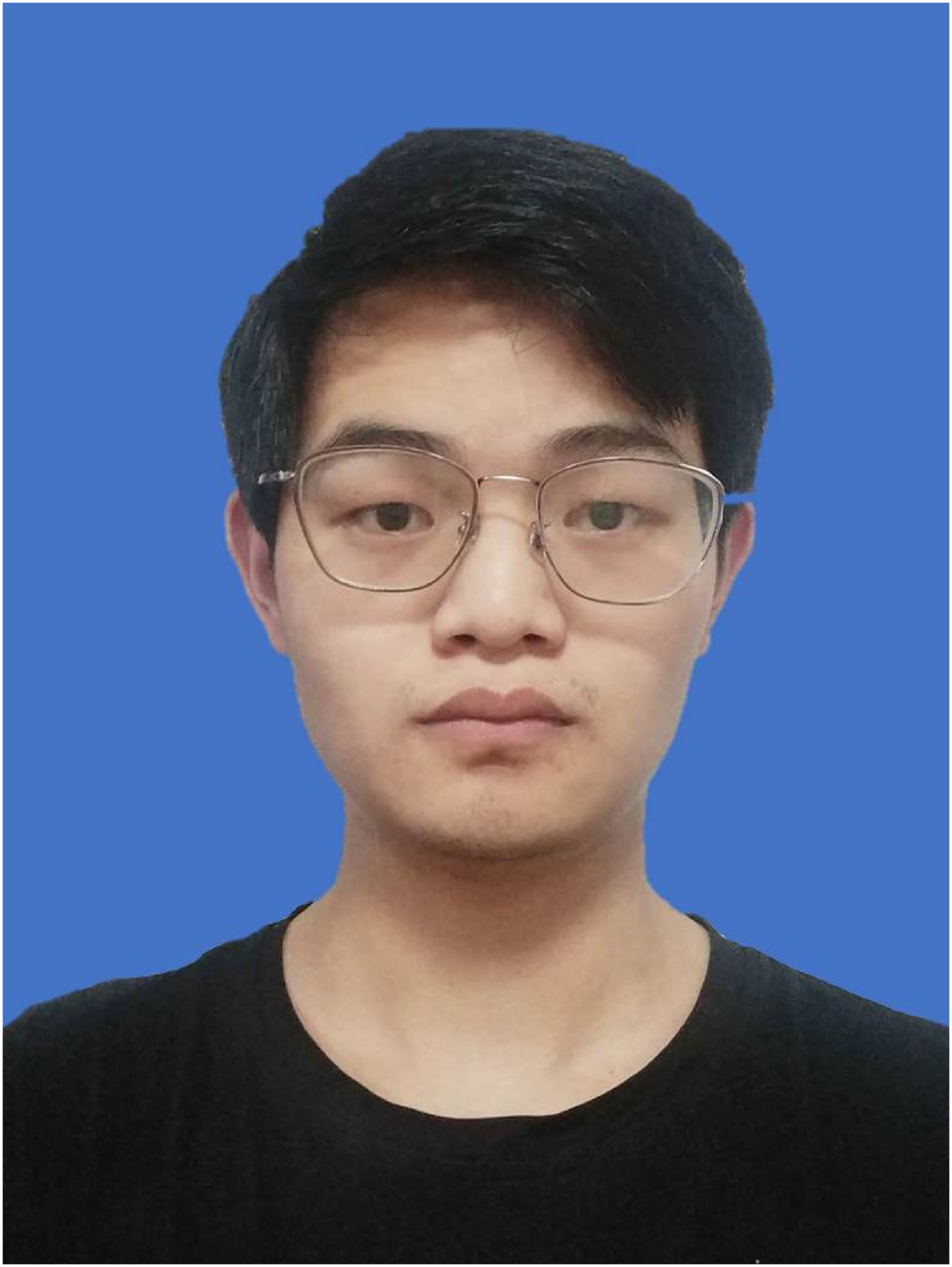}}]{Shuaichao Li}
		received the B.S. degree in University of Electronic Science and Technology of China, China, in 2018. He is currently pursuing the master’s degree with the Department of Electronic Engineering, Shanghai Jiao Tong University, China. His research interests include grant-free random access, deep learning and convex optimization.
	\end{IEEEbiography}

	\begin{IEEEbiography}[{\includegraphics[height=1.25in]{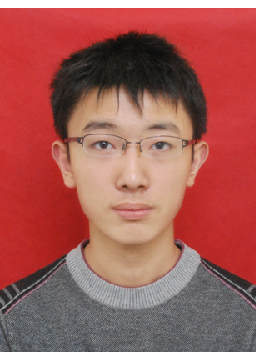}}]{Wanqing Zhang}
		received the B.S. degree in Shandong University, China, in 2018. He is currently pursuing the Ph.D. degree with the Department of Electronic Engineering, Shanghai Jiao Tong University, China. His research interests include grant-free random access, deep learning and convex optimization.
	\end{IEEEbiography}
\end{document}